\PassOptionsToPackage{unicode}{hyperref}
\PassOptionsToPackage{hyphens}{url}
\PassOptionsToPackage{dvipsnames,svgnames,x11names}{xcolor}
\documentclass[12pt]{article}

\usepackage{amsmath,amsfonts,bm}

\def\1{\bm{1}}

\DeclareMathAlphabet{\mathsfit}{\encodingdefault}{\sfdefault}{m}{sl}
\SetMathAlphabet{\mathsfit}{bold}{\encodingdefault}{\sfdefault}{bx}{n}

\newcommand{\E}{\mathbb{E}}

\usepackage{algorithm,algorithmicx,algpseudocode}
\usepackage{bbm}
\usepackage{amsmath,amssymb}
\usepackage{makecell}
\usepackage{bibunits}
\defaultbibliographystyle{asa}
\defaultbibliography{references}
\makeatletter
\newcommand{\setbibunitprefix}[1]{%
  \def\@extra@binfo{#1}%
  \def\@extra@b@citeb{#1}}
\makeatother

\usepackage{diagbox}
\usepackage{caption}
\usepackage{amsthm}
\newtheorem{theorem}{Theorem}[section]
\newtheorem{lemma}{Lemma}[section]

\newtheorem{assumption}{Assumption}
\newtheorem{proposition}{Proposition}[section]
\newtheorem{remark}{Remark}[section]
\newtheorem{corollary}{Corollary}[section]

\theoremstyle{definition}
\newtheorem{example}{Example}
\theoremstyle{plain}

\usepackage{soul}

\usepackage{iftex}
\ifPDFTeX
  \usepackage[T1]{fontenc}
  \usepackage[utf8]{inputenc}
  \usepackage{textcomp} 
\else 
  \usepackage{unicode-math}
  \defaultfontfeatures{Scale=MatchLowercase}
  \defaultfontfeatures[\rmfamily]{Ligatures=TeX,Scale=1}
\fi
\usepackage{lmodern}
\ifPDFTeX\else  
\fi
\IfFileExists{upquote.sty}{\usepackage{upquote}}{}
\IfFileExists{microtype.sty}{
  \usepackage[]{microtype}
  \UseMicrotypeSet[protrusion]{basicmath} 
}{}
\makeatletter
\@ifundefined{KOMAClassName}{
  \IfFileExists{parskip.sty}{%
    \usepackage{parskip}
  }{
    \setlength{\parindent}{0pt}
    \setlength{\parskip}{6pt plus 2pt minus 1pt}}
}{
  \KOMAoptions{parskip=half}}
\makeatother
\usepackage{xcolor}
\makeatletter
\ifx\textbf\undefined\else
  \let\oldparagraph\textbf
  \renewcommand{\textbf}{
    \@ifstar
      \xxxParagraphStar
      \xxxParagraphNoStar
  }
  \newcommand{\xxxParagraphStar}[1]{\oldparagraph*{#1}\mbox{}}
  \newcommand{\xxxParagraphNoStar}[1]{\oldparagraph{#1}\mbox{}}
\fi
\ifx\subparagraph\undefined\else
  \let\oldsubparagraph\subparagraph
  \renewcommand{\subparagraph}{
    \@ifstar
      \xxxSubParagraphStar
      \xxxSubParagraphNoStar
  }
  \newcommand{\xxxSubParagraphStar}[1]{\oldsubparagraph*{#1}\mbox{}}
  \newcommand{\xxxSubParagraphNoStar}[1]{\oldsubparagraph{#1}\mbox{}}
\fi
\makeatother

\usepackage{longtable,booktabs,array}
\usepackage{calc} 
\usepackage{etoolbox}
\makeatletter
\patchcmd\longtable{\par}{\if@noskipsec\mbox{}\fi\par}{}{}
\makeatother
\IfFileExists{footnotehyper.sty}{\usepackage{footnotehyper}}{\usepackage{footnote}}
\makesavenoteenv{longtable}
\usepackage{graphicx}
\usepackage{pdflscape}
\makeatletter
\def\maxwidth{\ifdim\Gin@nat@width>\linewidth\linewidth\else\Gin@nat@width\fi}
\def\maxheight{\ifdim\Gin@nat@height>\textheight\textheight\else\Gin@nat@height\fi}
\makeatother
\setkeys{Gin}{width=\maxwidth,height=\maxheight,keepaspectratio}
\makeatletter
\def\fps@figure{htbp}
\makeatother

\makeatletter
\@ifpackageloaded{caption}{}{\usepackage{caption}}
\AtBeginDocument{%
\ifdefined\contentsname
  \renewcommand*\contentsname{Table of contents}
\else
  \newcommand\contentsname{Table of contents}
\fi
\ifdefined\listfigurename
  \renewcommand*\listfigurename{List of Figures}
\else
  \newcommand\listfigurename{List of Figures}
\fi
\ifdefined\listtablename
  \renewcommand*\listtablename{List of Tables}
\else
  \newcommand\listtablename{List of Tables}
\fi
\ifdefined\figurename
  \renewcommand*\figurename{Figure}
\else
  \newcommand\figurename{Figure}
\fi
\ifdefined\tablename
  \renewcommand*\tablename{Table}
\else
  \newcommand\tablename{Table}
\fi
}
\@ifpackageloaded{float}{}{\usepackage{float}}
\floatstyle{ruled}
\@ifundefined{c@chapter}{\newfloat{codelisting}{h}{lop}}{\newfloat{codelisting}{h}{lop}[chapter]}
\floatname{codelisting}{Listing}

\makeatother
\makeatletter
\@ifpackageloaded{caption}{}{\usepackage{caption}}
\@ifpackageloaded{subcaption}{}{\usepackage{subcaption}}
\makeatother

\ifLuaTeX
  \usepackage{selnolig}  
\fi
\usepackage[]{natbib}
\bibpunct{(}{)}{;}{a}{}{,}
\usepackage{bookmark}
\usepackage{setspace}
\usepackage{booktabs}
\usepackage{multirow}
\usepackage[normalem]{ulem}
\usepackage{stmaryrd}

\IfFileExists{xurl.sty}{\usepackage{xurl}}{} 
\hypersetup{
  pdftitle={},
  pdfauthor={},
  pdfkeywords={},
  colorlinks=true,
  linkcolor={blue},
  filecolor={Maroon},
  citecolor={Blue},
  urlcolor={Blue},
  pdfcreator={LaTeX via pandoc}}

\newcommand{\anon}{1}

\definecolor{HJblue}{RGB}{0,85,170}

\begin{document}

\def\spacingset#1{\renewcommand{\baselinestretch}%
{#1}\small\normalsize} \spacingset{1}

\makeatletter
\newcommand{\printfnsymbol}[1]{%
  \textsuperscript{\@fnsymbol{#1}}%
}
\makeatother

\if1\anon
{
  \title{\LARGE\bf Interval Localization for Multiple Change-points with Error Rate Control}
    \author{Zijian Wei$^1$\thanks{All authors are listed in order of seniority.}, Yajie Bao$^2$, Haojie Ren$^1$, and Nan Chen$^3$\\
$^1$ {\normalsize School of Mathematical Sciences, Shanghai Jiao Tong University}\\
$^2$ {\normalsize School of Management, Fudan University}\\
$^3$ {\normalsize Department of Industrial Systems Engineering and Management,}\\
{\normalsize  National University of Singapore}
}
  \maketitle
} \fi

\if0\anon
{
  \bigskip
  \bigskip
  \bigskip
  \begin{center}
    {\LARGE\bf Interval Localization for Multiple Change-points with Error Rate Control}
\end{center}
  \medskip
} \fi

\bigskip

\begin{abstract}
    Interval localization in multiple change-point detection aims to identify intervals containing true change-points while controlling false discoveries. We formulate this task as a multiple testing problem over detector-reported intervals and develop an algorithm-agnostic and distribution-free framework that controls the interval-level FDR, the expected proportion of retained intervals containing no true change-points. The proposed permutation-based order-preserving interval selection (POIS) uses order-preserving sample splitting to separate detection from permutation testing. To improve power, POIS+ reuses the detection-side evidence and introduces detection-conditional $e$-values as unnormalized weights for permutation $p$-values. Both procedures achieve finite-sample FDR control. Under the stated conditions, POIS attains asymptotically full power while POIS+ has asymptotically no lower power than POIS. Simulation results show that POIS maintains high power and POIS+ further improves power with empirical FDR below the nominal level for both procedures. Two real data applications demonstrate the practical effectiveness of the proposed methods.
\end{abstract}

\noindent%
\vfill

\newpage
\spacingset{1.7} 


\section{Introduction}

Multiple change-point analysis has attracted sustained attention in statistics and related fields \citep{truong2020selective}. Given an ordered data sequence, the classical objective is to estimate the number and locations of the underlying distributional changes. A broad range of methods to detect and localize the unknown change-points has been developed \citep{fryzlewicz2014wild,baranowski2019narrowest,kovacs2023seeded}.

Beyond consistent estimation, recent work has increasingly studied statistical inference problems after change-point detection. Most of this literature aims to test whether estimated {\it locations} are true change-points. For example, \citet{hyun2021post} and \citet{jewell2022testing} test whether the estimated change-points are true under the Gaussian mean-change models. Recently, \citet{jia2024tune} develop an algorithm-agnostic framework with asymptotic error control for detected change-points under specific distributional assumptions, while \citet{cui2026art} provide distribution-free and finite-sample error control for reported intervals.

However, in many applications, one often deploys change-point methods to support {\it regional} scientific claims. For example, in array-based copy-number analysis \citep{olshen2004circular}, the data are observed along ordered genomic loci, and the practically meaningful statement is often that a local genomic region contains a copy-number boundary, rather than that one particular locus is exactly the change-point. Local-window procedures \citep{niu2012screening,eichinger2018mosum} assess a candidate location by comparing observations in windows to its left and right. The union of these windows defines a candidate interval to be tested for a change. Multiscale and recursive procedures, such as wild binary segmentation \citep{fryzlewicz2014wild}, narrowest-over-threshold \citep{baranowski2019narrowest}, and seeded binary segmentation \citep{kovacs2023seeded}, search over prespecified collections of candidate intervals. Each detected location can then be accompanied by the interval on which the evidence for that candidate was assessed. However, a sensitive search may also report candidate intervals arising from random fluctuations. These developments raise a natural question: how to quantify the uncertainty of these candidate intervals and select those that contain at least one true change-point.



Let \((X_1,\ldots,X_T)\) be an independent data sequence with an unknown number $K^*$ of change points at locations $0<\tau_1^*<\cdots<\tau_{K^*}^*<T$. The observations are independent and identically distributed (i.i.d.) within each segment, with different distributions across adjacent segments.
Let $\mathcal T^*=\{\tau_1^*,\ldots,\tau_{K^*}^*\}$ denote the true change-point set. A user-chosen detector reports a data-dependent family of candidate intervals \(\{I_k:k\in\widehat{\mathcal S}\}\), where \(\widehat{\mathcal S}\) is the index set of the reported intervals. Each interval $I_k$ represents the claim that at least one change occurs inside it. We associate these claims with the following multiple testing problem:
\begin{equation}
    H_{0k}:
    \mathcal T^*\cap I_k=\varnothing
    \quad\text{vs.}\quad
    H_{1k}:
    \mathcal T^*\cap I_k\neq\varnothing,\quad k\in \widehat{\mathcal{S}}.
    \label{eq:intro_interval_null}
\end{equation}

The aim is to obtain intervals containing true changes while limiting those false intervals. Instead of controlling the probability of any false report {in existing work \citep{jia2024tune,cui2026art}}, we aim to select a subset $\widehat{\mathcal R}\subseteq\widehat{\mathcal S}$ by measuring the expected proportion of false discoveries among the selected intervals.  
We define the interval-level false discovery rate (FDR) by
\begin{equation*}
    \mathrm{FDR}
    =
    \mathbb E
    \left[
    \frac{
    \sum_{k\in\widehat{\mathcal R}}
    \mathbbm 1\{\mathcal T^*\cap I_k=\varnothing\}
    }{
    |\widehat{\mathcal R}|\vee1
    }
    \right].
\end{equation*}
To accommodate different detectors and reported intervals, we aim to develop an algorithm-agnostic post-detection framework while controlling the interval-level FDR at a target level $\alpha \in (0,1)$ in finite samples.

\subsection{Our framework and contributions}
\label{subsec:intro_contributions}


For each interval, evidence for a change can be measured by a suitable score, but its null distribution generally depends on the unknown within-segment distribution. For a fixed interval, permutation tests provide a common calibration by exploiting exchangeability under the interval null
\citep{lehmann2005testing}. Comparing the observed score with scores computed after permuting the observations yields a finite-sample valid $p$-value for a fixed interval. However, this validity does not hold after the interval is reported by the detector using the same data. 
To address this detection bias, 
we propose a new framework Permutation-based Order-preserving Interval Selection, abbreviated as \texttt{POIS}, which is motivated by the order-preserving sample splitting strategy of \citet{zou2020consistent}. 
The detector reports intervals \(\{I_k:k\in\widehat{\mathcal S}\}\) based on the odd-indexed observations, while the even-indexed observations are reserved for permutation testing. 
Conditional on the odd-side data, the reported intervals are fixed, and exchangeability within each null interval ensures the validity of the even-side permutation $p$-values. Applying the Benjamini--Hochberg (BH) procedure \citep{benjamini1995controlling} to these $p$-values yields a selected subset 
$\widehat{\mathcal{R}}$ with finite-sample interval-level FDR control when the intervals in $\widehat{\mathcal{S}}$ are pairwise disjoint.

Sample splitting leaves the odd-side evidence unused in the final testing step, which may reduce the power of interval selection. We further develop \texttt{POIS+} to reuse this information by introducing \emph{detection-conditional} $e$-values as weights. \citet{ignatiadis2024values} establish the use of $e$-values as
unnormalized weights for a fixed family of hypotheses. In our setting,
the same data determine both the tested intervals and their weights,
so marginal $e$-value validity alone is insufficient for FDR control.
The central challenge is to retain useful odd-side evidence
while accounting for its role in determining the tested family.
We address this challenge by constructing detection-conditional
$e$-values through detection-preserving local permutations. Their
conditional validity allows weighted BH
\citep{genovese2006false} to reuse odd-side
information while retaining the finite-sample FDR guarantee.
Figure~\ref{fig:pois_plus_flowchart} summarizes the proposed framework's workflow.

\begin{figure}[htbp]
    \centering
    
    \includegraphics[width=\textwidth]{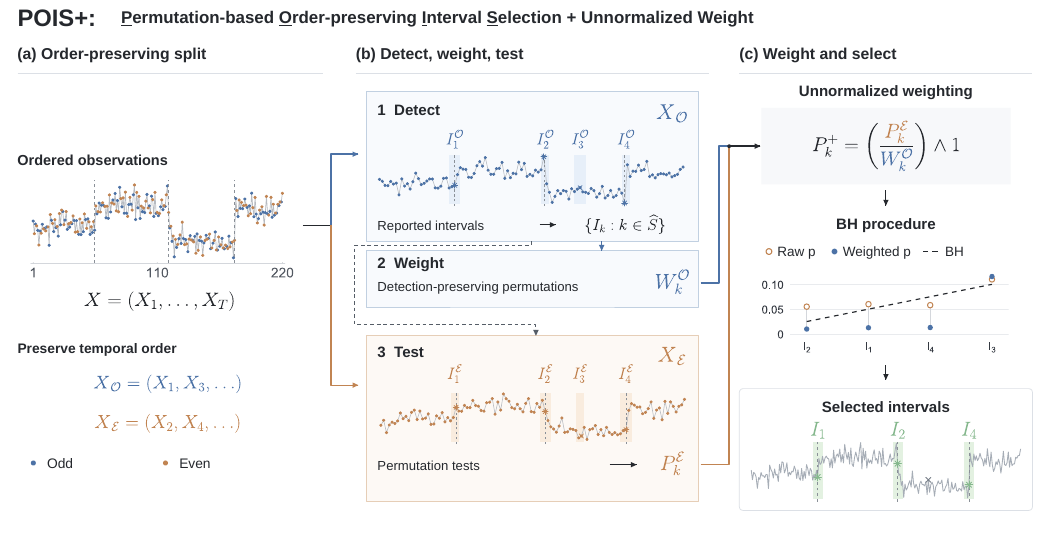}
    \caption{Illustration of \texttt{POIS+} for interval-level localization of multiple change-points. The odd subsequence $X_{\mathcal O}$ is used to run the detector, report intervals in $\widehat{\mathcal{S}}$ and contribute detection-conditional $e$-values as weights. The even subsequence $X_{\mathcal E}$ yields a permutation $p$-value within each reported interval for testing. The weighted $p$-values are then passed to BH, yielding the selected subset $\widehat{\mathcal{R}}$.}
    \label{fig:pois_plus_flowchart}
\end{figure}

Our contributions are listed as follows.
\begin{itemize}
    \item We formulate post-detection interval localization as a multiple testing problem and develop a general  \texttt{POIS} for interval-level FDR control. The proposed \texttt{POIS} procedure accommodates different detectors and interval scores with finite-sample guarantees. 

    \item We construct detection-conditional $e$-values for a data-dependent tested family. And we use these $e$-values to develop \texttt{POIS+}, which allows the same data to run the detector and provide unnormalized weights without compromising FDR control.

    \item We provide tractable implementations of \texttt{POIS+} for two practical classes of change-point detectors, avoiding full detector reruns under local permutations. 
    We also prove that the expected proportion of true change-points covered by
    \texttt{POIS} tends to one, and \texttt{POIS+} is asymptotically no worse than \texttt{POIS} under regular conditions.
\end{itemize}

\subsection{Related work}
\label{subsec:intro_related_work}

Existing work on change-point inference has primarily focused on assessing detected changes through individual tests and then constructing a selective $p$-value to control the type-I error. Under Gaussian data models, \citet{hyun2018exact} study change-points detected as jumps in the piecewise constant mean fitted by the fused lasso, and construct exact tests for mean contrasts across these detected change-points by conditioning on polyhedral selection events. Further, \citet{jewell2022testing} and \citet{carrington2025improving} develop tests for mean changes that condition on less information to improve power. Related methods address binary and optimal segmentation \citep{hyun2021post,duy2020computing}, changes detected by recurrent neural networks \citep{shiraishi2024selective}, piecewise polynomial signals \citep{mehrizi2021valid}, and variance changes \citep{carrington2026variance}. These approaches are typically tailored to particular detection algorithms. In particular, \citet{jia2024tune} develop an algorithm-agnostic framework for testing multiple hypotheses associated with detected change-points, providing asymptotic family-wise error rate (FWER) control under specific distributional assumptions.

Another line of change-point inference work directly constructs intervals with a simultaneous guarantee that each contains a change. Narrowest Significance Pursuit and its robust version provide such guarantees in linear-model and median-change settings
\citep{fryzlewicz2024narrowest,fryzlewicz2024robust}. 
\citet{jang2024fast} combine local parametric or rank tests over a sparse multiscale collection of triplets of time points using a weighted Bonferroni correction. 
\citet{cui2026art} construct rank statistics over a data-independent collection of candidate intervals and use the permutation distribution of their maximum to determine a common threshold, yielding distribution-free, finite-sample FWER control for reported intervals. In contrast, we propose permutation tests to achieve interval-level FDR control.



FDR control has been widely studied in multiple change-point detection problems. Existing procedures address mean changes \citep{hao2013multiple,li2016fdrseg,cheng2020multiple}, quantile changes \citep{liu2024multiscale}, changes in predictor--response relationships \citep{wu2025multiple}, and coefficient contrasts \citep{liu2024generalized}. These procedures differ in their definitions of discoveries and data distributions.
Closely related to our work, \citet{chen2023data} use order-preserving sample splitting, combining statistics from the two splits and exploiting null sign symmetry to obtain asymptotic FDR control for detected change-points. We use permutation tests to obtain finite-sample FDR control for detector-reported intervals, where we require only the exchangeability of data within each null interval. 

\subsection{Organization and notations}

The rest of the paper is organized as follows. Section~2 formulates interval-level post-detection selection and introduces \texttt{POIS} together with its finite-sample FDR guarantee. Section~3 develops \texttt{POIS+} and gives computationally tractable constructions for local-window and narrowest-first detectors. Section~4 studies the asymptotic power of the two procedures. Section~5 presents simulation studies and Section~6 gives two real-data applications.

For a positive integer \(m\), write \([m]=\{1,\ldots,m\}\). For two integers \(a\le b\), write \([a,b]=\{a,a+1,\ldots,b\}\). For a deterministic index subset \(I\subseteq[T]\), let \(X_I=(X_t)_{t\in I}\) be the subsequence, ordered by the time index. We also denote \(\Pi_I\) as the set of all permutations of \(I\). For a permutation \(\pi \in \Pi_I\), write the permuted data subsequence as \(X_{\pi(I)}\).

\section{Permutation-based Interval Selection after Change-Point Detection}
\label{sec:method}

Let \(\mathcal I=\{I_k\subset[1,T]:k\in[K]\}\) be a pool of possible candidate intervals, where each \(I_k\) is an interval that the detector may report. Throughout, this pool is assumed to be deterministic or fixed conditional on auxiliary randomization independent of \(X_{[T]}\). A detector \(\mathbb A\) examines all data subsequences \(\{X_{I_k}:k\in[K]\}\), produces a data-dependent index set \(\widehat{\mathcal{S}}\subseteq[K]\), and reports the corresponding intervals \(\{I_k:k\in \widehat{\mathcal{S}}\}\).

\subsection{Examples of detector outputs as interval pools}
\label{sec:examples}

Reported intervals may be produced directly by a detector or obtained from estimated locations through a prespecified rule. We present three popular examples in the following.


\begin{example}[Prespecified interval pools]
\label{exam:native_pool}
Consider a pool of candidate intervals
\(\mathcal I=\{I_k=[\ell_k,r_k]:k\in[K]\}\),
with \(1\le \ell_k<r_k\le T\).
The endpoints are fixed in advance or generated independently of the data.
Seeded binary segmentation uses a deterministic interval pool \citep{kovacs2023seeded}, while wild binary segmentation (WBS) \citep{fryzlewicz2014wild} and narrowest-over-threshold (NOT) \citep{baranowski2019narrowest} use randomly drawn intervals. The search interval associated with each detected change-point can serve as its reported interval.
Conditional on the data-independent randomization, the pool is fixed. Narrowest significance pursuit (NSP) and its robust version \citep{fryzlewicz2024narrowest,fryzlewicz2024robust}
directly report intervals of significance,
which can also be represented using a fixed pool containing all intervals.
\end{example}

\begin{example}[Local-window outputs]
\label{exam:window_pool}
Local-window methods associate each candidate point with a neighborhood. 
Thus they can be represented by attaching a fixed window to each possible center. For a bandwidth \(h\), take \(K=T\) and set \(I_k=[(k-h+1)\vee 1,(k+h)\wedge T]\) for \(k\in[T]\). Procedures based on moving sum (MOSUM) statistics or the screening and ranking algorithm (SaRa) \citep{eichinger2018mosum,niu2012screening} compute a local diagnostic or scan statistic around the center \(k\), and then report centers that pass a threshold and a local-maximum criterion. In the interval-pool representation, reporting the center \(k\) means reporting the interval \(I_k\), the window on which its evidence is computed. 
\end{example}

\begin{example}[Midpoint intervals from point estimators]
\label{exam:midpoint_pool}
For a detector returning only change-point estimates \(0=\widehat\tau_0<\widehat\tau_1<\cdots<\widehat\tau_{\hat k}<\widehat\tau_{\hat k+1}=T\), a reporting rule is needed to turn them into intervals. 
One convenient choice is the midpoint cell rule \citep{li2016fdrseg,chen2023data}. Define \(I_j=\left[\lfloor(\widehat\tau_{j-1}+\widehat\tau_{j})/2\rfloor+1,\lfloor(\widehat\tau_j+\widehat\tau_{j+1})/2\rfloor\right]\) for \(j=1,\ldots,\hat k\). Although their endpoints depend on the detected locations, one may take the conceptual all-interval pool \(\mathcal I=\{[\ell,r]:1\le \ell\le r\le T\}\), with the detector reporting the elements corresponding to those $I_j$'s.
\end{example}

These examples show that the outputs of a broad class of change-point detectors can be represented as a family of reported intervals $\{I_k:k\in\widehat{\mathcal S}\}$, with each interval associated with the hypothesis in \eqref{eq:intro_interval_null}. 
We next introduce our method \texttt{POIS}, starting with permutation testing for a fixed interval before addressing data-dependent detection. 

\subsection{\texttt{POIS}: finite-sample FDR control via order-preserving splitting}
\label{sec:POIS-interval}

Let \(v\) be a scoring function that maps an ordered subsequence to a nonnegative real number, with larger values indicating stronger evidence against the null. For example, \(v\) may be an energy-distance or maximum mean discrepancy (MMD) statistic comparing observations before and after a prespecified split point in \(I\) \citep{matteson2014nonparametric,gretton2012kernel}. The score may also maximize such a statistic over candidate split points within $I$ \citep{matteson2014nonparametric,song2024practical}. Denote $\Pi_I$ as the set of all permutations of $I$, and write the permuted data subsequence as \(X_{\pi(I)}=(X_{\pi(t)})_{t\in I}\) for \(\pi\in\Pi_I\). 

\paragraph*{Permutation test in a fixed interval.}
For a fixed interval \(I_k\), let \(\pi_{1},\ldots,\pi_{B} \overset{\mathrm{i.i.d.}}{\sim}\operatorname{Unif}(\Pi_{I_k})\) be sampled permutations, independently of the data. Given the scoring function \(v_k\), define
the observed and permuted scores by $V_k:=v_k(X_{I_k})$ and $V_{k,b}:=v_k(X_{\pi_{b}(I_k)})$ for $b\in[B]$, respectively.
The Monte Carlo permutation \(p\)-value \citep{lehmann2005testing} is defined as
\begin{equation*}
    P_k
    :=
    \frac{1+\sum_{b=1}^B
    \mathbbm 1\{V_k\le V_{k,b}\}}{B+1}.
\end{equation*}
Under the null \(H_{0k}\) in \eqref{eq:intro_interval_null}, the observations in \(I_k\) are i.i.d., and hence exchangeable. Applying the same scoring function to the observed and permuted orderings gives that \(P_k\) is superuniform, namely \(\mathbb P_{H_{0k}}(P_k\le u)\le u\) for every \(u\in[0,1]\). Thus, the permutation $p$-value provides a finite-sample and distribution-free guarantee on testing whether the interval contains a change-point.


This fixed-interval validity does not automatically extend to an interval detected by applying \(\mathbb A\) to the same data. For a detector-reported interval \(I_k\), the detection event generally depends on the temporal ordering of \(X_{I_k}\), so exchangeability need not be preserved conditional on detection. Consequently, although \(P_k\) is marginally superuniform for each fixed \(k\), applying these \(p\)-values only to detector-reported intervals does not generally yield valid post-detection inference. To separate detection from inference while preserving temporal order, \texttt{POIS} uses the odd--even splitting strategy adopted in previous work on change-point inference  \citep{zou2020consistent,chen2023data}.

Define \(\mathcal O=\{t\in[T]:t\text{ is odd}\}\) and \(\mathcal E=\{t\in[T]:t\text{ is even}\}\), and let \(I_k^{\mathcal O}=I_k\cap\mathcal O\) and \(I_k^{\mathcal E}=I_k\cap\mathcal E\). The detector is applied only to the odd-side subsequences \(\{X_{I_k^{\mathcal O}}:k\in[K]\}\), yielding the reported index set \(\widehat{\mathcal S}\subseteq[K]\). Conditional on $X_{\mathcal O}$, the reported intervals are fixed for even-side testing. The scoring functions can also be chosen using $X_{\mathcal O}$, and are held fixed when evaluating the even-side observations and their permutations.

For each \(k\in\widehat{\mathcal S}\), we draw permutations \(\pi_{k,1},\ldots,\pi_{k,B}
\overset{\mathrm{i.i.d.}}{\sim} \operatorname{Unif}(\Pi_{I_k^{\mathcal E}})\) independently of the data and across intervals. Define the observed score $V_k^{\mathcal E}:=v_k(X_{I_k^{\mathcal E}})$ and the permuted scores $V_{k,b}^{\mathcal E} :=v_k(X_{\pi_{k,b}(I_k^{\mathcal E})})$. The even-side permutation \(p\)-value is then
\begin{equation}\label{eq:perm-p-split}
    P_k^{\mathcal E}
    :=
     \frac{1+\sum_{b=1}^B
    \mathbbm 1\{V_k^{\mathcal E}\le V_{k,b}^{\mathcal E}\}}{B+1},
    \qquad
    k\in\widehat{\mathcal S}.
\end{equation}
Under the null \(H_{0k}\), the observations in \(I_k^{\mathcal E}\) remain exchangeable conditional on \(X_{\mathcal O}\). Thus, the permutation $p$-values in \eqref{eq:perm-p-split} are conditionally superuniform.
To apply the BH procedure to these $p$-values, we impose the following condition on the output of the detector \(\mathbb A\).

\begin{assumption}
\label{assum:disjoint_intervals}
The reported intervals \(\{I_k:k\in\widehat{\mathcal S}\}\) are pairwise disjoint almost surely.
\end{assumption}

Under the standing temporal-independence model, Assumption~\ref{assum:disjoint_intervals} ensures that the even-side samples associated with the reported intervals are mutually independent conditional on \(X_{\mathcal O}\). Together with independent permutation draws, this yields conditional independence of \(\{P_k^{\mathcal E}:k\in\widehat{\mathcal S}\}\). This condition is mild and is satisfied by most change detection procedures. For native interval outputs in Example~\ref{exam:native_pool}, disjoint reporting is available through variants of NSP and robust NSP. For local-window outputs in Example~\ref{exam:window_pool}, the local-extremum rule in the SaRa FDR formulation of \citet[Section~2.3]{hao2013multiple} yields disjoint reporting windows. The midpoint cells in Example~\ref{exam:midpoint_pool} are disjoint by construction. For overlapping outputs, Appendix~\ref{app:overlap_graph} adapts the Independent Set BH procedure of \citet{nguyen2025graphical} to account for dependence through the even-side overlap graph, while retaining the original reported intervals.

Finally, apply the BH procedure to those \(p\)-values \(\{P_k^{\mathcal E}:k\in\widehat{\mathcal S}\}\). Let \(P_{(1)}^{\mathcal E}\le\cdots\le P_{(|\widehat{\mathcal S}|)}^{\mathcal E}\) be the ordered permutation \(p\)-values and let \(\hat{r}= \max\{1\le i\le |\widehat{\mathcal S}|: P_{(i)}^{\mathcal E}\le i\alpha/|\widehat{\mathcal S}|\}\), with the convention \(\hat{r}=0\) if the set is empty. The rejection set is given by
\begin{equation}
\label{eq:BH_rejected_set}
    \widehat{\mathcal R}
    =
    \left\{
    k\in\widehat{\mathcal S}:
    P_k^{\mathcal E}\le P_{(\hat{r})}^{\mathcal E}
    \right\}.
\end{equation}
The resulting selected intervals are \(\{I_k:k\in\widehat{\mathcal R}\}\). We name this procedure Permutation-based Order-preserving Interval Selection, abbreviated as \texttt{POIS}. The following theorem establishes its finite-sample FDR control.

\begin{theorem}
\label{thm:FDR-POIS}
Let \(\mathcal H_0=\{k\in[K]:I_k\cap\mathcal T^*=\varnothing\}\) denote the set of null interval indices. Suppose Assumption~\ref{assum:disjoint_intervals} holds. Then the rejection set \(\widehat{\mathcal R}\) produced by \texttt{POIS} satisfies
\[
    \mathrm{FDR}(\widehat{\mathcal R})
    \le
    \alpha\cdot
    \mathbb E\left[
        \frac{|\widehat{\mathcal S}\cap\mathcal H_0|}
        {|\widehat{\mathcal S}|\vee1}
    \right]
    \le
    \alpha.
\]
\end{theorem}

Theorem \ref{thm:FDR-POIS} shows the proposed \texttt{POIS} provides
algorithm-agnostic, distribution-free FDR control. The guarantee follows from the conditional validity of permutation $p$-values and requires only within-interval exchangeability. 

\section{Enhancing power with post-detection weighting}
\label{sec:POIS+}

The \texttt{POIS} procedure applies BH to the even-side permutation $p$-values.
The odd-side data $X_{\mathcal{O}}$ is only used to obtain the reported intervals \(\{I_k: k\in \widehat{\mathcal{S}}\}\), and may provide additional evidence about which reported intervals contain true changes.
To integrate this information, we develop a refined procedure, \texttt{POIS+}, which uses the odd-side data to weight the even-side $p$-values and then applies the weighted BH procedure \citep{genovese2006false}.
The key requirement is that the weights account for the detection of the intervals they are intended to prioritize.

\subsection{\texttt{POIS+}: selection via weighted BH}

We first describe the general steps of \texttt{POIS+}.
For each reported interval \(I_k\) with \(k\in \widehat{\mathcal S}\), let \(W_k^{\mathcal O}\ge0\) be a weight computed from the odd-side subsequence \(X_{I_k^{\mathcal O}}\). The weight is intended to measure evidence against the null interval without a change-point, with larger values indicating stronger evidence in the final selection. Given the even-side permutation \(p\)-value \(P_k^{\mathcal E}\) in \eqref{eq:perm-p-split}, we define the weighted \(p\)-values
\[
    P_k^+
    :=
    (P_k^{\mathcal E}/W_k^{\mathcal O})\wedge 1,
    \qquad k\in\widehat{\mathcal S}.
\]
Let \(P_{(1)}^+ \leq \cdots \leq P_{(|\widehat{\mathcal S}|)}^+\) be the ordered weighted \(p\)-values, and define \(\hat{r}^+ = \max\{1\leq i \leq |\widehat{\mathcal S}|: P_{(i)}^+ \leq i\alpha/|\widehat{\mathcal S}|\}\). Then \texttt{POIS+} returns the weighted BH rejection set
\begin{equation}
\label{eq:weighted_BH_rejected_set}
    \widehat{\mathcal R}^+
    :=
    \left\{
    k\in\widehat{\mathcal S}:
    P_k^+\leq P_{(\hat{r}^+)}^+
    \right\}.
\end{equation}

Classical weighted BH requires the weights to be independent of the null \(p\)-values and normalized so that their average over the tested family is at most one, and FDR can be successfully controlled at the target level \citep{genovese2006false}. For a fixed family of hypotheses, \citet{ignatiadis2024values} show that this normalization is not necessary when the weights are valid \(e\)-values, i.e., \(\E_{H_{0k}}[W_k^{\mathcal{O}}] \leq 1\), and are independent of the \(p\)-values. 
However, the odd-side data determine both the tested family $\{H_{0k}:k\in\widehat{\mathcal S}\}$ and its weights. Selection may favor intervals with larger weights, and so the validity of the weights must account for the detection events.


\subsubsection{A sufficient condition for FDR control of \texttt{POIS+}}

To obtain requirements for the odd-side weights, we proceed with establishing an upper bound for the FDR of the \texttt{POIS+} rejection set in \eqref{eq:weighted_BH_rejected_set}. Conditional on \(X_{\mathcal O}\), \(\widehat{\mathcal S}\) and the weights \(\{W_k^{\mathcal O}:k\in\widehat{\mathcal S}\}\) are fixed. The even-side permutation \(p\)-values retain their conditional null validity and their conditional independence under Assumption~\ref{assum:disjoint_intervals}. 
Therefore, the standard weighted BH bound gives the following result.

\begin{theorem}
\label{thm:FDR_POIS+}
Suppose Assumption~\ref{assum:disjoint_intervals} holds. Then the \texttt{POIS+} rejection set \(\widehat{\mathcal R}^+\) defined in \eqref{eq:weighted_BH_rejected_set} satisfies \(\mathrm{FDR}(\widehat{\mathcal R}^+) \leq \alpha\cdot \mathbb E\left[ (|\widehat{\mathcal S}|\vee 1)^{-1} \sum_{k\in\widehat{\mathcal S}\cap\mathcal H_0} W_k^{\mathcal O} \right].\)
\end{theorem}

Thus, \texttt{POIS+} controls the FDR at level \(\alpha\) whenever the following condition for the average over the reported index set holds,
\begin{equation} \label{eq:post_detection_Evalue_error} 
\mathbb E\left[ \left(|\widehat{\mathcal S}|\vee 1\right)^{-1} \sum_{k\in\widehat{\mathcal S}\cap\mathcal H_0} W_k^{\mathcal O} \right] \leq 1. 
\end{equation}
Classical weight normalization would immediately imply this condition. For unnormalized weights, we instead require validity conditional on detection. 
For each null index \(k\in\mathcal H_0\), let \(\Psi_k\) be a conditioning statistic chosen so that \(|\widehat{\mathcal S}|\) is measurable on the event \(\{k\in\widehat{\mathcal S}\}\). Formally, we call \(W_k^{\mathcal O}\) a \textit{detection-conditional} \(e\)-value if
\begin{equation}
\label{eq:cond_Evalue}
    \mathbb E_{H_{0k}}
    \left[
    W_k^{\mathcal O}
    \mid
    k\in\widehat{\mathcal S},\Psi_k
    \right]
    \leq 1.
\end{equation}
Here \(\Psi_k\) is chosen so that $|\widehat{\mathcal S}|$ is measurable on the event $\{k\in\widehat{\mathcal S}\}$.
This decouples the dependence between the numerator and denominator of \eqref{eq:post_detection_Evalue_error}. The next theorem shows that the conditional property \eqref{eq:cond_Evalue} can imply \eqref{eq:post_detection_Evalue_error} and therefore yield finite-sample FDR control for \texttt{POIS+}.

\begin{theorem}
\label{thm:sufficient_cond_Evalue}
Suppose that, for every \(k\in\mathcal H_0\), the following two conditions hold: (i) \(W_k^{\mathcal O}\) is a detection-conditional \(e\)-value satisfying \eqref{eq:cond_Evalue}; (ii) \(|\widehat{\mathcal S}|\) is measurable with respect to \(\Psi_k\) conditional on the event \(\{k\in\widehat{\mathcal S}\}\). Then the relation \eqref{eq:post_detection_Evalue_error} holds. Further, under the conditions of Theorem~\ref{thm:FDR_POIS+}, the rejection set of \texttt{POIS+} in \eqref{eq:weighted_BH_rejected_set} satisfies \(\mathrm{FDR}(\widehat{\mathcal R}^+)\leq\alpha\).
\end{theorem}

Theorem~\ref{thm:sufficient_cond_Evalue} shows that detection-conditional validity \eqref{eq:cond_Evalue} suffices for the average-weight bound in \eqref{eq:post_detection_Evalue_error}, provided the conditioning fixes the reported size. The weights therefore need not be normalized across the reported family. We next construct such weights by locally permuting the odd-side observations and retaining only those permutations that preserve the selection of the focal interval \(I_k\) and the reported size \(|\widehat{\mathcal{S}}|\).

\subsection{Construction of detection-conditional \(e\)-values}

For notational convenience, we express the reported index set in the function form $\widehat{\mathcal S} = \mathbb{A}\left(\{X_{I_k}^{\mathcal{O}}\}: k\in [K]\right)$.
For any \(k\in\widehat{\mathcal S}\), let \(\lbag X_t\rbag_{t\in I_k^{\mathcal O}}\) denote the unordered set of observations inside the interval \(I_k^{\mathcal O}\), and denote \(\widehat{\mathcal S}_{-k}:=\widehat{\mathcal S}\setminus\{k\}\). We define the conditioning statistic as
\begin{equation}
\label{eq:cond_statistic}
    \Psi_k
    =
    \left(
    \lbag X_t\rbag_{t\in I_k^{\mathcal O}},
    \ (X_t)_{t\in\mathcal O\setminus I_k^{\mathcal O}},
    \ |\widehat{\mathcal S}_{-k}|
    \right).
\end{equation}
The first two components of \(\Psi_k\) fix the local observations without their order and the ordered observations outside the interval.
Conditional on this information, the only remaining randomness is the ordering of the observations inside \(I_k^{\mathcal O}\), which is exchangeable under the null \(H_{0k}\). The last component fixes the number of other reported intervals, so that
$|\widehat{\mathcal S}|=1+|\widehat{\mathcal S}_{-k}|$ is measurable with respect to $\Psi_k$ on $\{k\in\widehat{\mathcal S}\}$.

Let \(\Pi_k\) be the collection of all permutations in \(I_k^{\mathcal O}\). For \(\pi \in \Pi_k\), we use $\widehat{\mathcal S}^{\pi,k}:=\mathbb A\left((X_t)_{t\in\mathcal O\setminus I_k^{\mathcal O}},X_{\pi(I_k^{\mathcal O})}\right)$ to denote the detector output after permuting the observations inside $I_k^{\mathcal O}$, with all outside observations kept in their original positions. Any independent randomization used by the detector is held fixed across these permutations.

To construct the odd-side weights satisfying the detection-conditional property \eqref{eq:cond_Evalue}, we first define the accepted subset of permutations that reproduce the detection event:
\begin{equation}
\label{eq:accepted_index}
   \widetilde{\Pi}_k
    =
    \left\{
    \pi \in \Pi_k: k\in\widehat{\mathcal S}^{\pi,k},
    |\widehat{\mathcal S}^{\pi,k}|
    =
    |\widehat{\mathcal S}|
    \right\}.
\end{equation}
Draw the permutations \(\pi_{k,1},\ldots,\pi_{k,B'}\) independently from \(\operatorname{Unif}(\Pi_k)\). We compute the permuted scores by \(V_{k,b}^{\mathcal O}:=v_k\left(X_{\pi_{k,b}(I_k^{\mathcal O})}\right)\) for \(b=1,\ldots,B'\). Conditional on \(\Psi_k\) and the event \(\{k\in\widehat{\mathcal S}\}\), the observed score \(V_k^{\mathcal O} := v_k(X_{I_k^{\mathcal{O}}})\) and the accepted permuted scores \(\{V_{k,b}^{\mathcal O}:\pi_{k,b}\in\widetilde{\Pi}_k\}\) are exchangeable under the null \(H_{0k}\).
Let \(g:\mathbb R\to[0,\infty)\) be a nonnegative transformation. We introduce the detection-conditional weights as
\begin{equation}
\label{eq:post_selection_Evalue}
    W_k^{\mathcal O}
    =
    \frac{
    \left(1+\sum_{b=1}^{B'} \mathbbm{1}\{\pi_{k,b}\in \widetilde{\Pi}_k\}\right)\cdot g(V_k^{\mathcal O})
    }{
    g(V_k^{\mathcal O})
    +
    \sum_{b=1}^{B'} \mathbbm{1}\{\pi_{k,b}\in \widetilde{\Pi}_k\}
    g(V_{k,b}^{\mathcal O})
    },\quad k\in \widehat{\mathcal{S}}.
\end{equation}
The weight in \eqref{eq:post_selection_Evalue} is the observed transformed score divided by the average over itself and the accepted scores. It can be viewed as a detection-conditional version of the soft-rank \(e\)-value \citep[Section~1.7]{ramdas2025hypothesis}, where the marginal version uses all permuted scores in the denominator.

\begin{remark}
The construction of detection-conditional $e$-values is closely related to the selective inference problems studied in conformal inference \citep{bao2024selective,wang2024conformalized,jin2025confidence}. In the conformal setting, the inferential target is often a selected test point with an unobserved label, while the calibration scores are exchangeable with the test score before selection. For example, \citet{bao2024selective} screen calibration scores to maintain the exchangeability between the calibration score and the test score given the selection event. Here, the accepted permutation set \eqref{eq:accepted_index} is designed to keep the exchangeability between the observed score and the permuted scores given the realized detection event.
\end{remark}

The following theorem shows the conditional validity of the proposed weight.

\begin{theorem}
\label{thm:perm-e-validity}
For any detector \(\mathbb A\), the weights defined in \eqref{eq:post_selection_Evalue} are detection-conditional \(e\)-values with respect to \(\Psi_k\) in \eqref{eq:cond_statistic}, that is, for every null interval \(k\in\mathcal H_0\), \(\mathbb E_{H_{0k}} \left[ W_k^{\mathcal O} \mid \ k\in\widehat{\mathcal S},\Psi_k \right] \leq 1.\)
Consequently, under Assumption~\ref{assum:disjoint_intervals}, \texttt{POIS+} enjoys finite-sample FDR control, i.e., \(\mathrm{FDR}(\widehat{\mathcal R}^+)\leq\alpha\).
\end{theorem}

We conclude the section with a concise implementation of \texttt{POIS+} in Algorithm \ref{alg:POIS_plus}. The algorithm obtains the reported intervals and even-side permutation \(p\)-values, computes the odd-side detection-conditional \(e\)-values, and applies BH to the resulting weighted \(p\)-values.

\begin{algorithm}[htbp]
\caption{\texttt{POIS+}: Boosting permutation-based POIS with odd-side \(e\)-values}
\label{alg:POIS_plus}
\begin{algorithmic}[1]
\Require Data sequence \(X_{[T]}\); detector \(\mathbb A\); scoring functions \(\{v_k\}_{k=1}^K\); permutation numbers \(B\) and \(B'\); target level \(\alpha\); nonnegative transformation \(g\).
    \Procedure{\texttt{POIS+}}{$X_{[T]},\mathbb A,v_k,B,B',\alpha,g$}
    \State Apply \(\mathbb A\) to \(X_{\mathcal O}\) to obtain
    \(\{I_k:k\in\widehat{\mathcal S}\}\).
    \State Compute \(\{P_k^{\mathcal E}:k\in\widehat{\mathcal S}\}\)
    using \eqref{eq:perm-p-split}.
    \For{\(k \in \widehat{\mathcal S}\)}
        \State Draw
        \(\pi_{k,1},\ldots,\pi_{k,B'}\)
        independently from \(\operatorname{Unif}(\Pi_k)\).
        \State Compute the accepted permutation set \(\widetilde{\Pi}_k\)
        in \eqref{eq:accepted_index}.
        \State Compute odd-side scores
        \(V_k^{\mathcal O}
        =v_k(X_{I_k^{\mathcal O}})\) and \(V_{k,b}^{\mathcal O}=v_k(X_{\pi_{k,b}(I_k^{\mathcal O})})\) for \(b \in [B']\).
        \State Compute the detection-conditional \(e\)-value $W_k^{\mathcal O}$ in \eqref{eq:post_selection_Evalue}.
        \State Compute the weighted \(p\)-value
        \(P_k^+=(P_k^{\mathcal E}/W_k^{\mathcal O})\wedge 1\).
    \EndFor
    \State Apply BH at level \(\alpha\) to
    \(\{P_k^+:k\in\widehat{\mathcal S}\}\), yielding the rejection set
    \(\widehat{\mathcal R}^+\).
    \State \Return \(\{I_k:k\in\widehat{\mathcal R}^+\}\).
\EndProcedure
\end{algorithmic}
\end{algorithm}

\begin{remark}
An alternative way is to construct the detection-conditional $e$-values in \eqref{eq:cond_Evalue} in the full data sequence and then apply the $e$-value version of the BH procedure ($e$-BH) \citep{wang2022false} to these conditional $e$-values. This can control FDR in finite samples without sample splitting or disjoint reported intervals; see Appendix~\ref{app:overlap}. However, few permutations may preserve the required detection event in \eqref{eq:accepted_index}, especially for short intervals. Hence, the resulting conditional $e$-values may be too small to produce any rejections under the $e$-BH procedure, as illustrated in Appendix~\ref{app:epois-signal-sensitivity}. We thus advocate using conditional $e$-values as weights for the even-side permutation $p$-values, as in \texttt{POIS+}.
\end{remark}


\subsection{Instantiations of \texttt{POIS+}}
\label{sec:instantiations}

A direct implementation of \texttt{POIS+} can be computationally demanding, as checking \eqref{eq:accepted_index} generally requires rerunning the detector after each local permutation.
For the local-window and narrowest-first detectors in Examples~\ref{exam:window_pool} and~\ref{exam:native_pool}, we give local acceptance rules that avoid these full reruns while retaining conditional validity. The rules strengthen \eqref{eq:accepted_index} by preserving the entire reported interval family. 
For the midpoint-cell conversion in Example~\ref{exam:midpoint_pool}, a local permutation may change both detected locations and the resulting cell boundaries, so checking \eqref{eq:accepted_index} generally requires rerunning the detector. Throughout this section, \(\widehat Q_k\) denotes a detection statistic computed from \(X_{I_k^{\mathcal{O}}}\), which is used only by the detector $\mathbb A$ to report $\widehat{\mathcal S}$.

\subsubsection{Local-window scans}

Consider a deterministic local-window detector in Example \ref{exam:window_pool}, representative of MOSUM and SaRa-type procedures \citep{hao2013multiple,eichinger2018mosum}.
For a bandwidth $h$, the candidate intervals are \(I_k = [(k-h+1)\vee1,(k+h)\wedge T]\) for $k\in [T]$.
Let \(\lambda\) denote the detection threshold, and let \(\rho_h\) be the local competition radius used in the final local maximum criterion. The detector first forms the set of over-threshold local maxima: \(\widehat{\mathcal C} = \left\{ k\in[T]: \widehat Q_k>\lambda \text{ and } \widehat Q_k> \max_{j\in[T]:\,0<|j-k|<h}\widehat Q_j \right\}.\)
It then applies a local-maximum criterion: \(D_k = \mathbbm 1 \left\{ k\in\widehat{\mathcal C} \text{ and } \widehat Q_k> \max_{j\in\widehat{\mathcal C}:\,0<|j-k|\le\rho_h} \widehat Q_j \right\}\). The reported index set is given by \(\widehat{\mathcal S}=\{k\in[T]:D_k=1\}\).

For \(k\in\widehat{\mathcal S}\), permuting the observations in \(I_k^{\mathcal O}\) can affect \(\widehat Q_j\) only if \(I_j^{\mathcal O}\cap I_k^{\mathcal O}\neq\varnothing\), which implies \(|j-k|\le 2h\). Since \(D_m\) depends only on scores indexed within \(h+\rho_h\) of \(m\), such a permutation can change \(D_m\) only if \(|m-k|\le 3h+\rho_h\). For each permutation \(\pi\in\Pi_k\), let \(D_m^\pi\) denote the recomputed detection indicator after replacing \(X_{I_k^{\mathcal O}}\) with \(X_{\pi(I_k^{\mathcal O})}\). We define the set of locally detection-preserving permutations as
\[
    \widetilde\Pi_k^{\rm win}
    =
    \left\{
        \pi\in\Pi_k:
        D_m^\pi=D_m
        \text{ for all }m\text{ satisfying }
        |m-k|\le 3h+\rho_h
    \right\}.
\]
This implies that only the potentially affected selection indicators need
to be checked. 
The next proposition gives the validity of \texttt{POIS+} weights by letting $\widetilde \Pi_k = \widetilde\Pi_k^{\rm win}$ in \eqref{eq:post_selection_Evalue}.

\begin{proposition}
\label{prop:window_cond_invariance}
Suppose that the detector is the deterministic local-window procedure above. Then, for every \(k\in\widehat{\mathcal S}\) and \(\pi\in\widetilde\Pi_k^{\rm win}\), \(\widehat{\mathcal S}^{\pi,k}=\widehat{\mathcal S}\). Moreover, for each null interval, the weight constructed from \(\widetilde\Pi_k^{\rm win}\) through \eqref{eq:post_selection_Evalue} satisfies \(\mathbb E_{H_{0k}} \left[ W_k^{\mathcal O} \mid k\in\widehat{\mathcal S},\Psi_k \right] \le 1.\)
\end{proposition}


Here, preserving the checked indicators is equivalent to preserving the complete selection vector. The conditional validity follows by conditioning on this vector and then averaging back to $\Psi_k$; see Appendix~\ref{app:proof_window_cond_invariance}.

\begin{proposition}
\label{prop:window_cost}
Let \(\mathfrak C_{\mathrm{win}}\) be the cost of one odd-side local-window scan. Assume that scores and local-maximum indicators on a block of \(m\) centers can be updated in \(O(m)\) time after a local permutation. Under Assumption~\ref{assum:disjoint_intervals}, the local-window implementation of \texttt{POIS+} has cost \(\mathfrak C_{\mathrm{POIS+}}^{\mathrm{win}}=\mathfrak C_{\mathrm{win}}+O\left\{(B+B')T+|\widehat{\mathcal S}|\log|\widehat{\mathcal S}|\right\}\).
\end{proposition}

The additional term \(O(B'T)\) comes from computing the odd-side weights. A local permutation around a reported center \(k\) affects only centers within distance \(3h+\rho_h\), so one permutation costs \(O\{\min(T,h+\rho_h)\}\). The two local-maximum conditions ensure that distinct reported centers are at least \(h\) apart and more than \(\rho_h\) apart, giving \(|\widehat{\mathcal S}|\min\{T,h+\rho_h\}=O(T)\). The remaining terms account for even-side calculation and
BH sorting.


\subsubsection{Narrowest-first interval searches}

We next consider a recursive narrowest-first search motivated by the narrowest-over-threshold principle \citep{baranowski2019narrowest}. On a search segment \(G=[u,v]\), the procedure chooses the shortest interval \(I_j\subseteq G\) satisfying \(\widehat Q_j>\lambda\). Among intervals of the same length, it chooses the one with the largest statistic, with any ties resolved by a deterministic rule. Denote the chosen index by \(\hat k(G)\), and set \(\hat k(G)=\varnothing\) if no interval exceeds the threshold.

The search starts from \([1,T]\). If \(\hat k(G)=\varnothing\), the recursion stops on \(G\). Otherwise, writing \(I_{\hat k(G)}=[\ell_{\hat k(G)},r_{\hat k(G)}]\), the procedure reports this interval and continues recursively on the two nonempty segments \([u,\ell_{\hat k(G)}-1]\) and \([r_{\hat k(G)}+1,v]\). Let \(\widehat{\mathcal S}\) collect all indices reported during the recursion. By construction, the reported intervals are pairwise disjoint.
Fix \(k\in\widehat{\mathcal S}\), and let \(\mathcal P_k\) be the observed path from \([1,T]\) to the segment on which \(I_k\) is reported. For \(\pi\in\Pi_k\), let \(\widehat Q_j^\pi\) be the statistic obtained after permuting only \(I_k^{\mathcal O}\), and let \(\hat k^\pi(G)\) be the index chosen on \(G\) by the same narrowest-first rule using the permuted statistics. Since
\(\widehat Q_j^\pi=\widehat Q_j\) whenever
\(I_j^{\mathcal O}\cap I_k^{\mathcal O}=\varnothing\), only statistics of intervals intersecting \(I_k^{\mathcal O}\) need to be recomputed. Define
\[
    \widetilde\Pi_k^{\rm nar}
    =
    \left\{
        \pi\in\Pi_k:
        \hat k^\pi(G)=\hat k(G)
        \text{ for all }G\in\mathcal P_k
    \right\}.
\]
Thus, membership in \(\widetilde\Pi_k^{\rm nar}\) can be checked along the observed path without rerunning the full recursion. Preserving the interval chosen at every segment in \(\mathcal P_k\) preserves all splits leading to \(I_k\), while all remaining branches are disjoint from \(I_k\), and therefore remain unchanged.

\begin{proposition}
\label{prop:narrow_cond_invariance}
Suppose that the odd-side detector is the narrowest-first interval search above. Then, for every \(k\in\widehat{\mathcal S}\) and \(\pi\in\widetilde\Pi_k^{\rm nar}\), we have \(\widehat{\mathcal S}^{\pi,k}=\widehat{\mathcal S}\). Moreover, for each null interval, the weight constructed from \(\widetilde\Pi_k^{\rm nar}\) through \eqref{eq:post_selection_Evalue} satisfies
\(\mathbb E_{H_{0k}}\!\left[W_k^{\mathcal O}\mid k\in\widehat{\mathcal S},\Psi_k\right]\le1\).
\end{proposition}

The accepted permutations preserve the observed path and its detection results. Conditioning on this information gives conditional exchangeability, and averaging back to $\Psi_k$ yields the
stated validity; see Appendix~\ref{app:proof_narrow_cond_invariance}.

\begin{proposition}
\label{prop:narrowest_cost}
Let \(\mathfrak C_{\mathrm{nar}}\) be the cost of one odd-side
narrowest-first interval search. For scores whose evaluation cost is linear in
the interval length, the narrowest-first implementation of \texttt{POIS+} has cost $\mathfrak C_{\mathrm{POIS+}}^{\mathrm{nar}}
=
\mathfrak C_{\mathrm{nar}}
+
O\left\{
BT+
B'\sum_{I_j\in\mathcal I} |I_j^{\mathcal O}|
+
|\widehat{\mathcal S}|\log|\widehat{\mathcal S}|
\right\}.$
\end{proposition}

The term \(O\{B'\sum_{I_j\in\mathcal I}|I_j^{\mathcal O}|\}\) comes from the detection-preserving checks.
Using the affected-competitor sets and cached comparison bounds described in Appendix~\ref{app:computational_cost}, each check evaluates each affected score only once and performs a linear number of comparisons in the number of affected competitors. The cumulative bound \eqref{eq:narrow_charging_bound} then gives the
stated cost. In particular, seeded intervals \citep{kovacs2023seeded} imply \(\sum_{I_j\in\mathcal I}|I_j^{\mathcal O}|
\le \sum_{I_j\in\mathcal I}|I_j|=O(T\log T)\), so the odd-side checking cost is \(O(B'T\log T)\).

\section{Power Analysis}
\label{subsec:asymp-TPR}

This section analyzes the asymptotic power of \texttt{POIS} and the effect of odd-side weighting on the power of \texttt{POIS+}. For a selected index set $\mathcal R\subseteq\widehat{\mathcal S}$, define the true positive proportion (TPP) and the true positive rate (TPR) as
\[
    \operatorname{TPP}(\mathcal R):=\frac{1}{K^*}\sum_{j=1}^{K^*}\mathbbm1\left\{\exists k\in\mathcal R:\tau_j^*\in I_k\right\},\qquad \operatorname{TPR}(\mathcal R)=\mathbb E\{\operatorname{TPP}(\mathcal R)\}.
\]
Thus TPR is the expected proportion of true change-points covered by at least one selected interval. Given the prespecified interval pool $\mathcal I=\{I_k:k\in[K]\}$, let $I_j^*:=I_{k_j^*}\in\mathcal I$ be a target interval containing $\tau_j^*$ for $j\in[K^*]$. We assume that the target indices $k_1^*,\ldots,k_{K^*}^*$ are distinct. Throughout this section, suppose $K^*\to\infty$ as $T\to\infty$.

\subsection{Asymptotic power of \texttt{POIS}}
\label{sec:TPR}

To analyze the power of \texttt{POIS}, we impose the following regularity conditions on the output of the detector and the even-side permutation scores.

\begin{assumption}
\label{ass:power_main}
As $T\to\infty$, the following conditions hold: (i) The reported index set satisfies $\mathbb P(\{k_1^*,\ldots,k_{K^*}^*\}\subseteq\widehat{\mathcal S})=1-o(1)$. For some fixed constant $C_0\ge1$, $\mathbb P(|\widehat{\mathcal S}|\le C_0K^*)=1-o(1)$ and $B+1>C_0/\alpha$; 
(ii) $(K^*)^{-1}\sum_{j=1}^{K^*}q_{j,T}^{\mathcal E}=o(1)$, where $q_{j,T}^{\mathcal E}:=\mathbb P\left(\max_{1\le b\le B}V_{k_j^*,b}^{\mathcal E}\ge V_{k_j^*}^{\mathcal E}\right)$. 

\end{assumption}

Condition~(i) requires the detector $\mathbb A$ to report the target intervals with probability tending to one, while keeping the number of tested intervals proportional to $K^*$. A related requirement is also imposed in \citet{chen2023data}. The restriction on the number of permutations $B$ ensures that the smallest attainable permutation $p$-value (i.e., $(B+1)^{-1}$) lies strictly below $\alpha/C_0$. Condition~(ii) requires the observed even-side scores to separate from their permutation counterparts on average over the target intervals. Similar signal-strength requirements appear in \citet[Assumption~3.3]{fryzlewicz2014wild} and \citet[Theorem~1]{baranowski2019narrowest}.

\begin{theorem}
\label{thm:pois_power}
Suppose Assumptions~\ref{assum:disjoint_intervals} and \ref{ass:power_main} hold. Then the rejection set $\widehat{\mathcal R}$ of \texttt{POIS} in \eqref{eq:BH_rejected_set} satisfies $\operatorname{TPR}(\widehat{\mathcal R})=1-o(1)$.
\end{theorem}

Theorem~\ref{thm:pois_power} establishes asymptotically full expected coverage of \texttt{POIS} with a fixed even-side permutation number $B$.
Under Condition~(ii), most target $p$-values attain the smallest possible value $1/(B+1)$. The reported size bound in Condition~(i) then makes these targets numerous enough to pass the BH threshold.




\subsection{Power comparison of \texttt{POIS} and \texttt{POIS+}}
\label{subsec:poisplus_power_transfer}

We next compare \texttt{POIS+} with  \texttt{POIS} on the same reported family, without imposing Assumption~\ref{ass:power_main}\textup{(ii)}. The BH rejection sets based on $P_k^{\mathcal E}$ and $P_k^+=(P_k^{\mathcal E}/W_k^{\mathcal O})\wedge1$ are denoted by $\widehat{\mathcal R}$ and $\widehat{\mathcal R}^{+}$, respectively. Since weighting changes individual $p$-values as well as the BH cutoff, the power comparison requires joint control of the even-side and odd-side permutation scores.

For each target interval $I_j^*$, let $\mathcal A_j:=\{b\in[B']:\pi_{k_j^*,b}\in\widetilde\Pi_{k_j^*}\}$ index its accepted sampled permutations if it is reported, and set $\mathcal A_j=\varnothing$ otherwise. 
Define $\mathcal J_T=\left\{j\in[K^*]:k_j^*\in\widehat{\mathcal S},\ V_{k_j^*}^{\mathcal E}>\max_{1\le b\le B}V_{k_j^*,b}^{\mathcal E}\right\}$, which contains the reported targets whose
even-side $p$-values attain the smallest possible value $1/(B+1)$.

\begin{assumption}
\label{ass:odd_nondeflation}
(i) There exists a fixed $\gamma>0$ such that the transformation function $g$ in \eqref{eq:post_selection_Evalue} satisfies $g(0)=0$, $g(v)>0$ for $v>0$, and $g(tv)\ge t^\gamma g(v)$ for $v>0$ and $t>1$. (ii) For this $\gamma$, there exist fixed $\delta,d\in(0,1)$ satisfying $\delta^\gamma+2\alpha(1-C_0^{-1})<1$ and $d>C_0/\{\alpha(B+1)\}$. Assume $\mathbb P\{|\mathcal J_T|<dK^*\}+\frac1{K^*}\sum_{j=1}^{K^*}q_{j,T}^{+}=o(1)$, where $q_{j,T}^{+}=\mathbb P\!\left[\left\{\max_{b\in\mathcal A_j}V_{k_j^*,b}^{\mathcal O}\ge\delta V_{k_j^*}^{\mathcal O}\right\}\cup\{j\notin\mathcal J_T,\ |\mathcal A_j|=0\}\right]$.


\end{assumption}

Every fixed power transform $g(v)=v^\gamma$, $\gamma>0$, satisfies Condition~(i). Condition~(ii) requires at least $dK^*$ targets with strictly separated even-side scores with probability tending to one. Apart from $o_p(K^*)$ reported targets, it also requires the observed odd-side score to be positive and dominate its accepted permutation scores, and those outside $\mathcal J_T$ have at least one accepted proposal. 
These requirements distinguish two cases in the power comparison. On a target satisfying the joint score condition,
a nonempty accepted set gives $W_{k_j^*}^{\mathcal O}\ge2/(1+\delta^\gamma)>1$.
An empty accepted set instead gives weight one and is allowed only when $j\in\mathcal J_T$. Typically, for $\alpha=0.1$, choosing $\delta^\gamma\le1/2$ ensures $\delta^\gamma+2\alpha(1-C_0^{-1})<1$ for every $C_0\ge1$. Appendices~\ref{sec:odd_nondeflation_score_conditions} and~\ref{app:power:acceptance} give sufficient conditions for (ii). 

\begin{theorem}
\label{thm:poisplus_BH_comparison}
Suppose Assumption~\ref{assum:disjoint_intervals}, Assumption~\ref{ass:power_main} (i), and Assumption~\ref{ass:odd_nondeflation} hold. Then $\liminf_{T\to\infty}\{\operatorname{TPR}(\widehat{\mathcal R}^{+})-\operatorname{TPR}(\widehat{\mathcal R})\}\ge0.$
\end{theorem}

Theorem \ref{thm:poisplus_BH_comparison} shows that \texttt{POIS+} has asymptotically no lower expected coverage proportion of true change-points than \texttt{POIS}.

\section{Simulation Studies}
\label{sec:simulations}

In this section, we examine the performance of \texttt{POIS} and \texttt{POIS+} in three simulation studies: distributional changes with different detectors, high-dimensional mean changes, and changes in profile data. We report empirical interval-level FDR and TPR over 500 repetitions. The nominal level is $\alpha=0.1$, with $B=B'=499$ permutations and $g(v)=v^{12}$ in \eqref{eq:post_selection_Evalue}.
The number of true change-points is set as $K^*=20$ in Section~\ref{subsec:sim_midpoint_detectors} and $K^*=10$ in Sections~\ref{subsec:sim_hd_mean_changes} and~\ref{subsec:sim_profile_pois}. 

We consider three related procedures as benchmarks, covering FDR-controlled change-point selection methods, and FWER-controlled post-detection inference. \texttt{MOPS} \citep{chen2023data} combines order-preserving sample splitting with mirror statistics to obtain asymptotic FDR control for detected change-points.
\texttt{SPD-OSE} \citep{wu2025multiple} uses smooth profile decomposition and mirror statistics for FDR-controlled change-point selection in profile data.
\texttt{ART} \citep{cui2026art} uses rank-based permutation calibration to provide finite-sample FWER control.
These procedures differ in their inferential targets and
calibration assumptions. In each experiment, \texttt{POIS} and \texttt{POIS+} share
the reported intervals and even-side $p$-values, so their
comparison isolates the effect of detection-conditional weighting.



\subsection{Distributional Changes with Different Detectors}
\label{subsec:sim_midpoint_detectors}

We first examine the framework with different detection procedures. We generate \(T=3000\) independent observations with \(K^*=20\) change-points at fixed, unequally spaced locations. Consecutive segments alternate between
\[
F_A=N(0,1),\qquad
F_B=\tfrac12 N\!\left(\delta-r\rho,\,r^2(1-\rho^2)\right)
    +\tfrac12 N\!\left(\delta+r\rho,\,r^2(1-\rho^2)\right),
\]
where the second argument of \(N\) denotes the variance. We fix \(r=1.25\) and set three scenarios \texttt{Weak}, \texttt{Strong}, and \texttt{Shift} with \((\delta,\rho)=(0.4,0.96)\), \((0.4,0.99)\), and \((0.6,0.99)\), respectively. In these three scenarios, the two states differ in mean,
variance, and shape. \texttt{Weak} and \texttt{Strong} differ only in $\rho$, which alters the shape of $F_B$ while keeping its mean and variance fixed. 
We also consider the \texttt{Shape} scenario with \((\delta,r,\rho)=(0,1,0.99)\), which isolates a change in distributional shape with $F_A$ and $F_B$ sharing the same mean and variance. 

We use PELT \citep{killick2012optimal} with an empirical CDF feature cost (\texttt{PELT-CDF}) or a Gaussian mean--variance cost (\texttt{Gaussian-MV}), Gaussian-kernel segmentation (\texttt{KCP}) \citep{arlot2019kernel}, and energy-based agglomeration (\texttt{E-Agglo}) \citep{james2015ecp}. Each detector is applied to the odd-indexed observations, and its estimated locations define disjoint midpoint intervals as in Example~\ref{exam:midpoint_pool}. For \texttt{POIS} and \texttt{POIS+}, the scoring function $v_k(\cdot)$ is an MMD scan statistic \citep{gretton2012kernel} using an average of Gaussian kernels with bandwidths determined from the odd-side data. \texttt{MOPS} applies the modified mirror statistic within the same intervals, using mean--variance score features. Detector settings are fixed across scenarios. Appendix~\ref{app:midpoint_detector_details} gives the detector settings and score definitions.

\begin{figure}[h]
\centering
\includegraphics[width=\textwidth]{
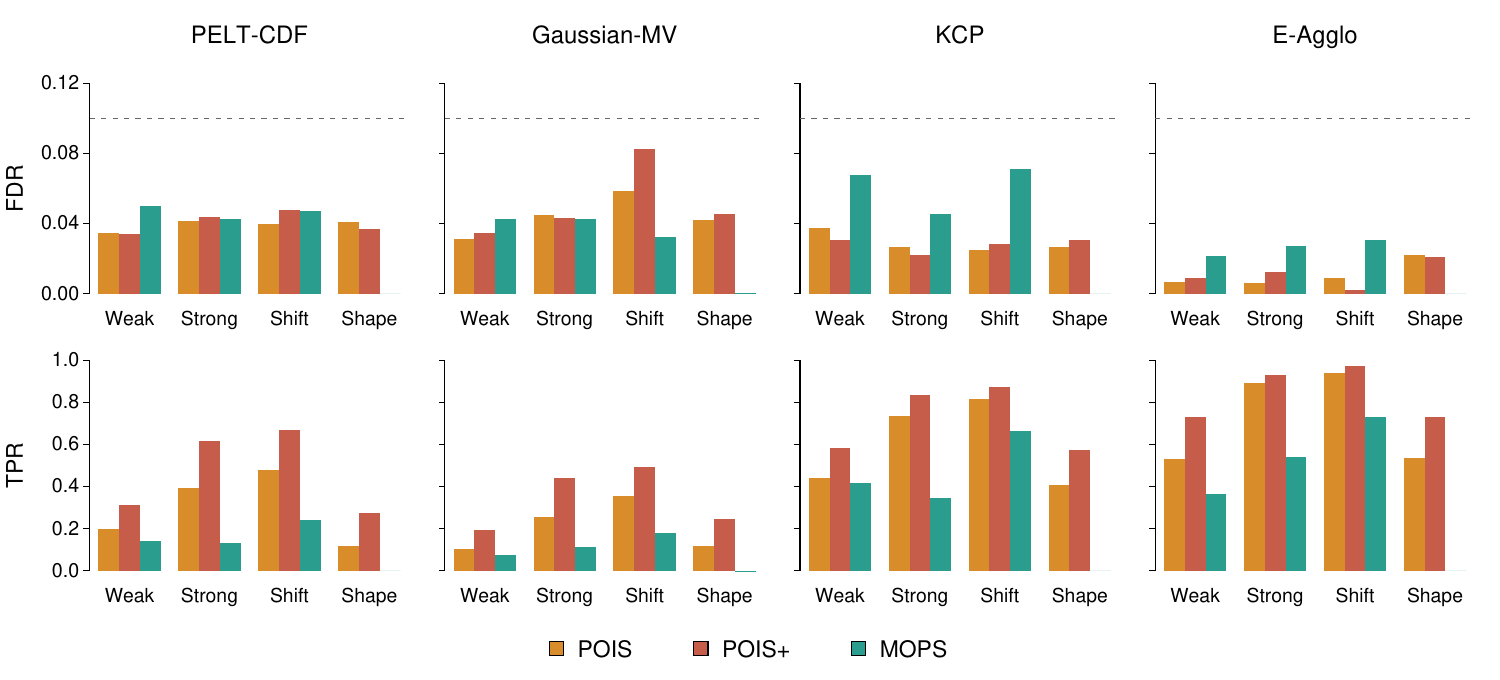}
\caption{Empirical interval-level FDR (top) and change-point coverage TPR (bottom) for \texttt{Weak}, \texttt{Strong}, \texttt{Shift}, and \texttt{Shape}. Columns correspond to \texttt{PELT-CDF}, \texttt{Gaussian-MV}, \texttt{KCP}, and \texttt{E-Agglo}. The dashed line marks \(\alpha=0.1\).}
\label{fig:sim_midpoint_detectors}
\end{figure}

Figure~\ref{fig:sim_midpoint_detectors} shows that \texttt{POIS} and \texttt{POIS+} keep empirical FDR below the nominal level for all four detectors.
\texttt{POIS+} attains the highest TPR in every scenario for each detector. Its gains over \texttt{POIS} are particularly pronounced for \texttt{PELT-CDF} and \texttt{Gaussian-MV}. 
For \texttt{KCP} and \texttt{E-Agglo}, 
both methods report smaller proportions of intervals containing no true change-point, so \texttt{POIS} already has high TPR after selection. 
\texttt{MOPS} also has low empirical FDR but generally covers fewer changes. Since \texttt{MOPS} is developed for parametric change-point models, the mean--variance scores used by \texttt{MOPS} provide no signal for the change, and its TPR is nearly zero in the \texttt{Shape} scenario. The results across the four detectors and four scenarios demonstrate that \texttt{POIS} and \texttt{POIS+} are algorithm-agnostic and have broader use. 

\subsection{Mean Changes in High-dimensional Data}
\label{subsec:sim_hd_mean_changes}

We next consider sparse mean changes under different noise distributions. 
We generate \(T\in\{2000,3000\}\) independent observations from \(X_t=\theta_t+\sigma\varepsilon_t\in\mathbb R^{50}\), with $K^*=10$  equally spaced change-points in the mean sequence $(\theta_t)_{1\le t \le T}$. 
Each mean jump affects ten randomly chosen coordinates independently. We consider four distributions of noise $\varepsilon_t$: Gaussian, Student \(t_4\), Chi-square \(\chi^2_3\), or Lognormal\((0,0.5^2)\). The noise coordinates are independent and standardized to mean zero and variance one. We define the signal-to-noise ratio (SNR) by \(\mathrm{SNR}=\|\theta_{\tau_k^*+1}-\theta_{\tau_k^*}\|_2 /\sigma\), and vary \(\sigma\) to obtain \(\mathrm{SNR}\in\{2.5,3,3.5\}\).

In this simulation, we use an INSPECT-based detector \citep{wang2018high} over a data-independent pool of seeded intervals. The detector reports the shortest interval exceeding the detection threshold and continues searching on either side of it, producing disjoint intervals indexed by $\widehat{\mathcal S}$.
For \texttt{POIS}, \texttt{POIS+}, and \texttt{MOPS}, the detector is run on the odd-indexed observations to obtain the reported intervals. For \texttt{POIS} and \texttt{POIS+}, the even-side score $v_k(\cdot)$ is a projected 20\% trimmed-mean contrast with the sparse direction estimated from the odd-side data and held fixed. 
\texttt{ART} uses the full sample detector and permutation calibration over the entire seeded pool. Appendix~\ref{app:hd_implementation} gives implementation details.

\begin{table}[h]
\centering
\begin{minipage}{\linewidth}
\setstretch{1}
\caption{Empirical interval-level FDR and change-point coverage TPR (both in percent) for high-dimensional mean changes with $T=2000$ and $T=3000$. Boldface marks the largest TPR within each setting.}
\label{tab:hd_mean_results}
\vspace{3pt}
\centering
\fontsize{10}{12}\selectfont
\setlength{\tabcolsep}{3pt}
\renewcommand{\arraystretch}{1}
\begin{tabular*}{\linewidth}{@{\extracolsep{\fill}}l*{12}{r}@{}}
\toprule
& \multicolumn{4}{c}{SNR $=2.5$} & \multicolumn{4}{c}{SNR $=3$} & \multicolumn{4}{c}{SNR $=3.5$} \\
\cmidrule(lr){2-5}\cmidrule(lr){6-9}\cmidrule(lr){10-13}
& \multicolumn{2}{c}{$T=2000$} & \multicolumn{2}{c}{$T=3000$} & \multicolumn{2}{c}{$T=2000$} & \multicolumn{2}{c}{$T=3000$} & \multicolumn{2}{c}{$T=2000$} & \multicolumn{2}{c}{$T=3000$} \\
\cmidrule(lr){2-3}\cmidrule(lr){4-5}\cmidrule(lr){6-7}\cmidrule(lr){8-9}\cmidrule(lr){10-11}\cmidrule(lr){12-13}
Method & FDR & TPR & FDR & TPR & FDR & TPR & FDR & TPR & FDR & TPR & FDR & TPR \\
\midrule
\multicolumn{13}{@{}c}{\textit{Gaussian}} \\[2pt]
\texttt{POIS} & 0.00 & 93.22 & 0.02 & 97.48 & 0.00 & \textbf{98.42} & 0.02 & 98.92 & 0.00 & 99.26 & 0.02 & 99.34 \\
\texttt{POIS+} & 0.00 & \textbf{93.28} & 0.02 & 97.58 & 0.00 & \textbf{98.42} & 0.02 & 98.94 & 0.00 & 99.28 & 0.02 & 99.36 \\
\texttt{MOPS} & 0.00 & 69.40 & 0.09 & \textbf{99.16} & 0.00 & 96.80 & 0.02 & \textbf{100.00} & 0.00 & \textbf{100.00} & 0.00 & \textbf{100.00} \\
\texttt{ART} & 0.00 & 18.32 & 0.00 & 16.12 & 0.00 & 23.18 & 0.00 & 19.60 & 0.00 & 28.70 & 0.00 & 22.98 \\
\midrule
\multicolumn{13}{@{}c}{\textit{Student $t_4$}} \\[2pt]
\texttt{POIS} & 3.34 & 53.50 & 4.18 & 52.52 & 2.34 & 74.48 & 3.20 & 71.22 & 1.62 & 85.30 & 2.34 & 83.76 \\
\texttt{POIS+} & 3.31 & \textbf{57.98} & 4.15 & \textbf{56.46} & 2.48 & \textbf{77.12} & 3.19 & \textbf{74.38} & 1.59 & \textbf{87.14} & 2.40 & \textbf{85.30} \\
\texttt{MOPS} & 5.24 & 25.98 & 6.44 & 29.18 & 5.69 & 51.86 & 7.36 & 52.14 & 5.27 & 74.28 & 7.10 & 72.02 \\
\texttt{ART} & 0.00 & 11.56 & 0.10 & 8.90 & 0.10 & 17.74 & 0.00 & 14.10 & 0.00 & 23.96 & 0.00 & 20.02 \\
\midrule
\multicolumn{13}{@{}c}{\textit{Chi-square $\chi^2_3$}} \\[2pt]
\texttt{POIS} & 1.96 & 64.26 & 3.80 & 61.28 & 1.41 & 81.00 & 2.59 & 79.96 & 0.79 & 89.58 & 1.76 & 88.36 \\
\texttt{POIS+} & 2.02 & \textbf{67.78} & 3.78 & \textbf{65.98} & 1.32 & \textbf{83.04} & 2.60 & \textbf{82.32} & 0.70 & \textbf{90.90} & 1.72 & \textbf{89.84} \\
\texttt{MOPS} & 5.44 & 46.76 & 5.95 & 35.80 & 4.76 & 72.06 & 6.20 & 63.76 & 3.45 & 86.80 & 5.12 & 80.22 \\
\texttt{ART} & 0.20 & 11.48 & 0.30 & 8.24 & 0.07 & 17.64 & 0.20 & 12.88 & 0.07 & 23.46 & 0.00 & 18.98 \\
\midrule
\multicolumn{13}{@{}c}{\textit{Lognormal}} \\[2pt]
\texttt{POIS} & 2.93 & 60.92 & 3.48 & 60.30 & 1.88 & 79.04 & 2.49 & 78.18 & 1.13 & 89.04 & 1.72 & 88.58 \\
\texttt{POIS+} & 2.72 & \textbf{64.80} & 3.23 & \textbf{64.46} & 1.81 & \textbf{81.58} & 2.44 & \textbf{81.00} & 1.01 & \textbf{90.34} & 1.66 & \textbf{90.04} \\
\texttt{MOPS} & 4.79 & 35.74 & 6.06 & 35.02 & 5.56 & 68.16 & 6.49 & 64.94 & 4.16 & 84.12 & 5.62 & 84.34 \\
\texttt{ART} & 0.00 & 10.84 & 0.20 & 8.58 & 0.00 & 17.38 & 0.07 & 14.14 & 0.05 & 23.54 & 0.00 & 19.64 \\
\bottomrule
\end{tabular*}
\end{minipage}
\end{table}

Table~\ref{tab:hd_mean_results} shows that \texttt{POIS} and \texttt{POIS+} keep empirical FDR below the nominal level across all settings. \texttt{MOPS} also maintains low empirical FDR, and \texttt{ART} yields very low FDR and covers far fewer changes. Under Gaussian noise, \texttt{POIS+} maintains the high TPR of \texttt{POIS}. Under the three non-Gaussian distributions, \texttt{POIS+} attains the highest TPR among the four methods, with a particularly large advantage over \texttt{MOPS} at low SNR. 
The improvement of \texttt{POIS+} over \texttt{POIS} is also more pronounced at low SNR under non-Gaussian noise. A weak change may be detected on the odd side but may not be selected after even-side testing. Since the two procedures use the same candidate intervals and even-side $p$-values, the higher TPR reflects the contribution of detection-conditional weighting. As SNR increases, \texttt{POIS} covers more true changes, leaving less room for improvement.

We also report average selected-interval length and runtime. Lengths are first averaged within each nonempty repetition
and then across such repetitions.
Runtime includes detection and the subsequent inferential calculations, averaged over all repetitions. Table~\ref{tab:hd_length_time_by_T} shows that \texttt{POIS} and \texttt{POIS+} have nearly identical average interval lengths. Thus, the higher TPR of \texttt{POIS+} is not accompanied by wider selected intervals on average. It is expected that \texttt{POIS+} takes additional time to compute the weights, and its runtime decreases as SNR increases.

\begin{table}[h]
\centering
\begin{minipage}{\linewidth}
\setstretch{1}
\caption{Average selected-interval length and runtime (seconds) for the settings in Table~\ref{tab:hd_mean_results}.}
\label{tab:hd_length_time_by_T}
\vspace{3pt}
\centering
\fontsize{10}{12}\selectfont
\setlength{\tabcolsep}{3pt}
\renewcommand{\arraystretch}{1}
\begin{tabular*}{\linewidth}{@{\extracolsep{\fill}}l*{12}{r}@{}}
\toprule
& \multicolumn{4}{c}{SNR $=2.5$} & \multicolumn{4}{c}{SNR $=3$} & \multicolumn{4}{c}{SNR $=3.5$} \\
\cmidrule(lr){2-5}\cmidrule(lr){6-9}\cmidrule(lr){10-13}
& \multicolumn{2}{c}{$T=2000$} & \multicolumn{2}{c}{$T=3000$} & \multicolumn{2}{c}{$T=2000$} & \multicolumn{2}{c}{$T=3000$} & \multicolumn{2}{c}{$T=2000$} & \multicolumn{2}{c}{$T=3000$} \\
\cmidrule(lr){2-3}\cmidrule(lr){4-5}\cmidrule(lr){6-7}\cmidrule(lr){8-9}\cmidrule(lr){10-11}\cmidrule(lr){12-13}
Method & Length & Time & Length & Time & Length & Time & Length & Time & Length & Time & Length & Time \\
\midrule
\multicolumn{13}{@{}c}{\textit{Gaussian}} \\[2pt]
\texttt{POIS} & 108.58 & 4.77 & 113.00 & 7.38 & 86.84 & 4.76 & 89.06 & 7.43 & 72.44 & 4.74 & 76.09 & 7.30 \\
\texttt{POIS+} & 108.57 & 27.96 & 112.98 & 33.70 & 86.84 & 21.38 & 89.05 & 25.02 & 72.43 & 16.72 & 76.09 & 20.15 \\
\texttt{MOPS} & 105.29 & 4.27 & 112.14 & 6.85 & 86.44 & 4.25 & 88.75 & 6.90 & 72.26 & 4.22 & 75.93 & 6.78 \\
\texttt{ART} & 63.99 & 8.83 & 71.97 & 14.40 & 56.11 & 8.81 & 60.32 & 14.48 & 52.10 & 8.80 & 52.29 & 14.31 \\
\midrule
\multicolumn{13}{@{}c}{\textit{Student $t_4$}} \\[2pt]
\texttt{POIS} & 65.26 & 4.98 & 70.32 & 7.78 & 59.33 & 4.87 & 62.84 & 7.59 & 55.39 & 4.82 & 57.83 & 7.42 \\
\texttt{POIS+} & 64.58 & 14.94 & 69.44 & 21.52 & 59.04 & 13.91 & 62.47 & 19.30 & 55.15 & 12.60 & 57.63 & 17.24 \\
\texttt{MOPS} & 60.18 & 4.26 & 65.93 & 6.87 & 57.70 & 4.21 & 60.52 & 6.81 & 54.29 & 4.20 & 57.25 & 6.73 \\
\texttt{ART} & 50.80 & 8.95 & 52.03 & 14.43 & 47.92 & 8.81 & 46.88 & 14.36 & 45.42 & 8.81 & 44.21 & 14.25 \\
\midrule
\multicolumn{13}{@{}c}{\textit{Chi-square $\chi^2_3$}} \\[2pt]
\texttt{POIS} & 66.40 & 4.97 & 71.61 & 7.75 & 59.90 & 4.81 & 63.45 & 7.52 & 55.65 & 4.80 & 58.01 & 7.48 \\
\texttt{POIS+} & 65.93 & 16.05 & 70.61 & 22.15 & 59.70 & 14.15 & 63.19 & 19.25 & 55.46 & 12.76 & 57.91 & 16.95 \\
\texttt{MOPS} & 63.16 & 4.31 & 67.46 & 6.94 & 58.40 & 4.21 & 62.41 & 6.81 & 54.76 & 4.24 & 58.29 & 6.82 \\
\texttt{ART} & 51.87 & 9.04 & 52.78 & 14.62 & 48.31 & 8.92 & 47.30 & 14.40 & 46.25 & 9.00 & 44.07 & 14.52 \\
\midrule
\multicolumn{13}{@{}c}{\textit{Lognormal}} \\[2pt]
\texttt{POIS} & 67.02 & 4.93 & 71.23 & 7.71 & 60.06 & 4.83 & 63.17 & 7.50 & 55.47 & 4.76 & 58.41 & 7.44 \\
\texttt{POIS+} & 66.45 & 15.83 & 70.39 & 22.21 & 59.84 & 14.06 & 62.78 & 19.47 & 55.34 & 12.57 & 58.27 & 17.22 \\
\texttt{MOPS} & 62.83 & 4.27 & 66.69 & 6.89 & 58.19 & 4.21 & 62.15 & 6.80 & 54.58 & 4.18 & 58.47 & 6.82 \\
\texttt{ART} & 52.19 & 8.82 & 52.26 & 14.41 & 48.55 & 8.89 & 47.36 & 14.31 & 46.22 & 8.78 & 44.31 & 14.43 \\
\bottomrule
\end{tabular*}
\end{minipage}
\end{table}

\subsection{Mean Changes in Profile Data}
\label{subsec:sim_profile_pois}

We finally consider the profile data, which refer to discretely observed functional data. We generate a sequence of \(T\in\{2000,3000\}\) profiles, each observed at 50 design points drawn independently from \(U(0,15)\), according to
\[
Y_{ij}=m_{s_i}(X_{ij})+Z_i r(X_{ij})+\sigma(X_{ij})\varepsilon_{ij},\qquad Z_i\sim N(0,1),
\]
where \(i=1,\ldots,T\) and \(j=1,\ldots,50\). Here, \(m_s\), \(r\), and \(\sigma\) specify the mean, random effect, and noise standard deviation, respectively.
The state \(s_i\in\{0,1\}\) alternates at $K^*=10$ randomly placed change-points separated by at least 75 profiles. The pointwise errors follow $N(0,1)$ or \(\sqrt{3/5}\,t_5\). We consider four scenarios. In \texttt{Mean}, the amplitude of the logistic mean curve alternates between 10 and 3. In \texttt{Anomaly}, one state adds a Gaussian-shaped bump of height 13.5 centered at \(x=4.5\), which persists throughout the corresponding segment. In \texttt{Adv-PC}, the mean shift is \(6.5\cos(2\pi x/15)\), while individual profiles contain the random component \(5.5Z_i\sin(2\pi x/15)\). In \texttt{Curvature}, the mean shift is a centered quadratic component, accompanied by random linear slopes across profiles. The random-effect and pointwise error distributions remain unchanged between states, so all four scenarios concern changes in the mean function. Given the change-point locations, profiles are independent, with \(Z_i\) shared across observations within profile \(i\). Appendix~\ref{app:profile_generation} specifies the functions \(m_s\), \(r\), and \(\sigma\).

In this simulation, we use a smooth profile decomposition (SPD) detector \citep{wu2025multiple} over a data-independent pool of fixed-width intervals. Each profile is represented by its fitted mean and variance functions at ten common grid points. The detector applies the local-maximum criteria in Section~\ref{sec:instantiations} to the scan statistics, producing disjoint intervals indexed by $\widehat{\mathcal S}$. For \texttt{POIS} and \texttt{POIS+}, the scoring function $v_k(\cdot)$ is a Wald-type statistic computed from the SPD summaries. We use a half-window of \(h=6\) profiles on each split, corresponding to a reporting radius of \(2h=12\) profiles on the original sequence. Permutations act on whole profiles.

Figure~\ref{fig:sim_profile} shows that \texttt{POIS} and \texttt{POIS+} keep empirical FDR below the nominal level under both Gaussian and standardized Student \(t_5\) errors. 
Both procedures attain higher TPR than \texttt{SPD-OSE} and \texttt{ART} across all scenarios. \texttt{SPD-OSE} covers many changes in \texttt{Mean}, but its TPR is substantially lower for the local or shape changes in the other scenarios. \texttt{ART}, which targets the more stringent FWER criterion, has very low empirical FDR and covers few changes. \texttt{POIS+} achieves higher TPR than \texttt{POIS} at both sequence lengths, with larger gains in \texttt{Adv-PC} and \texttt{Curvature}. In these scenarios, variation between profiles within the same segment can hide a mean change in a local window. The odd-side observations provide additional evidence from different realizations of this variation. In \texttt{Mean}, \texttt{POIS} already covers nearly all changes, leaving little room for improvement. The gains under Gaussian errors also show that odd-side weighting remains useful beyond heavy-tailed settings.

\begin{figure}[h]
\centering
\includegraphics[width=\textwidth]{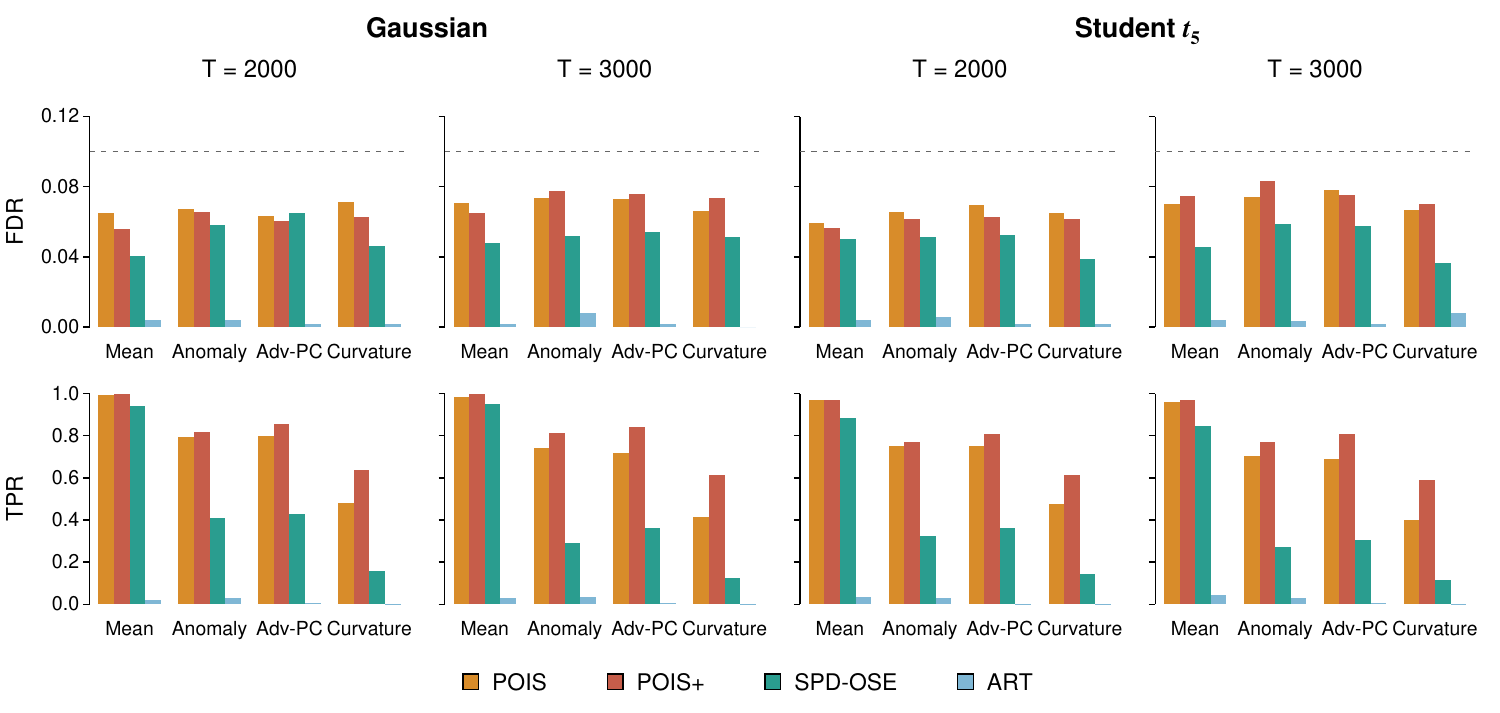}
\caption{Empirical interval-level FDR (top) and change-point coverage TPR (bottom) for the four profile scenarios with $T=2000$ and $T=3000$, under Gaussian and standardized Student \(t_5\) errors. The dashed line marks \(\alpha=0.1\).}
\label{fig:sim_profile}
\end{figure}

The improvement over \texttt{POIS} is larger for $T=3000$, although absolute TPR is generally lower than for $T=2000$. With $K^*$ and $h$ fixed, the longer sequence may produce a larger tested family without increasing the number of profiles in each even-side window. A larger tested family can make rejection more difficult, which helps explain the larger gain from odd-side weighting. Since the reported intervals have fixed width and are unchanged by weighting, the higher coverage of \texttt{POIS+} comes from selecting intervals that cover additional true change-points.

\section{Real Data Applications}
\label{sec:real_data}

We apply \texttt{POIS} and \texttt{POIS+} to stock returns and smartphone activity recordings, examining the selected intervals in relation to economic events and recorded activity transitions. 


\subsection{S\&P 500 market-stress signal}
\label{sec:real1_sp500}

We analyze 755 daily log returns for 170 S\&P 500 constituent stocks from January 2020 to December 2022. The market-stress signal in Figure~\ref{fig:real1_sp500} is the standardized logarithm of the daily cross-sectional mean square of centered, winsorized returns. All procedures use this univariate signal; implementation details are given in Appendix~\ref{app:realdata_implementation}. We use the narrowest-first construction in Section~\ref{sec:instantiations} with an odd-sample NOT detector \citep{baranowski2019narrowest}, which reports 24 disjoint candidate intervals. We compare the selected intervals with those from \texttt{NSP} \citep{fryzlewicz2024narrowest}, which uses a family-wise error criterion.

\begin{figure}[h]
    \centering
    \includegraphics[width=\textwidth]{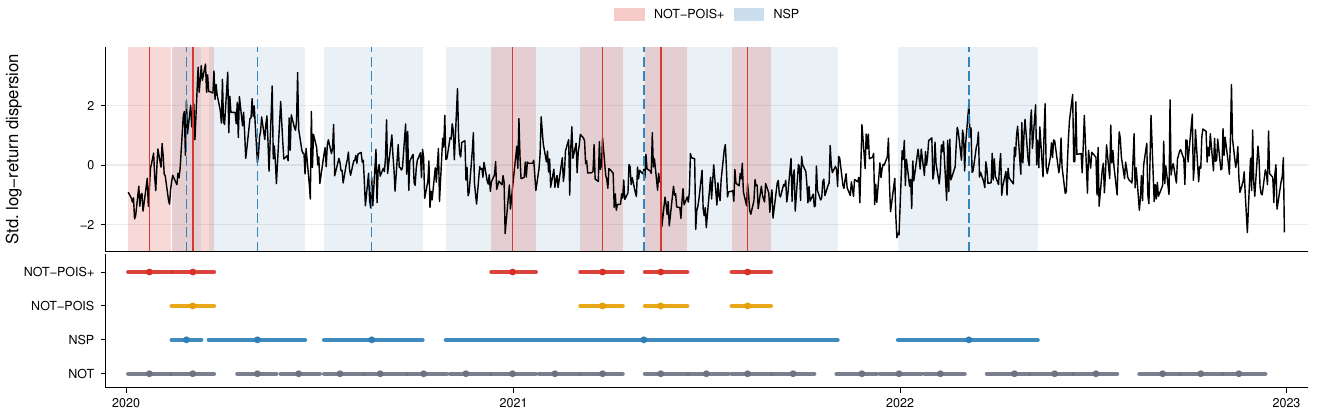}
    \caption{S\&P 500 market-stress signal (top) and reported intervals (bottom). The gray row shows detector-reported intervals. Dots mark representative locations.}
    \label{fig:real1_sp500}
\end{figure}

Figure~\ref{fig:real1_sp500} shows that \texttt{POIS+} selects six of the 24 candidate intervals, including all four selected by \texttt{POIS}. All candidates span 28 trading days, so weighting selects additional intervals without changing their width. \texttt{NSP} returns five intervals with a mean length of 99.8 trading days, several of which cover multiple shorter candidates.

The two additional intervals selected by \texttt{POIS+} cover the declaration of the COVID-19 international public-health emergency in early 2020 \citep{who2020emergency} and the first U.S. emergency authorization of a COVID-19 vaccine in December 2020 \citep{fda2020firstvaccine}. The remaining intervals cover emergency monetary easing \citep{fed2020marchstatement}, strengthening economic activity and rising inflation \citep{fed2021marchstatement,bls2021aprilcpi}, and discussions of reducing asset purchases \citep{powell2021outlook}. These events were considered after selection to help interpret changes in return variability.

\subsection{Human Activity and Postural Transitions}
\label{sec:real_har}

We analyze the second recording from each of 30 participants in the Human Activities and Postural Transitions (HAPT) data \citep{reyes2015hapt}. Each recording contains triaxial acceleration and angular velocity measured at 50 Hz by a waist-mounted smartphone. We retain the period between the first and last activity annotations, including internal gaps. Reference change ranges span the end of one sustained activity and the start of the next differently labeled activity. Each nonoverlapping 0.5-second block is represented by 73 features describing posture, movement, and temporal variation. The feature representation and filter coefficients are fixed from development on the first recordings of subjects 1--6. The second recordings were examined during earlier exploratory work. We apply the local-window construction in Section~\ref{sec:instantiations} using fixed-center $L_2$ contrasts and nominal 10-second intervals. We compare the selected intervals with the full detector, \texttt{MOPS}, and \texttt{ART}. Preprocessing and implementation details are given in Appendix~\ref{app:realdata_implementation}.

\begin{table}[t]
\centering
\caption{HAPT results averaged over the second recordings of 30 participants and three permutation seeds. Coverage is the proportion of reference change ranges intersected by a reported interval. Precision and $F_1$ use maximum-cardinality one-to-one interval matching. Full detector and odd candidates show outputs before interval selection.}
\label{tab:har_results}
\small
\begin{tabular}{@{}lrrrr@{}}
\toprule
Method & Mean count & Coverage (\%) & Precision (\%) & $F_1$ \\
\midrule
Full detector & 16.53 & 98.06 & 69.89 & 0.8151 \\
Odd candidates & 16.33 & 96.67 & 69.68 & 0.8089 \\
\texttt{POIS} & 14.01 & 89.54 & 75.59 & 0.8140 \\
\texttt{POIS+} & 14.79 & 93.61 & 74.67 & 0.8283 \\
\texttt{MOPS} & 11.67 & 71.67 & 74.11 & 0.6313 \\
\texttt{ART} & 6.52 & 46.80 & 85.22 & 0.5950 \\
\bottomrule
\end{tabular}
\end{table}

Table~\ref{tab:har_results} shows that \texttt{POIS+} reports fewer intervals than the full detector while maintaining high coverage and improving matching precision and $F_1$. Compared with \texttt{POIS}, it covers more activity changes, with slightly lower precision but higher $F_1$. \texttt{MOPS} has lower coverage and $F_1$, whereas \texttt{ART} returns the smallest collection, with higher precision but substantially lower coverage.

\begin{figure}[h]
\centering
\includegraphics[width=\linewidth]{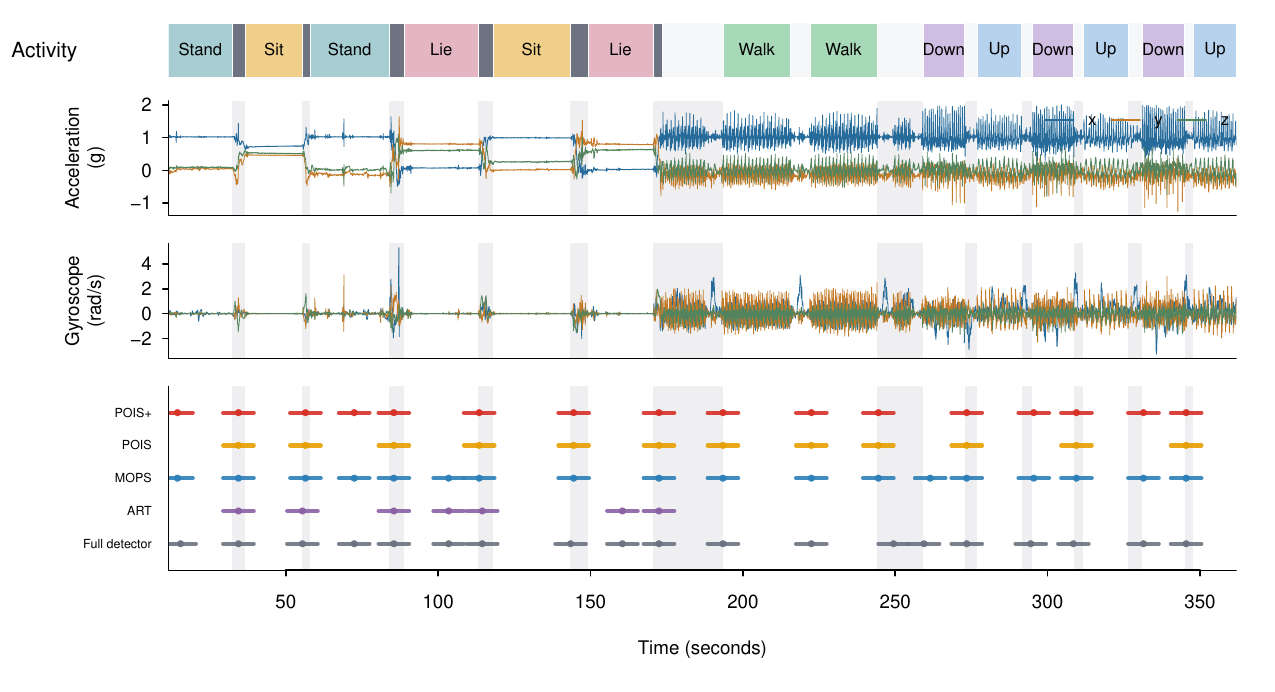}
\caption{HAPT subject 30: activity labels, raw signals, and reported intervals. Gray bands mark reference changes; dots mark fixed test splits.}
\label{fig:har_activity}
\end{figure}

Figure~\ref{fig:har_activity} shows the results for subject 30, where \texttt{POIS} misses two switches from ascending to descending stairs. Distinguishing these activities can be harder than distinguishing static postures. \texttt{POIS+} covers all 12 reference changes, with four unmatched intervals under one-to-one matching. \texttt{MOPS} also covers all changes but selects two more intervals than \texttt{POIS+}. All 12 reference changes remain covered by \texttt{POIS+} across three permutation seeds. This recording was chosen to illustrate how odd-side weighting recovers changes missed by \texttt{POIS}, while Table~\ref{tab:har_results} summarizes performance over all recordings. Appendix Figure~\ref{fig:har_additional} shows an additional recording where \texttt{MOPS} selects no intervals.

\section{Concluding Remarks}
\label{sec:conclusion}

This paper develops a unified framework to select detector-reported intervals of multiple change-points.
The proposed \texttt{POIS} separates detection from permutation-based tests
through order-preserving sample splitting. \texttt{POIS+} constructs detection-conditional $e$-values that allow the same data to select the tested family and generate valid weights. Both procedures provide finite-sample interval-level FDR control without specifying the algorithm and the distributions. We also provide implementations of \texttt{POIS+} for two typical classes of detectors. Under the stated conditions, \texttt{POIS} attains asymptotically full expected coverage of true change-points, while \texttt{POIS+} has asymptotically no lower expected coverage.


Future work includes extending the permutation testing step to serially dependent or locally stationary observations. Another direction is to allow reported intervals whose endpoints are obtained through refitting or adaptive post-processing, beyond detection using a fixed interval pool.

\phantomsection\label{supplementary-material}
\bigskip

\begin{center}

{\large\bf SUPPLEMENTARY MATERIAL}

\end{center}

The supplementary material contains additional methods, proofs, verification of the score conditions, and further numerical details.


\setstretch{.8}
\bibliography{references}

\newpage
\appendix
\begin{bibunit}[asa]
\setbibunitprefix{suppl-}
\setcounter{table}{0}
\renewcommand{\thetable}{S\arabic{table}}
\setcounter{figure}{0}
\renewcommand{\thefigure}{S\arabic{figure}}
\counterwithout{lemma}{section}
\setcounter{lemma}{0}
\renewcommand{\thelemma}{S\arabic{lemma}}
\providecommand{\theHlemma}{}
\renewcommand{\theHlemma}{appendix.\arabic{lemma}}
\counterwithout{theorem}{section}
\setcounter{theorem}{0}
\renewcommand{\thetheorem}{S\arabic{theorem}}
\providecommand{\theHtheorem}{}
\renewcommand{\theHtheorem}{appendix.\arabic{theorem}}

\numberwithin{equation}{section}
\allowdisplaybreaks
\spacingset{1.7}

\begin{center}
    {\LARGE\bf Supplementary Material for ``Interval Localization for Multiple Change-points with Error Rate Control''}
\end{center}

\section{Additional methods}

For the proofs below, \textup{(A1)} and \textup{(A2)} refer to the first and second parts of Assumption~\ref{ass:power_main}, and \textup{(B1)} and \textup{(B2)} to the first and second parts of Assumption~\ref{ass:odd_nondeflation}, respectively.
BH and weighted BH return the empty set if the reported family is empty or no rank meets its critical value. A zero odd-side weight gives $P_k^+=1$, and the weight in \eqref{eq:post_selection_Evalue} is zero when its denominator is zero. An empty accepted set gives weight one when the observed transformed score is positive. Throughout, $0<\alpha<1$, $\mathrm{FDR}$ denotes the interval-level false discovery rate, and $a_+=a\vee0$.

For a randomized interval pool or detector, we condition on its data-independent randomization and average over it at the end. Scores and local permutation draws are constructed as in the main text.

In arguments using absolute orderings, repeated observations may be
distinguished by independent continuous auxiliary labels, permuted with
the observations and ignored by the procedure. Conditional on the
unordered labeled sample, the observed ordering is uniform. Averaging
over the auxiliary labels recovers each stated conditional bound.

\subsection{Graph-based aggregation for overlapping reported intervals}
\label{app:overlap_graph}

We first extend \(p\)-value aggregation to overlapping reported intervals.
Throughout this subsection, \(X_1,\ldots,X_T\) follow the data model in the main text. The
inferential units remain the original reported intervals, with
\(H_{0k}:I_k\cap\mathcal T^*=\varnothing\).

Let \(\mathcal F_+^{\mathcal O}\) be the sigma-field generated by the
odd-side data and all odd-side randomization.  Conditional on
\(\mathcal F_+^{\mathcal O}\), the reported family
\(\widehat{\mathcal S}\), the reported intervals, the scoring functions, and the
weights \(\{W_k^{\mathcal O}:k\in\widehat{\mathcal S}\}\) are fixed.  Write
\(M=|\widehat{\mathcal S}|\), and return the empty set when \(M=0\).

For each \(k\in\widehat{\mathcal S}\), let \(U_k\) collect the permutation
draws used in \eqref{eq:perm-p-split}, which are independent of the data
and across intervals. By construction,
\begin{equation}
    \sigma(P_k^{\mathcal E})
    \subseteq
    \sigma\left(
        \mathcal F_+^{\mathcal O},
        X_{I_k^{\mathcal E}},
        U_k
    \right).
    \label{eq:graph_locality}
\end{equation}

Define the overlap graph \(\mathbb D_{\widehat{\mathcal S}}\) on
\(\widehat{\mathcal S}\) by placing an edge between distinct \(k\) and
\(\ell\) whenever \(I_k^{\mathcal E}\cap I_\ell^{\mathcal E}\ne\varnothing\).
Let \(N_k\) denote the closed neighborhood of \(k\), including \(k\) itself.

For a graph \(\mathbb D\) on a finite node set \(\mathcal M\), let
\(\operatorname{Ind}(\mathbb D)\) be its collection of independent sets.
For \(A\subseteq\mathcal M\), let \(p^A_i=p_i\) for \(i\in A\) and
\(p^A_i=1\) otherwise.  The Independent Set BH procedure of
\citet{nguyen2025graphical} is
\begin{equation}
    \operatorname{IndBH}_{\mathbb D,\alpha}(p)
    =
    \bigcup_{A\in\operatorname{Ind}(\mathbb D)}
    \mathcal R_{\alpha}^{\rm BH}(p^A),
    \label{eq:indbh}
\end{equation}
where every BH procedure uses all \(|\mathcal M|\) coordinates and hence the
critical values \(\alpha r/|\mathcal M|\).  The union in
\eqref{eq:indbh} may equivalently be restricted to maximal independent sets.

The following conditional IndBH bound allows fixed, unnormalized weights.

\begin{lemma}
\label{lem:weighted_indbh}
Let \(\mathcal G\) be a sigma-field such that, conditional on \(\mathcal G\),
a node set \(\mathcal M\), a graph \(\mathbb D\), and nonnegative weights
\(\{w_i:i\in\mathcal M\}\) are fixed.  Write \(M=|\mathcal M|\), assume \(M\ge1\), and let
\(N_i\) be the closed neighborhood of \(i\).  Suppose that every null
\(i\in\mathcal M\cap\mathcal H_0\) satisfies
\begin{equation}
    \mathbb P\{P_i\le u\mid\mathcal G\}\le u,
    \qquad
    P_i\ \perp\!\!\!\perp\
    (P_j)_{j\in\mathcal M\setminus N_i}
    \mid\mathcal G.
    \label{eq:weighted_indbh_conditions}
\end{equation}
Set \(Q_i=(P_i/w_i)\wedge1\), with \(Q_i=1\) when \(w_i=0\), and let
\(\mathcal R=\operatorname{IndBH}_{\mathbb D,\alpha}(Q)\).  Then
\begin{equation}
    \mathbb E\left[
        \left.
        \frac{|\mathcal R\cap\mathcal H_0|}
             {|\mathcal R|\vee1}
        \right|
        \mathcal G
    \right]
    \le
    \frac{\alpha}{M}
    \sum_{i\in\mathcal M\cap\mathcal H_0}w_i.
    \label{eq:weighted_indbh_bound}
\end{equation}
\end{lemma}

\begin{proof}
Condition on $\mathcal G$ and fix $i\in\mathcal M\cap\mathcal H_0$.
We bound the contribution of this index to the false discovery proportion
before summing over the nulls.  For each
\(A\in\operatorname{Ind}(\mathbb D)\) containing \(i\), let \(r_{i,A}\) be
the number of BH rejections from \(Q^A\) after replacing \(Q_i\) by zero, and
set \(r_i=\max_{A\ni i}r_{i,A}\).  Since \(A\setminus\{i\}\subseteq
\mathcal M\setminus N_i\), the variable \(r_i\) is measurable with respect
to \(\mathcal G\) and \((P_j)_{j\in\mathcal M\setminus N_i}\).  Moreover,
\(r_i\ge1\).

For each fixed $A$ containing $i$, lowering $Q_i$ to zero can only
increase the BH rejection count. If $i$ was rejected, restoring $Q_i$
leaves the modified rejection set feasible, because its critical value
is at least the original one. Conversely, if
$Q_i\le\alpha r_{i,A}/M$, all $r_{i,A}$ modified rejections remain
feasible after restoration. Thus $i$ is rejected in $A$ exactly when
$Q_i\le\alpha r_{i,A}/M$, and the rejection count then equals
$r_{i,A}$. Maximizing over $A$ gives \(i\in\mathcal R \Longleftrightarrow Q_i\le\frac{\alpha r_i}{M}\).
On this event, an independent set attaining $r_i$ contributes $r_i$
rejections to the union, so $|\mathcal R|\ge r_i$. Since
\(\alpha r_i/M<1\), the rejection event is empty when \(w_i=0\), and
otherwise equals
\(\{P_i\le\alpha w_i r_i/M\}\).  Hence, using
\eqref{eq:weighted_indbh_conditions},
\begin{equation*}
\mathbb E\left[
    \left.
    \frac{\mathbbm 1\{i\in\mathcal R\}}
         {|\mathcal R|\vee1}
    \right|
    \mathcal G
\right]
\le
\mathbb E\left[
    \left.
    \frac{\mathbbm 1\{P_i\le\alpha w_i r_i/M\}}{r_i}
    \right|
    \mathcal G
\right]
\le
\frac{\alpha w_i}{M}.
\end{equation*}
For the last inequality, condition further on $r_i$. It is a function
of the non-neighboring inputs, so \eqref{eq:weighted_indbh_conditions}
bounds the conditional probability of $P_i\le\alpha w_i r_i/M$ by
$\min\{\alpha w_i r_i/M,1\}$. Dividing by $r_i$ gives at most
$\alpha w_i/M$. If $w_i=0$, its rejection contribution is zero.
Summing over $i\in\mathcal M\cap\mathcal H_0$ proves
\eqref{eq:weighted_indbh_bound}.
\end{proof}

For POIS, define
\begin{equation}
    \widehat{\mathcal R}_{\rm G}
    =
    \operatorname{IndBH}_{\mathbb D_{\widehat{\mathcal S}},\alpha}
    \left(
        (P_k^{\mathcal E})_{k\in\widehat{\mathcal S}}
    \right).
    \label{eq:graph_pois}
\end{equation}
For POIS+, use the weighted \(p\)-values
\(P_k^+=(P_k^{\mathcal E}/W_k^{\mathcal O})\wedge1\), with \(P_k^+=1\)
when \(W_k^{\mathcal O}=0\), and define \(\widehat{\mathcal R}_{\rm G}^+ = \operatorname{IndBH}_{\mathbb D_{\widehat{\mathcal S}},\alpha} \left( (P_k^+)_{k\in\widehat{\mathcal S}} \right)\).
\begin{theorem}
\label{thm:graph_overlap_fdr}
Under the data model in the main text, the POIS rejection set in
\eqref{eq:graph_pois} satisfies
\begin{equation}
    \mathrm{FDR}
    (\widehat{\mathcal R}_{\rm G})
    \le
    \alpha
    \mathbb E\left[
        \frac{|\widehat{\mathcal S}\cap\mathcal H_0|}
             {|\widehat{\mathcal S}|\vee1}
    \right]
    \le\alpha.
    \label{eq:graph_pois_fdr}
\end{equation}

For the weights in \eqref{eq:post_selection_Evalue},
Theorem~\ref{thm:perm-e-validity} gives, for every \(k\in\mathcal H_0\),
\begin{equation}
    \mathbb E_{H_{0k}}
    \left[
        W_k^{\mathcal O}
        \,\middle|\,
        k\in\widehat{\mathcal S},\Psi_k
    \right]
    \le1,
    \label{eq:graph_weight_validity}
\end{equation}
and \eqref{eq:cond_statistic} makes \(|\widehat{\mathcal S}|\) measurable
with respect to \(\Psi_k\) on \(\{k\in\widehat{\mathcal S}\}\). Then
\begin{equation}
    \mathrm{FDR}
    (\widehat{\mathcal R}_{\rm G}^+)
    \le
    \alpha
    \mathbb E\left[
        \frac{
        \sum_{k\in\widehat{\mathcal S}\cap\mathcal H_0}
        W_k^{\mathcal O}}
        {|\widehat{\mathcal S}|\vee1}
    \right]
    \le
    \alpha
    \mathbb E\left[
        \frac{|\widehat{\mathcal S}\cap\mathcal H_0|}
             {|\widehat{\mathcal S}|\vee1}
    \right]
    \le\alpha.
    \label{eq:graph_pois_plus_fdr}
\end{equation}
Neither conclusion requires the reported intervals to be disjoint.
\end{theorem}

\begin{proof}
If $M=0$, both rejection sets are empty. Otherwise, condition on
$\mathcal F_+^{\mathcal O}$. We verify the two requirements in
\eqref{eq:weighted_indbh_conditions} for the actual even-side scores,
and then apply Lemma~\ref{lem:weighted_indbh}. For every
\(k\in\widehat{\mathcal S}\), \(I_k^{\mathcal E} \cap \bigcup_{\ell\in\widehat{\mathcal S}\setminus N_k} I_\ell^{\mathcal E} = \varnothing\).
Independence of the observations, independence of the variables \(U_k\),
and \eqref{eq:graph_locality} therefore imply \(P_k^{\mathcal E} \ \perp\!\!\!\perp\ (P_\ell^{\mathcal E})_{\ell\in \widehat{\mathcal S}\setminus N_k} \mid\mathcal F_+^{\mathcal O}\).
If \(k\in\mathcal H_0\), the observations in \(I_k^{\mathcal E}\) are
exchangeable conditional on \(\mathcal F_+^{\mathcal O}\), so the
fixed-interval permutation argument gives \(\mathbb P_{H_{0k}} \{P_k^{\mathcal E}\le u\mid\mathcal F_+^{\mathcal O}\} \le u\).
Applying Lemma~\ref{lem:weighted_indbh} with \(w_k=1\) yields the conditional
version of \eqref{eq:graph_pois_fdr}.  Taking expectations proves
\eqref{eq:graph_pois_fdr}.

For POIS+, apply the same lemma with \(w_k=W_k^{\mathcal O}\).  This gives \(\mathrm{FDR} (\widehat{\mathcal R}_{\rm G}^+) \le \alpha \mathbb E\left[ \frac{ \sum_{k\in\widehat{\mathcal S}\cap\mathcal H_0} W_k^{\mathcal O}} {|\widehat{\mathcal S}|\vee1} \right]\).
By \eqref{eq:graph_weight_validity} and the measurability of
\(|\widehat{\mathcal S}|\), each \(k\in\mathcal H_0\) satisfies
\begin{equation*}
\mathbb E\left[
    \frac{
    \mathbbm 1\{k\in\widehat{\mathcal S}\}W_k^{\mathcal O}}
    {|\widehat{\mathcal S}|\vee1}
\right]
\le
\mathbb E\left[
    \frac{\mathbbm 1\{k\in\widehat{\mathcal S}\}}
         {|\widehat{\mathcal S}|\vee1}
\right].
\end{equation*}
Summing over \(k\in\mathcal H_0\) proves
\eqref{eq:graph_pois_plus_fdr}.
\end{proof}

\begin{remark}[Adjacent overlap]
\label{rem:adjacent_overlap}
If the reported intervals can be ordered and relabeled so that
\(I_k^{\mathcal E}\cap I_\ell^{\mathcal E}=\varnothing\) whenever
\(|k-\ell|>1\), then \(\mathbb D_{\widehat{\mathcal S}}\) is a subgraph of a
path.  More generally, arbitrary overlaps produce an interval graph, and
Theorem~\ref{thm:graph_overlap_fdr} remains applicable.
\end{remark}

The next subsection constructs detection-conditional \(e\)-values from the
full sample for overlapping reported intervals.

\subsection{Full-sample detection-conditional e-BH}
\label{app:overlap}

We apply the detection-conditional construction to the full sample and
use e-BH to control FDR under arbitrary dependence among the corresponding
$e$-values \citep{wang2022false}. The inferential units are
the original reported intervals, with
\(H_{0k}:I_k\cap\mathcal T^*=\varnothing\).

Write \(M=|\widehat{\mathcal S}|\), and order the nonnegative inputs as
\(E_{[1]}\ge\cdots\ge E_{[M]}\). The e-BH rule rejects the hypotheses
corresponding to the largest \(\widehat r_e\) inputs, where \(\widehat r_e=\max\left\{1\le r\le M: E_{[r]}\ge\frac{M}{\alpha r}\right\}\).
The maximum is zero when the set is empty, including when \(M=0\).

\begin{theorem}
\label{thm:selected_family_ebh}
For the e-BH rejection set \(\widehat{\mathcal R}_e\),
\begin{equation}
    \mathrm{FDR}(\widehat{\mathcal R}_e)
    \le\alpha\,\mathbb E\left[
    \frac{\sum_{k\in\widehat{\mathcal S}\cap\mathcal H_0}E_k}{M\vee1}
    \right].
    \label{eq:selected_family_ebh_bound}
\end{equation}
In particular, suppose that every \(k\in\mathcal H_0\) satisfies
\begin{equation}
    \mathbb E_{H_{0k}}[E_k\mid k\in\widehat{\mathcal S},\Psi_k]\le1,
    \label{eq:selected_family_conditional_e}
\end{equation}
and that \(M\) is measurable with respect to \(\Psi_k\) on
\(\{k\in\widehat{\mathcal S}\}\). Then
\(\mathrm{FDR}(\widehat{\mathcal R}_e)\le\alpha\).
\end{theorem}

\begin{proof}
Let \(R_e=|\widehat{\mathcal R}_e|\), and extend \(E_k\) by zero outside
the reported family. Each rejected input satisfies
\(E_k\ge M/(\alpha R_e)\). Hence
\begin{equation*}
\frac{\mathbbm 1\{k\in\widehat{\mathcal R}_e\}}{R_e\vee1}
    \le\frac{\alpha\,\mathbbm 1\{k\in\widehat{\mathcal S}\}E_k}{M\vee1}.
\end{equation*}
Summing over null indices and taking expectations proves
\eqref{eq:selected_family_ebh_bound}. For the second claim, condition each
null contribution from a reported interval on \(\Psi_k\). Its denominator is fixed under
this conditioning, so \eqref{eq:selected_family_conditional_e} bounds its
expectation by that of \(\mathbbm 1\{k\in\widehat{\mathcal S}\}/(M\vee1)\).
Summing gives the bound
\(\alpha\,\mathbb E[|\widehat{\mathcal S}\cap\mathcal H_0|/(M\vee1)]\le\alpha\).
\end{proof}

Run the detector on \(X_{1:T}\), using the fixed candidate pool
\(\mathcal I\). For \(k\in\widehat{\mathcal S}\), let \(\Pi_k\) be the
permutation group on \(I_k\), and write
\(\widehat{\mathcal S}^{\pi,k}=\mathbb A(X^{\pi,k};\mathcal I)\).
To preserve the target interval and the number of reported intervals, use
\begin{equation*}
\widetilde\Pi_k=\{\pi\in\Pi_k:
    k\in\widehat{\mathcal S}^{\pi,k},\
    |\widehat{\mathcal S}^{\pi,k}|=M\}.
\end{equation*}
With a nonnegative transformation \(g\), the full-sample conditional e-value is
\begin{equation}
    E_k^{\rm cond}=
    \frac{|\widetilde\Pi_k|\,g\{v_k(X_{I_k})\}}
    {\sum_{\pi\in\widetilde\Pi_k}g\{v_k(X^{\pi,k}_{I_k})\}}.
    \label{eq:overlap_conditional_evalue}
\end{equation}
The identity permutation is included once, and a zero denominator gives zero.
Preserving an additional detection feature, such as the reported split,
the whole reported set, or the detector path to \(I_k\), corresponds to
conditioning on that feature as well.

To verify conditional validity, fix a null index \(k\), the unordered
labeled observations inside \(I_k\), and the ordered observations outside
\(I_k\). Under \(H_{0k}\), the observations in \(I_k\) are i.i.d.
by the model in the main text. Hence the observed absolute ordering
\(\rho_0\) is uniform over \(\Pi_k\) under this conditioning. Write \(x^\rho\) for the data in absolute ordering \(\rho\), and
\(\widehat{\mathcal S}(\rho)=\mathbb A(x^\rho;\mathcal I)\). For a
nonnegative integer \(m\), define the set of orderings
\begin{equation*}
\mathfrak F_{k,m}=\{\rho\in\Pi_k:
    k\in\widehat{\mathcal S}(\rho),\
    |\widehat{\mathcal S}(\rho)\setminus\{k\}|=m\}.
\end{equation*}
Conditional on \(k\in\widehat{\mathcal S}\) and \(M=m+1\),
\(\rho_0\) is uniform over \(\mathfrak F_{k,m}\).
Composition with \(\rho_0\) maps the accepted relative permutations
bijectively onto \(\mathfrak F_{k,m}\). Thus the denominator in
\eqref{eq:overlap_conditional_evalue} is constant on this set, and
averaging its numerator over the set gives conditional mean one when
the denominator is positive, and zero otherwise. The same argument
applies after conditioning on additional exact detection features.
Since \(M=m+1\) on the detection event, Theorem~\ref{thm:selected_family_ebh}
proves FDR control without splitting or disjointness.

For a finite Monte Carlo implementation, draw \(B'\) independent uniform
proposals from \(\Pi_k\). Let \(\mathcal A_k^{\rm full}\) be the indices of
the proposals in \(\widetilde\Pi_k\), and let
\(N_k=|\mathcal A_k^{\rm full}|\). With
\(V_k=v_k(X_{I_k})\) and \(V_{k,b}=v_k(X^{\pi_{k,b},k}_{I_k})\), use
\begin{equation}
    \widehat E_k^{\rm cond}=
    \frac{(N_k+1)g(V_k)}
    {g(V_k)+\sum_{b\in\mathcal A_k^{\rm full}}g(V_{k,b})},
    \label{eq:full_conditional_mc_evalue}
\end{equation}
again setting a zero-denominator ratio to zero. Here \(B'\) counts
proposals, and \(N_k\) is random. Conditional on the set \(\mathfrak F_{k,m}\),
the acceptance probability is the same for every observed ordering.
Conditional further on the accepted proposal indices, the observed
ordering and the accepted orderings are independent uniform draws from
this set. Their transformed scores are therefore exchangeable, and
averaging the \(N_k+1\) numerator terms gives conditional mean at most one.
Theorem~\ref{thm:selected_family_ebh} consequently also applies to
\(\widehat E_k^{\rm cond}\).

This implementation satisfies \(\widehat E_k^{\rm cond}\le N_k+1\).
Every e-BH rejection requires an e-value of at least \(1/\alpha\), and
the threshold is higher when only a small fraction of the family is
rejected. Thus, with a fixed number of permutation draws, a low detection-preserving
acceptance rate can limit power even when the unconditioned local score
is large.

\subsection{Simulation comparison of $e$-BH and \texttt{POIS+}}
\label{app:epois-signal-sensitivity}

We compare the two graph-based procedures in Appendix~\ref{app:overlap_graph}
with full-sample conditional e-BH from Appendix~\ref{app:overlap}, and include
the full-sample detector output as a reference. We consider univariate
weak- and strong-signal settings, with 12 equally spaced changes. For mean
changes, each segment has length 35 and
\(X_t=\mu_\ell+1.25\varepsilon_t\) on segment \(\ell\), where
\begin{equation*}
(\mu_1,\ldots,\mu_{13})=
    \frac{c_\mu}{1.25}(0,3.2,-2.2,2.7,-3.1,2.3,-2.6,3.4,-1.7,2.9,-2.4,3.1,-2.0).
\end{equation*}
We use interval-removal NOT on a seeded interval pool with geometric
length ratio \(1/\sqrt2\) and minimum lengths 18 and 20 for the odd-side
and full-sample sequences, respectively. On each interval \([a,b]\), we
maximize the absolute CUSUM over splits \(u=a+8,\ldots,b-8\).
The threshold is the empirical 0.5th percentile of these interval scores,
or their minimum when the pool has at most 120 intervals, and is
recomputed for each permutation. For variance changes, each segment has length 40,
\(X_t=\sigma_\ell\varepsilon_t\), and
\begin{equation*}
(\log\sigma_1,\ldots,\log\sigma_{13})=
    c_\sigma(0,1.0,-0.6,1.1,-0.7,0.9,-0.8,1.2,-0.5,0.95,-0.7,1.1,-0.6).
\end{equation*}
We apply SaRa to \(\log(X_t^2+10^{-8})\) with half-window length \(h=8\),
using its original local windows, which may overlap. The weak setting uses \(c_\mu=c_\sigma=1\), whereas
the strong setting uses \(c_\mu=2\) and \(c_\sigma=1.7\).

The independent innovations follow \(N(0,1)\), \(t_4/\sqrt2\),
\({\rm Laplace}(0,1/\sqrt2)\), or
\(\{(1-\xi)Z+3.5\xi Z'\}/\sqrt{0.94+0.06(3.5)^2}\), where
\(Z,Z'\) are independent standard normal variables and
\(\xi\sim{\rm Bernoulli}(0.06)\) is independent of them.
Scores and conditional weights use the absolute CUSUM
at the detected split,
\begin{equation*}
 \sqrt{\frac{n_1n_2}{n_1+n_2}}\,|\bar X_1-\bar X_2|,
\end{equation*}
where \(n_1,n_2\) and \(\bar X_1,\bar X_2\) are the group sizes and
means on either side of the split. For variance changes, the means are
computed from the log-transformed observations.
Each configuration has 200 independent repetitions. All methods use
\(\alpha=0.1\). Both conditional constructions use the main-text defaults
\(B=B'=499\) and \(g(v)=v^{12}\). Both detectors report at most 30 intervals, without further
model-selection pruning. For each conditional
weight or e-value, a proposal is accepted only if the target interval,
its reported split location, and the number of reported intervals are preserved.
Graph-POIS and Graph-POIS+ use odd-side detection and even-side testing.
Full-sample conditional e-BH uses the unsplit sequence and
\eqref{eq:full_conditional_mc_evalue}. We use the e-BH rule.

Figures~\ref{fig:univar-weak-epois} and~\ref{fig:univar-strong-epois}
show that \texttt{Graph-POIS} and \texttt{Graph-POIS+} keep empirical FDR
below the nominal level under all four noise distributions and both signal
strengths. The full-sample detector covers most changes but reports many
intervals containing no true change-point, resulting in FDR above the
nominal level.

For mean changes, the two graph-based procedures have similar TPRs, and
both cover more true changes under stronger signals. For variance changes,
\texttt{Graph-POIS+} attains higher TPR than \texttt{Graph-POIS} under all
four noise distributions, with larger gains under weak signals. Since the
two procedures use the same reported intervals and even-side $p$-values,
the higher coverage reflects the contribution of odd-side weighting.
As the signal becomes stronger, \texttt{Graph-POIS} already covers most
changes, leaving less room for improvement.

Full-sample conditional $e$-BH makes no rejections under either signal
strength. The requirement to preserve the reported interval, split
location, and family size leaves few accepted permutations. Each $e$-value
in \eqref{eq:full_conditional_mc_evalue} is bounded by one plus the number
of accepted permutations, so the low acceptance rate limits its ability
to reach the $e$-BH rejection thresholds with the fixed number of permutation draws.

\begin{figure}[t]
    \centering
    \includegraphics[width=0.98\textwidth]{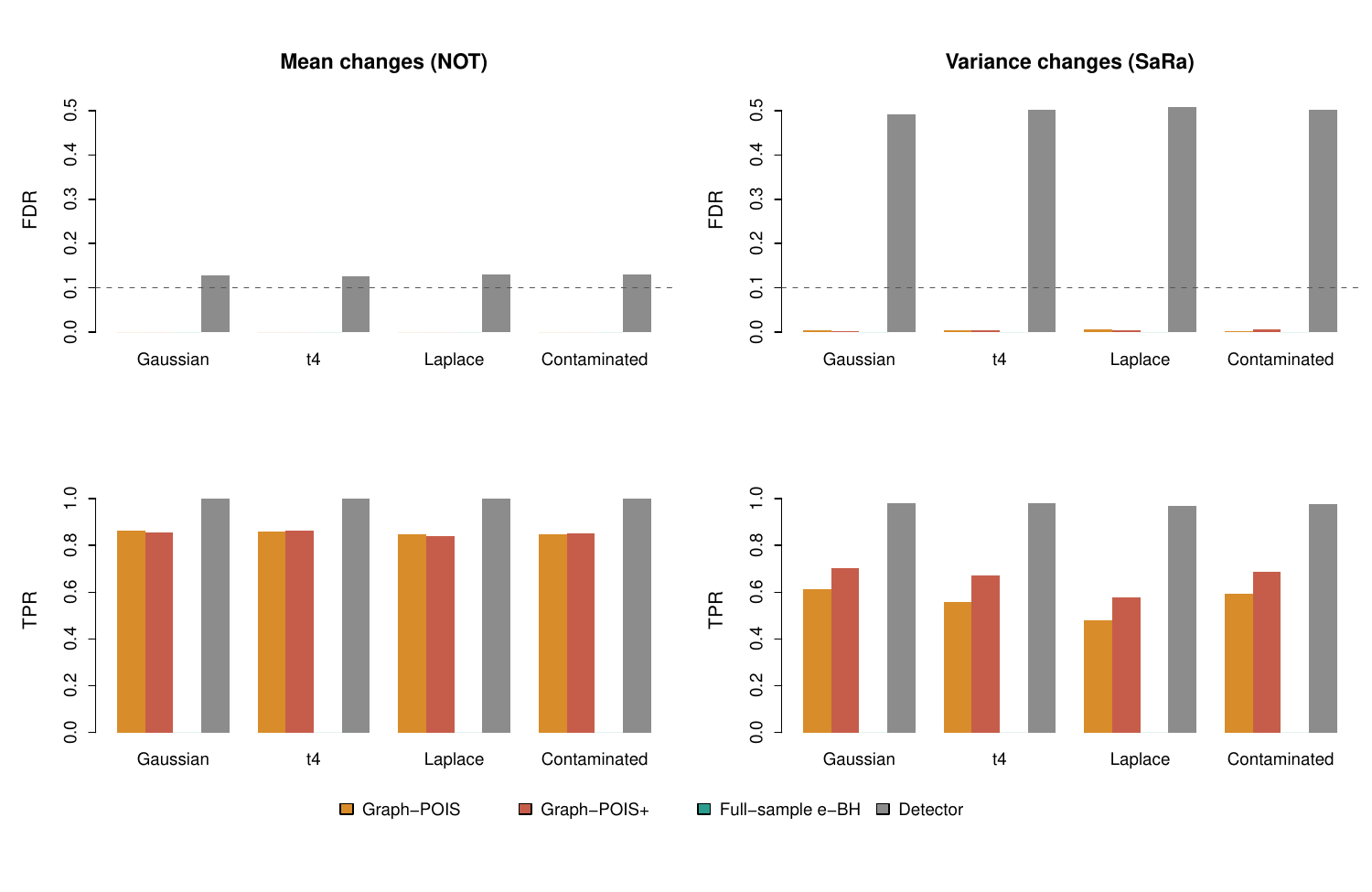}
    \caption{Weak-signal univariate experiments with
    \(c_\mu=c_\sigma=1\). Bars show means over 200 repetitions for
    Graph-POIS, Graph-POIS+, full-sample conditional e-BH, and the
    full-sample detector. Columns correspond to mean changes (NOT) and
    variance changes (SaRa). Rows show empirical interval-level FDR and
    coverage TPR. The dashed line marks \(\alpha=0.1\). Full-sample
    conditional e-BH has zero FDR and TPR because it makes no rejections.}
    \label{fig:univar-weak-epois}
\end{figure}

\begin{figure}[t]
    \centering
    \includegraphics[width=0.98\textwidth]{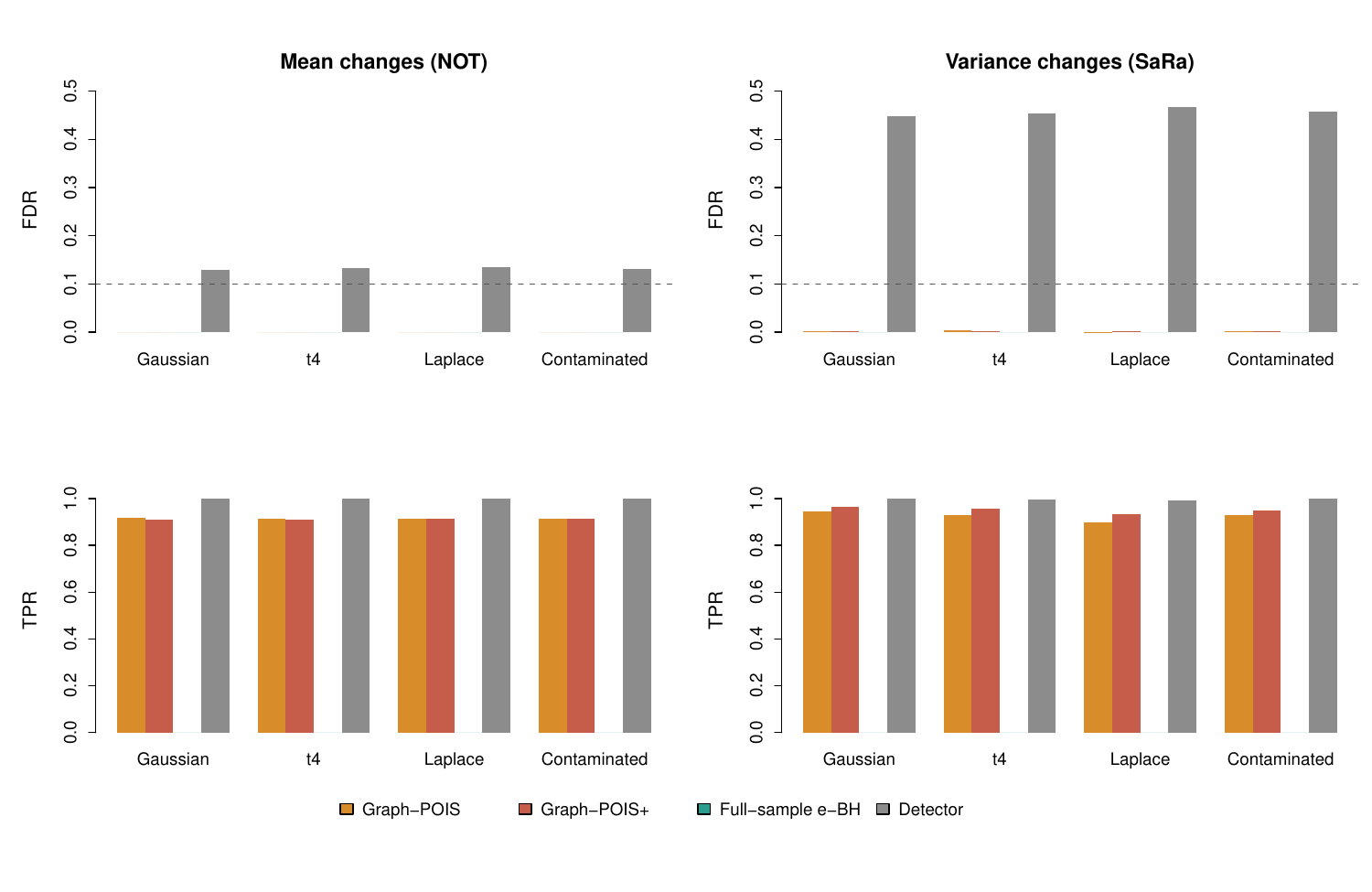}
    \caption{Strong-signal univariate experiments with
    \(c_\mu=2\) and \(c_\sigma=1.7\), using the same methods, numbers of permutation draws,
    transform, and plotting conventions as
    Figure~\ref{fig:univar-weak-epois}. Graph-based aggregation has higher
    coverage than under weak signals. Full-sample conditional e-BH still
    makes no rejections with this number of permutation draws.}
    \label{fig:univar-strong-epois}
\end{figure}

\section{Proofs of the FDR control theorems}
\subsection{Proof of Theorem~\ref{thm:FDR-POIS}}
\begin{proof}
Condition on the odd-side sequence $X_{\mathcal O}$ and any
data-independent randomization used to generate the interval pool or run
the detector.
For readability, write conditioning on $X_{\mathcal O}$ for this fixed
information in the displays below. Then the reported
index set \(\widehat{\mathcal S}\) and the reported intervals
\(\{I_k:k\in\widehat{\mathcal S}\}\) are fixed. If
\(|\widehat{\mathcal S}|=0\), then \(\widehat{\mathcal R}=\varnothing\) and
the claim is trivial. For any \(k\in\widehat{\mathcal S}\cap\mathcal H_0\), the interval null
\(H_{0k}\) implies that all observations in
\(I_k\) belong to the same stationary segment. Under the independent
multiple change-point model, the even-side observations in
\(I_k^{\mathcal E}=I_k\cap\mathcal E\) are therefore independent and
identically distributed. Hence they are exchangeable conditional on
\(X_{\mathcal O}\). To see why this gives validity of the Monte Carlo
p-value in \eqref{eq:perm-p-split}, condition further on the unordered
even observations in this interval. The observed labeled ordering is
uniform, and composing it with independent uniform permutations makes
the $B+1$ absolute orderings exchangeable. Applying the same score
functional therefore gives exchangeable scores
$V_k^{\mathcal E},V_{k,1}^{\mathcal E},\ldots,V_{k,B}^{\mathcal E}$.

For each coordinate, form its upper rank, counting ties with the
inequality $\ge$. Among $B+1$ scores, at most $\lfloor t(B+1)\rfloor$
coordinates have upper rank at most $t(B+1)$. Exchangeability gives the
same probability to each coordinate, so the probability that coordinate
zero has such a rank is at most $t$. This rank is the numerator in
\eqref{eq:perm-p-split}. Averaging over the unordered observations yields \(\mathbb P\left( P_k^{\mathcal E}\le t \mid X_{\mathcal O} \right) \le t\) for \(0\le t\le1\).
Thus the null even-side permutation \(p\)-values are conditionally
superuniform.

By Assumption~\ref{assum:disjoint_intervals}, the reported intervals are
disjoint. Therefore the even-side blocks
$\{I_k^{\mathcal E}:k\in\widehat{\mathcal S}\}$ are disjoint. Since the original
observations and the Monte Carlo permutations are independent across disjoint
blocks, the $p$-values
$\{P_k^{\mathcal E}:k\in\widehat{\mathcal S}\}$ are conditionally independent
given $X_{\mathcal O}$. Thus, applying the BH procedure \citep{benjamini1995controlling} to the
$p$-values indexed by $\widehat{\mathcal S}$ gives \(\mathbb E\left[ \frac{|\widehat{\mathcal R}\cap\mathcal H_0|} {|\widehat{\mathcal R}|\vee 1} \mid X_{\mathcal O} \right] \le \alpha \frac{|\widehat{\mathcal S}\cap\mathcal H_0|} {|\widehat{\mathcal S}|}\),
where $\widehat{\mathcal R}$ is the BH rejection set defined in
\eqref{eq:BH_rejected_set}. Taking expectations and using the convention that
the right-hand side is zero when $|\widehat{\mathcal S}|=0$, we obtain
\begin{equation*}
\mathrm{FDR}(\widehat{\mathcal R})
    =
    \mathbb E\left[
    \frac{|\widehat{\mathcal R}\cap\mathcal H_0|}
    {|\widehat{\mathcal R}|\vee 1}
    \right]
    \le
    \alpha
    \mathbb E\left[
    \frac{|\widehat{\mathcal S}\cap\mathcal H_0|}
    {|\widehat{\mathcal S}|\vee 1}
    \right]
    \le \alpha.
\end{equation*}
This proves the theorem.
\end{proof}

\subsection{Proof of Theorem~\ref{thm:FDR_POIS+}}
\begin{proof}
Let $\mathcal F_+^{\mathcal O}$ be the sigma-field generated by the odd-side
sequence $X_{\mathcal O}$ and by any additional odd-side randomization used to
construct the weights $\{W_k^{\mathcal O}:k\in\widehat{\mathcal S}\}$. Conditional
on $\mathcal F_+^{\mathcal O}$, the reported index set $\widehat{\mathcal S}$, the
reported intervals $\{I_k:k\in\widehat{\mathcal S}\}$, and the weights
$\{W_k^{\mathcal O}:k\in\widehat{\mathcal S}\}$ are fixed with respect to the
even-side permutation $p$-values. If $|\widehat{\mathcal S}|=0$, the claim is
trivial. Hence assume $|\widehat{\mathcal S}|\ge 1$.

By Assumption~\ref{assum:disjoint_intervals}, the even-side blocks
$\{I_k^{\mathcal E}:k\in\widehat{\mathcal S}\}$ are disjoint.
Each $P_k^{\mathcal E}$ is computed from its own even-side block and
independent permutation draws. The argument in the proof of
Theorem~\ref{thm:FDR-POIS} therefore applies:
the even-side permutation $p$-values
$\{P_k^{\mathcal E}:k\in\widehat{\mathcal S}\}$ are conditionally independent
given $\mathcal F_+^{\mathcal O}$, and $P_k^{\mathcal E}$ is conditionally
superuniform for every $k\in\widehat{\mathcal S}\cap\mathcal H_0$.

The rejection set $\widehat{\mathcal R}^+$ in
\eqref{eq:weighted_BH_rejected_set} is the BH rejection set applied to
the weighted $p$-values
$P_k^+=(P_k^{\mathcal E}/W_k^{\mathcal O})\wedge 1$, with the convention
$P_k^+=1$ when $W_k^{\mathcal O}=0$. We give the conditional weighted BH calculation explicitly
\citep{ignatiadis2024values}. For a fixed null $k$, let $r_k$
be the BH rejection count after replacing its weighted $p$-value by zero.
The leave-one-out identity proved in Lemma~\ref{lem:weighted_indbh},
applied to the independent set of all reported indices, gives \(k\in\widehat{\mathcal R}^{+} \Longleftrightarrow P_k^{\mathcal E}\le\alpha W_k^{\mathcal O}r_k/|\widehat{\mathcal S}|\) and \(|\widehat{\mathcal R}^{+}|=r_k\) on this event,
when $W_k^{\mathcal O}>0$. A zero weight contributes no rejection.
Here $r_k\ge1$ is determined by the other even-side permutation p-values and
$\mathcal F_+^{\mathcal O}$. Conditional independence and
superuniformity imply
\begin{equation*}
\begin{aligned}
\mathbb E\left[
 \frac{\mathbbm1\{k\in\widehat{\mathcal R}^{+}\}}
 {|\widehat{\mathcal R}^{+}|\vee1}\,
 \middle|\,\mathcal F_+^{\mathcal O}\right]
 &=\mathbb E\left[
 \frac{\mathbb P(P_k^{\mathcal E}\le\alpha W_k^{\mathcal O}r_k/|\widehat{\mathcal S}|
                  \mid\mathcal F_+^{\mathcal O},r_k)}{r_k}
 \,\middle|\,\mathcal F_+^{\mathcal O}\right]\\
 &\le\mathbb E\left[
 \frac{\min\{\alpha W_k^{\mathcal O}r_k/|\widehat{\mathcal S}|,1\}}{r_k}
 \,\middle|\,\mathcal F_+^{\mathcal O}\right]
 \le\frac{\alpha W_k^{\mathcal O}}{|\widehat{\mathcal S}|}.
\end{aligned}
\end{equation*}
Summing over the reported null indices gives
\begin{equation*}
\mathbb E\left[
    \frac{|\widehat{\mathcal R}^+\cap\mathcal H_0|}
    {|\widehat{\mathcal R}^+|\vee 1}
    \mid
    \mathcal F_+^{\mathcal O}
    \right]
\le
    \alpha
    \frac{1}{|\widehat{\mathcal S}|}
    \sum_{k\in\widehat{\mathcal S}\cap\mathcal H_0}
    W_k^{\mathcal O}.
\end{equation*}
Taking expectations and using the convention that the right-hand side is zero
when $|\widehat{\mathcal S}|=0$, we obtain
\begin{equation*}
\mathrm{FDR}(\widehat{\mathcal R}^+)
    =
    \mathbb E\left[
    \frac{|\widehat{\mathcal R}^+\cap\mathcal H_0|}
    {|\widehat{\mathcal R}^+|\vee 1}
    \right]\\
    \le
    \alpha
    \mathbb E\left[
    \frac{1}{|\widehat{\mathcal S}|\vee 1}
    \sum_{k\in\widehat{\mathcal S}\cap\mathcal H_0}
    W_k^{\mathcal O}
    \right].
\end{equation*}
This proves the theorem.
\end{proof}

\subsection{Proof of Theorem~\ref{thm:sufficient_cond_Evalue}}

\begin{proof}
Theorem~\ref{thm:FDR_POIS+} reduces the claim to the expected normalized
sum of null weights in \eqref{eq:post_detection_Evalue_error}. We bound
each null contribution conditional on its own detection information,
then sum. We may define $W_k^{\mathcal O}=0$ when $k$ is unreported,
since the corresponding contribution is multiplied by its detection
indicator. On the event
\(\{k\in\widehat{\mathcal S}\}\), the number of reported intervals
\(|\widehat{\mathcal S}|\) is measurable with respect to \(\Psi_k\). Hence
the factor
\(\mathbbm{1}\{k\in\widehat{\mathcal S}\}/(|\widehat{\mathcal S}|\vee 1)\)
is measurable with respect to the conditioning in
\eqref{eq:cond_Evalue}. Therefore, for each \(k\in\mathcal H_0\),
\begin{equation*}
\mathbb E\left[
    \frac{\mathbbm{1}\{k\in\widehat{\mathcal S}\}}
    {|\widehat{\mathcal S}|\vee 1}
    W_k^{\mathcal O}
    \right]
=
    \mathbb E\left[
    \frac{\mathbbm{1}\{k\in\widehat{\mathcal S}\}}
    {|\widehat{\mathcal S}|\vee 1}
    \mathbb E\left[
    W_k^{\mathcal O}
    \mid
    k\in\widehat{\mathcal S},\Psi_k
    \right]
    \right] 
 \le
    \mathbb E\left[
    \frac{\mathbbm{1}\{k\in\widehat{\mathcal S}\}}
    {|\widehat{\mathcal S}|\vee 1}
    \right],
\end{equation*}
where the inequality uses the detection-conditional \(e\)-value property
\eqref{eq:cond_Evalue}.

Summing over \(k\in\mathcal H_0\) gives
\begin{equation*}
\begin{aligned}
\mathbb E\left[
 \frac{1}{|\widehat{\mathcal S}|\vee1}
 \sum_{k\in\widehat{\mathcal S}\cap\mathcal H_0}W_k^{\mathcal O}
 \right]
 &=\sum_{k\in\mathcal H_0}
 \mathbb E\left[
 \frac{\mathbbm1\{k\in\widehat{\mathcal S}\}}{|\widehat{\mathcal S}|\vee1}
 W_k^{\mathcal O}\right]\\
 &\le\sum_{k\in\mathcal H_0}
 \mathbb E\left[
 \frac{\mathbbm1\{k\in\widehat{\mathcal S}\}}{|\widehat{\mathcal S}|\vee1}
 \right]
 =\mathbb E\left[
 \frac{|\widehat{\mathcal S}\cap\mathcal H_0|}{|\widehat{\mathcal S}|\vee1}
 \right]\le1.
\end{aligned}
\end{equation*}
This proves \eqref{eq:post_detection_Evalue_error}. The FDR control
\(\mathrm{FDR}(\widehat{\mathcal R}^+)\le\alpha\) then follows from
Theorem~\ref{thm:FDR_POIS+}.
\end{proof}

\subsection{Proof of Theorem~\ref{thm:perm-e-validity}}

\begin{proof}
Fix $k\in\mathcal H_0$. If the detector $\mathbb A$ uses external
randomization, condition on its realized value and use the same value
for every local permutation. This conditioning is removed at the end
by the tower property.

Fix the unordered observations inside $I_k^{\mathcal O}$ and the ordered odd observations outside it. Under $H_{0k}$, the observed ordering is uniform. Given any observed ordering, the independent uniform permutations produce independent uniform reorderings whose joint distribution does not depend on that ordering. Thus the observed and permuted subsequences are independent uniform arrangements under the stated conditioning. Set $\pi_{k,0}$ to the identity permutation. Then $\widehat{\mathcal S}^{\pi_{k,0},k}=\widehat{\mathcal S}$ and
$v_k(X_{\pi_{k,0}(I_k^{\mathcal O})})=V_k^{\mathcal O}$.
The detector and the complete scoring function are applied in the same way
to every arrangement, so the $B'+1$ pairs \(\left(\widehat{\mathcal S}^{\pi_{k,b},k},\, v_k(X_{\pi_{k,b}(I_k^{\mathcal O})})\right),\) are exchangeable for \(b=0,\ldots,B'\). All fractions below are zero when their denominator
is zero, and $W_k^{\mathcal O}$ is extended by zero when $k$ is unreported.

Fix a possible realized value \(m\) of the size component
\(|\widehat{\mathcal S}_{-k}|\) in \(\Psi_k\). By the definition of the
size-preserving accepted set in \eqref{eq:accepted_index}, on the event
\(k\in\widehat{\mathcal S}\) and
\(|\widehat{\mathcal S}_{-k}|=m\),
\begin{equation*}
\begin{aligned}
&\mathbbm{1}\left\{
    k\in\widehat{\mathcal S},\
    |\widehat{\mathcal S}_{-k}|=m
    \right\}
    W_k^{\mathcal O}\\
&=\mathbbm{1}\left\{
k\in\widehat{\mathcal S},\
|\widehat{\mathcal S}_{-k}|=m
\right\}\times
\frac{
\left[
\sum_{b=0}^{B'}
\mathbbm{1}\left\{
k\in\widehat{\mathcal S}^{\pi_{k,b},k},\
|\widehat{\mathcal S}^{\pi_{k,b},k}_{-k}|=m
\right\}
\right]
g(V_k^{\mathcal O})
}{
\sum_{b=0}^{B'}
\mathbbm{1}\left\{
k\in\widehat{\mathcal S}^{\pi_{k,b},k},\
|\widehat{\mathcal S}^{\pi_{k,b},k}_{-k}|=m
\right\}
g\{v_k(X_{\pi_{k,b}(I_k^{\mathcal O})})\}
}.
\end{aligned}
\end{equation*}
Conditional on the unordered observations inside \(I_k^{\mathcal O}\) and the
ordered observations outside \(I_k^{\mathcal O}\), exchangeability of the $B'+1$ pairs above gives
\begin{equation*}
\begin{aligned}
&\mathbb E_{H_{0k}}\Bigg[
\mathbbm{1}\left\{
k\in\widehat{\mathcal S},\
|\widehat{\mathcal S}_{-k}|=m
\right\}
W_k^{\mathcal O}
\ \Bigg|
\lbag X_t\rbag_{t\in I_k^{\mathcal O}},
(X_t)_{t\in\mathcal O\setminus I_k^{\mathcal O}}
\Bigg]
\\
&=
\frac{1}{B'+1}
\mathbb E_{H_{0k}}\Bigg[
\sum_{j=0}^{B'}
\mathbbm{1}\left\{
k\in\widehat{\mathcal S}^{\pi_{k,j},k},\
|\widehat{\mathcal S}^{\pi_{k,j},k}_{-k}|=m
\right\}
\\
&\qquad\qquad\times
\frac{
\left[
\sum_{b=0}^{B'}
\mathbbm{1}\left\{
k\in\widehat{\mathcal S}^{\pi_{k,b},k},\
|\widehat{\mathcal S}^{\pi_{k,b},k}_{-k}|=m
\right\}
\right]
g\{v_k(X_{\pi_{k,j}(I_k^{\mathcal O})})\}
}{
\sum_{b=0}^{B'}
\mathbbm{1}\left\{
k\in\widehat{\mathcal S}^{\pi_{k,b},k},\
|\widehat{\mathcal S}^{\pi_{k,b},k}_{-k}|=m
\right\}
g\{v_k(X_{\pi_{k,b}(I_k^{\mathcal O})})\}
}
\ \Bigg|
\lbag X_t\rbag_{t\in I_k^{\mathcal O}},
(X_t)_{t\in\mathcal O\setminus I_k^{\mathcal O}}
\Bigg]
\\
&\le
\frac{1}{B'+1}
\mathbb E_{H_{0k}}\Bigg[
\sum_{b=0}^{B'}
\mathbbm{1}\left\{
k\in\widehat{\mathcal S}^{\pi_{k,b},k},\
|\widehat{\mathcal S}^{\pi_{k,b},k}_{-k}|=m
\right\}
\ \Bigg|
\lbag X_t\rbag_{t\in I_k^{\mathcal O}},
(X_t)_{t\in\mathcal O\setminus I_k^{\mathcal O}}
\Bigg]
\\
&=
\mathbb P_{H_{0k}}\left(
k\in\widehat{\mathcal S},\
|\widehat{\mathcal S}_{-k}|=m
\ \middle|\
\lbag X_t\rbag_{t\in I_k^{\mathcal O}},
(X_t)_{t\in\mathcal O\setminus I_k^{\mathcal O}}
\right).
\end{aligned}
\end{equation*}
The first equality symmetrizes over the \(B'+1\) exchangeable pairs
of detector outputs and scores. In the
inequality, when the denominator is positive, the denominator cancels with the
corresponding weighted sum in the numerator. When the denominator is zero, the
fraction is zero by convention. The last equality again follows from
exchangeability.

Now condition additionally on
\(\{k\in\widehat{\mathcal S},|\widehat{\mathcal S}_{-k}|=m\}\). Dividing the
previous display by the corresponding conditional probability, whenever this
probability is positive, gives
\begin{equation*}
\begin{aligned}
&\mathbb E_{H_{0k}}\left[
    W_k^{\mathcal O}
    \,\middle|\,
    k\in\widehat{\mathcal S},\
    |\widehat{\mathcal S}_{-k}|=m,\
    \lbag X_t\rbag_{t\in I_k^{\mathcal O}},\
    (X_t)_{t\in\mathcal O\setminus I_k^{\mathcal O}}
    \right]
    \le 1 .
\end{aligned}
\end{equation*}
Since \(m\) is an arbitrary realized value of
\(|\widehat{\mathcal S}_{-k}|\), this proves \(\mathbb E_{H_{0k}}\left[ W_k^{\mathcal O} \,\middle|\, k\in\widehat{\mathcal S},\Psi_k \right] \le 1\).
Removing the temporary conditioning on detector randomization by the tower
property gives the same inequality without conditioning on that randomization.
Thus \eqref{eq:cond_Evalue} holds.

It remains to verify the measurability condition used in
Theorem~\ref{thm:sufficient_cond_Evalue}. By the definition of \(\Psi_k\), the
size \(|\widehat{\mathcal S}_{-k}|\) is \(\Psi_k\)-measurable. Hence, on
\(k\in\widehat{\mathcal S}\), \(|\widehat{\mathcal S}|=1+
|\widehat{\mathcal S}_{-k}|\) is \(\Psi_k\)-measurable. Therefore the
conditions of Theorem~\ref{thm:sufficient_cond_Evalue} hold for every
\(k\in\mathcal H_0\). Under Assumption~\ref{assum:disjoint_intervals},
combining Theorems~\ref{thm:sufficient_cond_Evalue} and
\ref{thm:FDR_POIS+} gives \(\mathrm{FDR}(\widehat{\mathcal R}^+)\le\alpha\).
This proves the theorem.
\end{proof}

\subsection{Proof of Proposition~\ref{prop:window_cond_invariance}}
\label{app:proof_window_cond_invariance}

\begin{proof}
Fix $k\in\widehat{\mathcal S}$. We prove the two claims separately:
the local check preserves the entire detection vector, and its exact
accepted set gives the conditional e-value property. For $\pi\in\Pi_k$, let
\(D_m^\pi\) denote the detection indicator obtained by applying the
local-window detector to \(X_{\mathcal O}^{\pi,k}\). Let \(\pi\in\widetilde\Pi_k^{\rm win}\). By definition,
\(D_m^\pi=D_m\) whenever \(|m-k|\le3h+\rho_h\). In particular,
\(D_k^\pi=D_k=1\).

Now consider \(m\in[T]\) satisfying
\(|m-k|>3h+\rho_h\). The score \(\widehat Q_t\) depends only on the
observations in \(I_t^{\mathcal O}\), and hence
\(\widehat Q_t^\pi=\widehat Q_t\) whenever
\(I_t^{\mathcal O}\cap I_k^{\mathcal O}=\varnothing\). For the
local-window family, two odd-side windows can intersect only if their
centers are within distance \(2h\). Thus
\(\widehat Q_t^\pi=\widehat Q_t\) whenever \(|t-k|>2h\).

The indicator \(D_m\) is a deterministic function of
\(\{\widehat Q_t:|t-m|\le h+\rho_h\}\). For every such \(t\), \(|t-k| \ge |m-k|-|t-m| > (3h+\rho_h)-(h+\rho_h) = 2h\).
Therefore every score entering the rule for \(D_m\) is unchanged, and
\(D_m^\pi=D_m\). Combining the two regions shows that
\(D_m^\pi=D_m\) for every \(m\in[T]\). Consequently,
\begin{equation*}
\widehat{\mathcal S}^{\pi,k}
    =
    \{m\in[T]:D_m^\pi=1\}
    =
    \{m\in[T]:D_m=1\}
    =
    \widehat{\mathcal S}.
\end{equation*}

For the conditional e-value property, we need the converse as well.
If \(\pi\in\Pi_k\) preserves the complete detection
vector, then \(D_m^\pi=D_m\) in particular for every
\(|m-k|\le3h+\rho_h\), and hence
\(\pi\in\widetilde\Pi_k^{\rm win}\). Together with the preceding
detection-invariance argument, this gives
\begin{equation}
    \widetilde\Pi_k^{\rm win}
    =
    \left\{
    \pi\in\Pi_k:
    (D_m^\pi)_{m\in[T]}
    =
    (D_m)_{m\in[T]}
    \right\}.
    \label{eq:window_selection_event}
\end{equation}

Now suppose additionally that \(k\in\mathcal H_0\), and define \(\Phi_k^{\rm win}=\left(\Psi_k,\,(D_m)_{m\in[T]}\right)\). Under \(H_{0k}\), conditional on the unordered observations inside
\(I_k^{\mathcal O}\) and the ordered observations outside \(I_k^{\mathcal O}\), the observed absolute ordering is uniform. Conditioning further on \(\Phi_k^{\rm win}\) restricts it to all orderings with the observed detection vector. By \eqref{eq:window_selection_event}, these are exactly the orderings generated by accepted relative permutations. The same scoring function
is evaluated on each arrangement. Thus the symmetrization calculation in Theorem~\ref{thm:perm-e-validity} applies with this refined condition and gives \(\mathbb E_{H_{0k}} \left[ W_k^{\mathcal O} \mid k\in\widehat{\mathcal S},\Phi_k^{\rm win} \right] \le1\).
Since \(\Phi_k^{\rm win}\) refines \(\Psi_k\), the tower property
yields \(\mathbb E_{H_{0k}} \left[ W_k^{\mathcal O} \mid k\in\widehat{\mathcal S},\Psi_k \right] \le1\).
The FDR statement follows from
Theorems~\ref{thm:sufficient_cond_Evalue}
and~\ref{thm:FDR_POIS+}.
\end{proof}

\subsection{Proof of Proposition~\ref{prop:narrow_cond_invariance}}
\label{app:proof_narrow_cond_invariance}

\begin{proof}
Fix \(k\in\widehat{\mathcal S}\), and hold any detector randomization fixed
throughout. Write \(\mathcal P_k=(G_{k,0},\ldots,G_{k,L_k})\), with
\(G_{k,0}=[1,T]\), and let \(j_\ell=\hat k(G_{k,\ell})\), so that
\(j_{L_k}=k\).

Take \(\pi\in\widetilde\Pi_k^{\rm nar}\), and let \(G_{k,\ell}^{\pi}\)
denote the successive search segments containing \(I_k\) in the permuted
recursion. The path check gives \(\hat k^\pi(G_{k,\ell})=j_\ell\) for
every \(0\le\ell\le L_k\). Both recursions start at
\(G_{k,0}^{\pi}=G_{k,0}=[1,T]\). Suppose that
\(G_{k,\ell}^{\pi}=G_{k,\ell}\) for some \(\ell<L_k\). Both recursions
then choose \(I_{j_\ell}\), and the same split produces the same child
segment containing \(I_k\), so \(G_{k,\ell+1}^{\pi}=G_{k,\ell+1}\).
Induction therefore gives
\begin{equation*}
G_{k,\ell}^{\pi}=G_{k,\ell},
 \qquad
 \hat k^\pi(G_{k,\ell}^{\pi})=\hat k(G_{k,\ell})=j_\ell,
 \qquad 0\le\ell\le L_k.
\end{equation*}
In particular, \(\hat k^\pi(G_{k,L_k}^{\pi})=j_{L_k}=k\), so the permuted
recursion also reports \(I_k\) on \(G_{k,L_k}\).

Every branch separating from \(\mathcal P_k\), and every segment generated
after \(I_k\) is reported, is disjoint from \(I_k\). For any such segment
\(G\) and any candidate \(I_j\subseteq G\),
\(I_j^{\mathcal O}\cap I_k^{\mathcal O}=\varnothing\), so the local
permutation leaves the observations used by \(\widehat Q_j\) unchanged.
The deterministic search and tie-breaking rules consequently give
\begin{equation*}
\widehat Q_j^\pi=\widehat Q_j\quad\text{for all }I_j\subseteq G,
 \qquad
 \hat k^\pi(G)=\hat k(G).
\end{equation*}
Applying this argument recursively on these branches yields
\(\widehat{\mathcal S}^{\pi,k}=\widehat{\mathcal S}\).

For an arrangement \(x\) of the observations in \(I_k^{\mathcal O}\), let
\(\Theta_k(x)\) be the ordered list of search segments and chosen indices
on the path to \(I_k\), with the outside observations fixed. Set
\(\Theta_k(x)=\varnothing\) when \(I_k\) is not reported. The preceding
induction shows that every accepted permutation preserves this list.
Conversely, preserving the list satisfies the path check at every observed
segment. Hence,
\begin{equation*}
\widetilde\Pi_k^{\rm nar}
 =\left\{\pi\in\Pi_k:
 \Theta_k\!\left(X_{\pi(I_k^{\mathcal O})}\right)
 =\Theta_k\!\left(X_{I_k^{\mathcal O}}\right)\right\}.
\end{equation*}
Define \(\Phi_k^{\rm nar}=(\Psi_k,\Theta_k(X_{I_k^{\mathcal O}}))\).
Now let \(k\in\mathcal H_0\). Conditional on the unordered observations
inside \(I_k^{\mathcal O}\) and the ordered observations outside it, the
observed absolute ordering is uniform. Conditioning further on
\(\Phi_k^{\rm nar}\) restricts it to all arrangements with the observed
path. By the equivalence above, these are exactly the arrangements
generated by accepted relative permutations. Since they also give the
same reported family, \(|\widehat{\mathcal S}_{-k}|\) imposes no further
restriction.

Conditional also on which proposals are accepted, the observed and
accepted arrangements are independent and uniform on this same set.
The same scoring function is applied to every
arrangement. Thus the symmetrization calculation in
Theorem~\ref{thm:perm-e-validity}, with the same zero-denominator
convention, applies under this refined detection information. Averaging
over the accepted proposal indices gives \(\mathbb E_{H_{0k}}\left[ W_k^{\mathcal O}\mid k\in\widehat{\mathcal S},\Phi_k^{\rm nar} \right]\le1\).
Since \(\Phi_k^{\rm nar}\) refines \(\Psi_k\), the tower property yields \(\mathbb E_{H_{0k}}\left[ W_k^{\mathcal O}\mid k\in\widehat{\mathcal S},\Psi_k \right]\le1\).
Finally, the tower property removes the conditioning on detector randomization.
\end{proof}

\subsection{Proofs of the computational cost statements}
\label{app:computational_cost}

The bounds below are deterministic, conditional on the reported interval
collection and on any external randomization used by the detector. For
\(B,B'\ge1\), we use the linear score-evaluation cost stated in
Propositions~\ref{prop:window_cost} and~\ref{prop:narrowest_cost}:
evaluating a score and applying a local permutation on a block \(J\)
costs \(O(|J|)\). The score-evaluation cost is also linear in the interval length
for the contrast statistics used in narrowest-over-threshold searches
\citep{baranowski2019narrowest}.

\begin{proof}[Proof of Proposition~\ref{prop:window_cost}]
The initial odd-side local-window scan costs
\(\mathfrak C_{\mathrm{win}}\).
Conditional on the odd-side data, the reported intervals are fixed and
disjoint. Their even-side projections are therefore disjoint as well, so
\begin{equation}
    \sum_{k\in\widehat{\mathcal S}}
    |I_k^{\mathcal E}|
    \le T.
    \label{eq:even_projection_total_length}
\end{equation}
For each reported interval, \texttt{POIS} computes one observed
score and \(B\) permuted scores. The total cost of computing these scores is thus \(O\left\{ (B+1)\sum_{k\in\widehat{\mathcal S}}|I_k^{\mathcal E}| \right\}=O(BT)\).
Sorting the even-side permutation \(p\)-values and applying BH costs
\(O\{|\widehat{\mathcal S}|\log(|\widehat{\mathcal S}|\vee2)\}\).
The sorting cost is zero for an empty reported family.
These calculations also apply to any detector with disjoint reported
intervals, replacing \(\mathfrak C_{\mathrm{win}}\) by its initial cost
\(\mathfrak C_{\mathbb A}\).

A generic implementation of \texttt{POIS+} reruns the detector for each
reported interval and each of its \(B'\) local permutations, at cost
\(O(B'|\widehat{\mathcal S}|\mathfrak C_{\mathbb A})\).
Local permutations and odd-side score evaluations contribute
\(O\{B'\sum_{k\in\widehat{\mathcal S}}|I_k^{\mathcal O}|\}=O(B'T)\)
by disjointness. This term can be omitted when it is already included in
the detector cost. The local checks below avoid these complete reruns.

It remains to bound the odd-side detection-preserving cost. For a reported
center \(k\), Proposition~\ref{prop:window_cond_invariance} requires
recomputing the affected scan scores and detection indicators only for
centers within distance \(3h+\rho_h\) of \(k\). By the blockwise
linear-time assumption in Proposition~\ref{prop:window_cost}, one local
permutation consequently costs \(O\{\min(T,h+\rho_h)\}\), since there are
at most \(T\) centers. Thus the total odd-side
cost is
\begin{equation}
    O\left\{
    B'
    \sum_{k\in\widehat{\mathcal S}}
    \min\{T,h+\rho_h\}
    \right\}.
    \label{eq:window_odd_raw_cost}
\end{equation}

Distinct reported centers are at least \(h\) apart by the first
local-maximum condition and more than \(\rho_h\) apart by the second.
Set \(d_h=\max\{h,\rho_h\}\). If \(d_h\le T\), then
\((|\widehat{\mathcal S}|-1)d_h\le T-1\). If \(d_h>T\), there is at most
one reported center. Hence, \(|\widehat{\mathcal S}|\min\{T,h+\rho_h\} \le 2|\widehat{\mathcal S}|\min\{T,d_h\} \le 4T\).
This reduces \eqref{eq:window_odd_raw_cost} to \(O(B'T)\).

Finally, sorting the weighted \(p\)-values costs
\(O(|\widehat{\mathcal S}|\log|\widehat{\mathcal S}|)\). Combining the
initial scan, even-side permutation testing, odd-side detection-preserving checks,
and sorting gives \(\mathfrak C_{\mathrm{POIS+}}^{\mathrm{win}} = \mathfrak C_{\mathrm{win}} + O\left\{ (B+B')T + |\widehat{\mathcal S}| \log|\widehat{\mathcal S}| \right\}\).
\end{proof}

\begin{proof}[Proof of Proposition~\ref{prop:narrowest_cost}]
Retain the notation from the preceding proof, and let \(L_{\rm pool}^{\mathcal O}:=\sum_{I_j\in\mathcal I}|I_j^{\mathcal O}|\). The initial narrowest-first search costs
\(\mathfrak C_{\mathrm{nar}}\). The cost of even-side permutation testing is \(O(BT)\) by
\eqref{eq:even_projection_total_length} and the preceding proof.
Likewise, the odd-side scores on the reported intervals, including the
detector statistics on those intervals when these are computed separately, cost \(O\left\{ B'\sum_{k\in\widehat{\mathcal S}} |I_k^{\mathcal O}| \right\} = O(B'L_{\rm pool}^{\mathcal O})\).
It remains to bound the detection-preserving path checks. For this purpose
only, define
\begin{equation*}
\mathcal B_k^{\rm upd}
:=
\bigcup_{G\in\mathcal P_k}
\left\{
I_j\in\mathcal I:
I_j\subseteq G,\;
j\ne\hat k(G),\;
|I_j|\le |I_{\hat k(G)}|,\;
I_j^{\mathcal O}\cap I_k^{\mathcal O}\ne\varnothing
\right\}.
\end{equation*}
and let \(\mathfrak L_{\rm upd} := \sum_{k\in\widehat{\mathcal S}} \sum_{I_j\in\mathcal B_k^{\rm upd}} |I_j^{\mathcal O}|\).
The union defining \(\mathcal B_k^{\rm upd}\) is understood as a set, so a
candidate appearing on several segments of \(\mathcal P_k\) is counted only
once.

Fix \(I_j\in\mathcal B_k^{\rm upd}\). For some
\(G\in\mathcal P_k\), we have
\(|I_j|\le |I_{\hat k(G)}|\) and
\(I_j^{\mathcal O}\cap I_k^{\mathcal O}\ne\varnothing\).
Since \(I_k\) is an over-threshold candidate contained in \(G\), the
narrowest-first rule implies
\(|I_{\hat k(G)}|\le |I_k|\). Thus every pair \((j,k)\) counted above satisfies \(|I_j|\le |I_k|\) and \(I_j\cap I_k\ne\varnothing\). For a fixed \(I_j\), there are at most two reported intervals with these properties. Indeed, if three pairwise disjoint reported intervals of length at least \(|I_j|\) all intersected \(I_j\), the middle one would be strictly contained in \(I_j\), contradicting its length. It follows that
\begin{equation}
    \mathfrak L_{\rm upd}
    \le
    2\sum_{I_j\in\mathcal I}
    |I_j^{\mathcal O}|
    =
    2L_{\rm pool}^{\mathcal O}.
    \label{eq:narrow_charging_bound}
\end{equation}

The bound \eqref{eq:narrow_charging_bound} controls score recomputations.
To count the comparisons, write
\(\mathcal P_k=(G_{k,0},\ldots,G_{k,L_k})\), where
\(G_{k,L_k}=G_k\), and let \(j_\ell=\hat k(G_{k,\ell})\), so that
\(j_{L_k}=k\). The segments \(G_{k,0},\ldots,G_{k,L_k-1}\) are those
visited along this path before \(I_k\) is reported.

For \(0\le\ell\le L_k\), preserving the selected index \(j_\ell\) on \(G_{k,\ell}\) is
equivalent to \(\widehat Q_{j_\ell}^\pi>\lambda\) and the following
conditions for all other candidates \(I_j\subseteq G_{k,\ell}\):
\begin{equation*}
\begin{aligned}
 \widehat Q_j^\pi&\le\lambda,
 &&  \text{for }|I_j|<|I_{j_\ell}|,\\
 \widehat Q_j^\pi&\le\widehat Q_{j_\ell}^\pi,
 && \text{for }|I_j|=|I_{j_\ell}|,\quad j\ne j_\ell.
\end{aligned}
\end{equation*}
The second inequality is strict when the fixed tie-breaking rule prefers
\(j\) to \(j_\ell\). Longer candidates cannot change the chosen index.
For \(\ell<L_k\), the reported intervals are disjoint from \(I_k\), so
\begin{equation*}
I_{j_\ell}\cap I_k=\varnothing,
 \qquad
 \widehat Q_{j_\ell}^\pi=\widehat Q_{j_\ell}>\lambda,
 \qquad 0\le\ell<L_k.
\end{equation*}
Comparisons involving unchanged candidate scores therefore remain satisfied
on these segments. For a fixed affected candidate, the remaining conditions
are upper bounds on the single value \(\widehat Q_j^\pi\), with fixed
right-hand sides. When such bounds exist, retain their minimum. Equality
is allowed only if every condition attaining that minimum is non-strict.
Thus these conditions require at most one comparison per distinct
affected candidate. The initial
search stores these bounds and whether equality is allowed, together with
\(\mathcal B_k^{\rm upd}\). This one-time construction is included in
\(\mathfrak C_{\mathrm{nar}}\).

On the final segment \(G_k\), require \(\widehat Q_k^\pi>\lambda\) and
apply the same two inequalities with \(j_{L_k}=k\), using the updated
score \(\widehat Q_k^\pi\). Each affected competitor needs one comparison.
Unchanged shorter candidates remain below threshold. Among unchanged
candidates of length \(|I_k|\), it suffices to compare the one with the
largest observed score against \(\widehat Q_k^\pi\), resolving ties by
the fixed index order. This candidate, when present, is also stored during
the initial search.

For each \((k,\pi)\), compute the detector statistic once on each distinct
interval in \(\{I_k\}\cup\mathcal B_k^{\rm upd}\) and reuse it in these
comparisons.
Since \(I_k^{\mathcal O}\ne\varnothing\), every
\(I_j\in\mathcal B_k^{\rm upd}\) also has a nonempty odd-side block, so \(1+|\mathcal B_k^{\rm upd}| \le |I_k^{\mathcal O}| +\sum_{I_j\in\mathcal B_k^{\rm upd}}|I_j^{\mathcal O}|\).
The comparisons are therefore absorbed by the score-evaluation costs.
Summing over the reported intervals and the \(B'\)
permutations, and using \eqref{eq:narrow_charging_bound}, gives the total
odd-side cost \(O\!\left\{B'\left( \sum_{k\in\widehat{\mathcal S}}|I_k^{\mathcal O}| +\mathfrak L_{\rm upd}\right)\right\} =O(B'L_{\rm pool}^{\mathcal O})\).
Finally, sorting the weighted \(p\)-values costs
\(O\{|\widehat{\mathcal S}|\log|\widehat{\mathcal S}|\}\). Combining the
preceding bounds gives
\begin{equation*}
\mathfrak C_{\mathrm{POIS+}}^{\mathrm{nar}}
    =
    \mathfrak C_{\mathrm{nar}}
    +
    O\left\{
    BT
    +
    B'L_{\rm pool}^{\mathcal O}
    +
    |\widehat{\mathcal S}|
    \log|\widehat{\mathcal S}|
    \right\}.
\end{equation*}
For seeded intervals \citep{kovacs2023seeded},
\(L_{\rm pool}^{\mathcal O}
\le\sum_{I_j\in\mathcal I}|I_j|
=O(T\log T)\), which gives the stated
\(O(B'T\log T)\) odd-side cost.
\end{proof}

\section{Proofs of the power results}
\label{app:power}

All probabilities below include the Monte Carlo randomization. A
random interval pool independent of the observations can be handled by
conditioning on its generating randomization and then taking expectations.
Throughout, $\alpha\in(0,1)$ is fixed. The numbers $B\ge1$ and $B'\ge1$
of even-side and odd-side permutations may depend on $T$. The distinct target
indices are fixed conditional on the interval pool. Let $\mathcal F^{\mathcal O}$ be generated by the odd-side observations and the independent randomization
specifying the detector and scoring functions, excluding the odd-side
proposal draws and all even-side observations and permutation draws.
Set $\mathcal A_j=\varnothing$ for unreported targets. For reported
targets, $\mathcal A_j$ indexes the accepted draws among the $B'$ proposals.
Repeated accepted draws are counted separately, and the observed score
is not one of these proposals.

\subsection{Proof of Theorem~\ref{thm:pois_power}}
\label{app:power:pois_proof}
\label{app:power:main_reduction}

\begin{corollary}[Realized and simultaneous coverage]
\label{cor:pois_realized_coverage}
Under the assumptions of Theorem~\ref{thm:pois_power},
$\operatorname{TPP}(\widehat{\mathcal R})=1-o_{\mathbb P}(1)$.
If Assumption~\ref{ass:power_main}\textup{(ii)} is strengthened to
$\sum_{j=1}^{K^*}q_{j,T}^{\mathcal E}=o(1)$, then
$\mathbb P\{\operatorname{TPP}(\widehat{\mathcal R})=1\}=1-o(1)$.
\end{corollary}

\begin{proof}[Proof of Theorem~\ref{thm:pois_power} and Corollary~\ref{cor:pois_realized_coverage}]
We use Assumption~\ref{ass:power_main}\textup{(i)} to identify change-point coverage with
rejection of the target intervals. Assumption~\ref{ass:power_main}\textup{(ii)} then shows
that most targets have the smallest attainable permutation p-value.
The remaining step is to verify that these targets form a feasible set
for the BH rule in \eqref{eq:BH_rejected_set}.
Write $m_T=|\widehat{\mathcal S}|$, $x_0=(B+1)^{-1}$, and let \(\mathcal M_T=\bigl\{\{k_1^*,\ldots,k_{K^*}^*\}\subseteq\widehat{\mathcal S},\ m_T\le C_0K^*\bigr\}\).
By Assumption~\ref{ass:power_main}\textup{(i)}, $\mathbb P(\mathcal M_T^c)\to0$. On $\mathcal M_T$, each
change-point is contained in its reported target, and no other reported
interval contains that change-point by disjointness. If one target contained a second change-point, it would intersect that
change-point's distinct reported target. This would contradict
disjointness. Conversely, any additional reported non-null interval would
intersect the reported target of a change-point it contains. Therefore
the reported non-null intervals are exactly the
$K^*$ targets, and each contains exactly one change-point. In particular,
$m_T\ge K^*$ and, for every $\mathcal R\subseteq\widehat{\mathcal S}$,
\begin{equation}
 \operatorname{TPP}(\mathcal R)
 =\frac1{K^*}\sum_{j=1}^{K^*}\mathbbm1\{k_j^*\in\mathcal R\}
 \quad\text{on }\mathcal M_T.
 \label{eq:coverage_target_equality}
\end{equation}
The right-hand side remains a lower bound outside $\mathcal M_T$.

For the BH rule, let
$q_1,\ldots,q_m\in[0,1]$ and write
$r=\max\{s\in[m]:q_{(s)}\le\alpha s/m\}$, with $r=0$ when the set is empty.
When $r>0$, exactly $r$ values are at most $\alpha r/m$: otherwise, for
some $s>r$, $q_{(s)}\le\alpha r/m\le\alpha s/m$, a contradiction.
The same reasoning excludes a tie beyond rank $r$ at $q_{(r)}$.
Consequently, the rejection set is $\{i:q_i\le\alpha r/m\}$, and
\begin{equation}
 \bigl\{q_i\le\alpha|A|/m\text{ for every }i\in A\bigr\}
 \quad\Longrightarrow\quad
 A\subseteq\mathcal R_{\rm BH}(q),
 \qquad A\subseteq[m].
 \label{eq:BH_feasible_subset}
\end{equation}
Indeed, for $A\ne\varnothing$ the displayed condition gives
$q_{(|A|)}\le\alpha|A|/m$, hence $r\ge|A|$. The empty case is immediate.
This argument applies to the weighted $p$-values as well.

By the definition of the permutation $p$-value,
\begin{equation}
 \{P_{k_j^*}^{\mathcal E}>x_0\}
 =\left\{\max_{1\le b\le B}V_{k_j^*,b}^{\mathcal E}
                            \ge V_{k_j^*}^{\mathcal E}\right\}.
 \label{eq:minimal_permutation_event}
\end{equation}
Let $n_T=\sum_{j=1}^{K^*}\mathbbm1\{P_{k_j^*}^{\mathcal E}=x_0\}$.
Then, for each fixed $d\in(0,1)$, Markov's inequality gives
\begin{equation*}
\mathbb P(n_T<dK^*)
 \le\frac{\mathbb E(K^*-n_T)}{(1-d)K^*}
 =\frac{\sum_{j=1}^{K^*}q_{j,T}^{\mathcal E}}{(1-d)K^*}
 \longrightarrow0,
\end{equation*}
where the equality follows from \eqref{eq:minimal_permutation_event} and
the convergence follows from Assumption~\ref{ass:power_main}\textup{(ii)}.

Since $B+1$ is an integer strictly exceeding $C_0/\alpha$,
$\sup_T C_0x_0/\alpha<1$. Fix $\varepsilon\in(0,1)$ and choose
$d\in(\max\{1-\varepsilon,\sup_T C_0x_0/\alpha\},1)$.
On $\mathcal M_T\cap\{n_T\ge dK^*\}$, all targets counted by $n_T$ are
reported and satisfy \(P_{k_j^*}^{\mathcal E}=x_0<\frac{\alpha d}{C_0} \le\frac{\alpha n_T}{m_T}\).
They are therefore rejected by \eqref{eq:BH_feasible_subset}, and their
coverage is at least $d>1-\varepsilon$. It follows that
\begin{equation*}
\begin{aligned}
 \mathbb P\{\operatorname{TPP}(\widehat{\mathcal R})\le1-\varepsilon\}
 &\le\mathbb P(\mathcal M_T^c)+\mathbb P(n_T<dK^*)\\
 &\le\mathbb P(\mathcal M_T^c)
       +\frac{\sum_{j=1}^{K^*}q_{j,T}^{\mathcal E}}{(1-d)K^*}
 \longrightarrow0.
 \end{aligned}
\end{equation*}
Thus $\operatorname{TPP}(\widehat{\mathcal R})\to1$ in probability.
Since $0\le\operatorname{TPP}(\widehat{\mathcal R})\le1$, uniform integrability
also gives $\operatorname{TPR}(\widehat{\mathcal R})\to1$.

Under $\sum_jq_{j,T}^{\mathcal E}=o(1)$, the union bound and
\eqref{eq:minimal_permutation_event} give
$\mathbb P(n_T<K^*)\le\sum_jq_{j,T}^{\mathcal E}\to0$.
On $\mathcal M_T\cap\{n_T=K^*\}$, all target values equal
$x_0<\alpha/C_0\le\alpha K^*/m_T$, so another application of
\eqref{eq:BH_feasible_subset} selects all targets. Thus \(\mathbb P\{\operatorname{TPP}(\widehat{\mathcal R})<1\} \le\mathbb P(\mathcal M_T^c)+\sum_{j=1}^{K^*}q_{j,T}^{\mathcal E} \longrightarrow0\).
This concludes the proof.
\end{proof}

\subsection{Proof of Theorem~\ref{thm:poisplus_BH_comparison}}
\label{app:power:comparison_proof}

For each target, define
\begin{equation*}
\mathcal G_{j,T}=\left\{\max_{b\in\mathcal A_j}V_{k_j^*,b}^{\mathcal O}<\delta V_{k_j^*}^{\mathcal O}\right\}\cap\left(\{j\in\mathcal J_T\}\cup\{|\mathcal A_j|\ge1\}\right).
\end{equation*}
Let \(D_T=\sum_{\substack{j\in[K^*]\\k_j^*\in\widehat{\mathcal S}}}\mathbbm1\{\mathcal G_{j,T}^{c}\}\).
Since $\mathcal A_j=\varnothing$ and $j\notin\mathcal J_T$ whenever $k_j^*\notin\widehat{\mathcal S}$, we have $\mathcal G_{j,T}\subseteq\{k_j^*\in\widehat{\mathcal S}\}$. Thus the definition of $q_{j,T}^{+}$ in Assumption~\ref{ass:odd_nondeflation} gives $q_{j,T}^{+}=\mathbb P(\mathcal G_{j,T}^{c})$ and
\begin{equation}
 \sum_{j=1}^{K^*}q_{j,T}^{+}=\mathbb E D_T+\sum_{j=1}^{K^*}\mathbb P(k_j^*\notin\widehat{\mathcal S}).
 \label{eq:power_failure_count_identity}
\end{equation}

\begin{corollary}[Comparison of the selected nonnull intervals]
\label{cor:poisplus_simultaneous_retention}
Under the assumptions of Theorem~\ref{thm:poisplus_BH_comparison},
if $\mathbb E D_T=o(1)$, then \(\mathbb P\{\widehat{\mathcal R}\cap([K]\setminus\mathcal H_0) \subseteq\widehat{\mathcal R}^{+}\}=1-o(1)\).
Consequently,
$\mathbb P\{\operatorname{TPP}(\widehat{\mathcal R}^{+})
             \ge\operatorname{TPP}(\widehat{\mathcal R})\}=1-o(1)$.
\end{corollary}

\begin{proof}[Proof of Theorem~\ref{thm:poisplus_BH_comparison} and Corollary~\ref{cor:poisplus_simultaneous_retention}]
We bound the weights on $\mathcal G_{j,T}$ and the number of null rejections, then construct a feasible set for weighted BH.
Let $c=2/(1+\delta^\gamma)$. Retain $m_T,x_0,\mathcal M_T$ from the preceding proof, and put
\begin{equation*}
\kappa=1-\alpha(1-C_0^{-1})-c^{-1}>0,\qquad
 \eta=d-C_0x_0/\alpha>0.
\end{equation*}
Assumption~\ref{ass:odd_nondeflation}\textup{(ii)} gives $\kappa>0$.
Since $B+1$ is an integer and $d>C_0/\{\alpha(B+1)\}$,
$\eta$ is bounded away from zero. Let \(\mathcal C_T=\{|\mathcal J_T|\ge dK^*\}\).
Assumption~\ref{ass:odd_nondeflation}\textup{(ii)} gives $\mathbb P(\mathcal C_T^c)=o(1)$ and
\begin{equation}
 \mathbb E D_T\le\sum_{j=1}^{K^*}q_{j,T}^{+}=o(K^*).
 \label{eq:power_bad_weight_count}
\end{equation}

We first relate the score condition to the weights. On $\mathcal G_{j,T}$, nonnegativity of the scores and the convention for an empty
maximum give $v:=V_{k_j^*}^{\mathcal O}>0$.
Every accepted score $x$ lies in $[0,\delta v)$. For $x>0$, apply the
growth condition on $g$ with $t=v/x>1$. For $x=0$, use $g(0)=0$.
Since $g(v)>0$, in either case \({g(x)}/{g(v)}\le\left({x}/{v}\right)^\gamma<\delta^\gamma\). The weight denominator is therefore positive. If $|\mathcal A_j|\ge1$,
division by $g(v)$ in \eqref{eq:post_selection_Evalue} yields
\begin{equation}
 W_{k_j^*}^{\mathcal O}
 =\frac{|\mathcal A_j|+1}
 {1+\sum_{b\in\mathcal A_j}g(V_{k_j^*,b}^{\mathcal O})/g(v)}
 \ge\frac{|\mathcal A_j|+1}{1+|\mathcal A_j|\delta^\gamma}
 \ge\frac{2}{1+\delta^\gamma}=c.
 \label{eq:odd_soft_rank_lower}
\end{equation}
If $\mathcal A_j=\varnothing$, then
$W_{k_j^*}^{\mathcal O}=1$ and $\mathcal G_{j,T}$ implies $j\in\mathcal J_T$.
The strict even-side inequality defining $\mathcal J_T$ and
\eqref{eq:minimal_permutation_event} then give $P_{k_j^*}^{\mathcal E}=x_0$.

We next control the number of small null $p$-values, uniformly over the
BH critical values. Let $\mathcal F_+^{\mathcal O}$ contain
$\mathcal F^{\mathcal O}$ and all odd-side proposal draws. Conditional on
this sigma-field, the reported index set and the weights are fixed, and the
even-side null $p$-values are independent and superuniform by the
argument in the proof of Theorem~\ref{thm:FDR_POIS+}.
Write $m_{0,T}=|\widehat{\mathcal S}\cap\mathcal H_0|$ and \(N_{0,T}(t)=\sum_{k\in\widehat{\mathcal S}\cap\mathcal H_0}\mathbbm1\{P_k^{\mathcal E}\le t\}\).
For $m_T>0$, define
\begin{equation}
 \Delta_T=\max_{1\le r\le m_T}
 \left\{N_{0,T}(\alpha r/m_T)
 -\frac{\alpha m_{0,T}r}{m_T}\right\}_+.
 \label{eq:power_null_count_error}
\end{equation}
When $m_T=0$, set $\Delta_T=0$. If $m_{0,T}=0$, then $\Delta_T=0$.
Otherwise, conditional on $\mathcal F_+^{\mathcal O}$, each threshold
$\alpha r/m_T$ is fixed and
$\mathbb E\{N_{0,T}(\alpha r/m_T)\mid\mathcal F_+^{\mathcal O}\}
\le\alpha m_{0,T}r/m_T$.
Hoeffding's inequality \citep{hoeffding1963probability}, applied to the
independent indicators, and the union bound give, for $x>0$,
\begin{equation*}
 \mathbb P(\Delta_T>x\mid\mathcal F_+^{\mathcal O})
 \le\sum_{r=1}^{m_T}
 \mathbb P\{N_{0,T}(\alpha r/m_T)-\alpha m_{0,T}r/m_T>x
             \mid\mathcal F_+^{\mathcal O}\}
 \le m_Te^{-2x^2/m_{0,T}}.
\end{equation*}
Since $\mathcal M_T$ is $\mathcal F_+^{\mathcal O}$-measurable and
$m_{0,T}\le m_T\le C_0K^*$ on this event, taking expectations yields
\begin{equation}
 \mathbb P(\mathcal M_T\cap\{\Delta_T>x\})
 \le1\wedge\left[C_0K^*\exp\left\{-\frac{2x^2}{C_0K^*}\right\}\right].
 \label{eq:power_null_count_tail}
\end{equation}
For fixed $\varepsilon>0$, taking $x=\varepsilon K^*$ in
\eqref{eq:power_null_count_tail} and using $K^*\to\infty$ gives
$\mathbbm1\{\mathcal M_T\}\Delta_T/K^*\to0$ in probability.
Since $0\le\mathbbm1\{\mathcal M_T\}\Delta_T/K^*\le C_0$, the same
uniform-integrability argument yields
\begin{equation}
 \mathbb E\{\mathbbm1\{\mathcal M_T\}\Delta_T\}/K^*\longrightarrow0.
 \label{eq:power_null_count_mean}
\end{equation}
Work on $\mathcal M_T\cap\mathcal C_T$.
Since \(x_0<\frac{\alpha d}{C_0}\le\frac{\alpha|\mathcal J_T|}{m_T}\),
all targets in $\mathcal J_T$ are rejected by the original BH rule.
Thus \(|\widehat{\mathcal R}|\ge|\mathcal J_T|\ge dK^*,\frac{\alpha |\widehat{\mathcal R}|}{m_T}>x_0\).
First suppose $D_T\le\eta K^*$. Let
$\mathcal K_T=\{k_j^*:j\in\mathcal J_T,\mathcal G_{j,T}\text{ holds}\}$.
Every target in $\mathcal K_T$ has weight at least one and even-side $p$-value
$x_0$. Moreover,
\begin{equation*}
|\mathcal K_T|\ge|\mathcal J_T|-D_T\ge(d-\eta)K^*
       =\frac{C_0K^*x_0}{\alpha}\ge\frac{m_Tx_0}{\alpha}.
\end{equation*}
Consequently, $P_k^+\le x_0\le\alpha|\mathcal K_T|/m_T$ for every $k\in \mathcal K_T$,
so $\mathcal K_T$ is feasible for weighted BH by \eqref{eq:BH_feasible_subset}.

The BH identity and \eqref{eq:power_null_count_error} give
\begin{equation*}
\begin{aligned}
 |\widehat{\mathcal R}\cap\mathcal H_0|
 &=N_{0,T}(\alpha |\widehat{\mathcal R}|/m_T)
 \le\frac{\alpha m_{0,T}|\widehat{\mathcal R}|}{m_T}+\Delta_T\\
 &=\alpha\left(1-\frac{K^*}{m_T}\right)|\widehat{\mathcal R}|+\Delta_T
 \le\alpha(1-C_0^{-1})|\widehat{\mathcal R}|+\Delta_T.
\end{aligned}
\end{equation*}
Here $m_{0,T}=m_T-K^*$ and $m_T\le C_0K^*$ on $\mathcal M_T$.
Since at most $D_T$ original true discoveries fail $\mathcal G_{j,T}$,
the number satisfying this condition is at least
\begin{equation*}
 |\widehat{\mathcal R}\setminus\mathcal H_0|-D_T\ge\{1-\alpha(1-C_0^{-1})\}|\widehat{\mathcal R}|-\Delta_T-D_T.
\end{equation*}
If $\Delta_T+D_T\le\kappa |\widehat{\mathcal R}|$, this lower bound is at least $|\widehat{\mathcal R}|/c$.
For any original true discovery satisfying $\mathcal G_{j,T}$ with a nonempty accepted set,
\eqref{eq:odd_soft_rank_lower} gives \(P_k^+\le\frac{\alpha |\widehat{\mathcal R}|}{cm_T}\).
For any with an empty accepted set, $\mathcal G_{j,T}$ gives
$P_k^+=x_0\le\alpha|\mathcal K_T|/m_T$. Therefore the union of $\mathcal K_T$ and these
discoveries is feasible for weighted BH: its size is at least
both $|\mathcal K_T|$ and $|\widehat{\mathcal R}|/c$, and every member satisfies the corresponding
BH inequality. All these discoveries are selected, so only original
true discoveries failing $\mathcal G_{j,T}$ can be lost. Their number
is at most $D_T$.
If $\Delta_T+D_T>\kappa |\widehat{\mathcal R}|$, the number of lost true discoveries is at most
$|\widehat{\mathcal R}|<(\Delta_T+D_T)/\kappa$.

When $D_T>\eta K^*$, the loss in coverage is at most
$1<D_T/(\eta K^*)$. Combining the cases with
\eqref{eq:coverage_target_equality} and $0<\kappa<1$ gives
\begin{equation}
 \{\operatorname{TPP}(\widehat{\mathcal R})-
                   \operatorname{TPP}(\widehat{\mathcal R}^{+})\}_+
 \le1\wedge\left\{\frac{\Delta_T+D_T}{\kappa K^*}
                         +\frac{D_T}{\eta K^*}\right\}
 \quad\text{on }\mathcal M_T\cap\mathcal C_T.
 \label{eq:power_BH_loss_bound}
\end{equation}

The finite-sample bound \eqref{eq:power_BH_loss_bound} now separates the
null-count fluctuation $\Delta_T$ from the number $D_T$ of targets that
fail the joint score condition. Outside $\mathcal M_T\cap\mathcal C_T$, the loss is at most one. Hence
\begin{align}
 &\mathbb E\bigl[\{\operatorname{TPP}(\widehat{\mathcal R})-
                     \operatorname{TPP}(\widehat{\mathcal R}^{+})\}_+\bigr]
 \nonumber\\
 &\qquad\le\mathbb P(\mathcal M_T^c)+\mathbb P(\mathcal C_T^c)
 +\frac{\mathbb E\{\mathbbm1\{\mathcal M_T\}\Delta_T\}
                  +\mathbb E D_T}{\kappa K^*}
 +\frac{\mathbb E D_T}{\eta K^*}
 \longrightarrow0.
 \label{eq:power_expected_loss_bound}
\end{align}
The bound uses \eqref{eq:power_bad_weight_count}. The convergence follows
from Assumption~\ref{ass:power_main}\textup{(i)}, \eqref{eq:power_null_count_mean}, and
Assumption~\ref{ass:odd_nondeflation}. Since
$\operatorname{TPR}(\widehat{\mathcal R}^{+})-
\operatorname{TPR}(\widehat{\mathcal R})$ is at least the negative of
the expectation in \eqref{eq:power_expected_loss_bound}, the theorem follows.

For Corollary~\ref{cor:poisplus_simultaneous_retention}, the strengthened condition gives
$\mathbb P(D_T>0)\le\mathbb E D_T=o(1)$.
On \(\mathcal M_T\cap\mathcal C_T\cap\{D_T=0\} \cap\{\Delta_T\le\kappa dK^*/2\}\),
we have $D_T\le\eta K^*$, $|\widehat{\mathcal R}|\ge dK^*$, and
$\Delta_T+D_T\le\kappa |\widehat{\mathcal R}|$.
All original true discoveries satisfy $\mathcal G_{j,T}$.
The feasible union constructed above therefore includes every one of
these discoveries.
Using \eqref{eq:power_null_count_tail} at $x=\kappa dK^*/2$ gives
\begin{equation*}
\mathbb P\{\widehat{\mathcal R}\cap([K]\setminus\mathcal H_0)
                         \not\subseteq\widehat{\mathcal R}^{+}\}
 \le\mathbb P(\mathcal M_T^c)+\mathbb P(\mathcal C_T^c)
       +\mathbb E D_T
       +C_0K^*\exp\left\{-\frac{\kappa^2d^2K^*}{2C_0}\right\}
 \longrightarrow0.
\end{equation*}

Selecting every original non-null rejection preserves coverage of every
change-point covered by the original rejection set. This proves the remaining
assertion of the corollary and completes the proof.
\end{proof}

\section{Verification of the score-separation conditions}

We derive local two-sample bounds for Condition~\textup{(A2)} and the joint score condition \textup{(B2)}, and then bound the empty-acceptance probabilities.
\label{app:poisplus_power}

We verify the score assumptions for bounded-kernel two-sample scores and
for bounded, sub-Gaussian, and finite-moment absolute linear contrasts.
The labels $U$, $L$, $SG$, and $L^p$ in the displays below refer to
these four cases, respectively. Related change-point statistics include the MMD-based scan $B$-statistic
\citep{li2019scan}, cumulative sum (CUSUM) tests based on bounded transformations
\citep{duerre2019robust}, normalized mean CUSUMs under sub-Gaussian errors
\citep{wang2020univariate}, and MOSUM/CUSUM-based segmentation allowing
heavy tails and finite-moment conditions \citep{cho2022two}.
Throughout, $c,C>0$ denote constants that may change from one occurrence
to the next. Unless a dependence is stated, they are independent of the
target, side sizes, split point, and numbers of permutations. Constants with
subscript $p$ may depend on the fixed moment order.

\subsection{Local score bounds}
\label{app:power:local_scores}

The following local calculations allow the observations within either
nominal side of a split to have different laws. This is needed later
when an estimated split does not coincide exactly with a true change.
Let $Z_1,\ldots,Z_n$ be independent observations, and
partition their indices into two fixed groups of sizes $n_1,n_2\ge1$.
Write $s=(n_1n_2/n)^{1/2}$. For a signed linear contrast, let $c_i=s/n_1$
on the first group and $c_i=-s/n_2$ on the second. Then
\begin{equation}
 \sum_i c_i=0,\qquad \sum_i c_i^2=1,\qquad
 \sum_i|c_i|^p\le1\quad(p\ge2).
 \label{eq:contrast_coefficient_norms}
\end{equation}
The last relation follows from $|c_i|\le1$. Let $\pi$ be a uniform
permutation of $[n]$, independent of the observations. The functions $h$
and $\psi$ below are fixed. The same bounds hold conditionally on any
sigma-field under which these independence, boundedness, or moment
conditions hold with the stated constants.

We will also use a finite-population version of Hoeffding's inequality.
For fixed $z_1,\ldots,z_n\in[-M,M]$, let $\bar z_1^\pi,\bar z_2^\pi$
be their averages after a uniform random partition of sizes $n_1,n_2$.
Interchange the groups if necessary so that $n_1\le n_2$. If $A$ is the
first group and $\bar z=n^{-1}\sum_i z_i$, then
$s(\bar z_1^\pi-\bar z_2^\pi)=sn(\bar z_A-\bar z)/n_2$.
By the sampling-without-replacement inequality of
\citet{hoeffding1963probability}, for every $t>0$,
\begin{equation}
 \mathbb P_\pi\{|s(\bar z_1^\pi-\bar z_2^\pi)|\ge t\} \le2\exp\left\{-\frac{2n_1}{(2M)^2} \left(\frac{tn_2}{sn}\right)^2\right\} =2\exp\left\{-\frac{t^2n_2}{2M^2n}\right\} \le2\exp\left\{-\frac{t^2}{4M^2}\right\}.
 \label{eq:power_finite_population_contrast}
\end{equation}
The equality uses $s^2=n_1n_2/n$, and the last inequality uses $n_2/n\ge1/2$.
In particular, the bound is uniform over the two side sizes.

\subsubsection{Bounded-kernel two-sample scores}
\label{sec:appendix-ustat-route}

Let $h$ be a symmetric positive semidefinite kernel with $|h|\le M_h$
for $M_h>0$. For disjoint groups
$J_1,J_2$ of sizes $n_1,n_2\ge2$, define the two-sample $U$-statistic
\begin{equation}
 D_h^{\rm u}(J_1,J_2) ={}\binom{n_1}{2}^{-1}\sum_{i<i'\in J_1}h(Z_i,Z_{i'}) +\binom{n_2}{2}^{-1}\sum_{i<i'\in J_2}h(Z_i,Z_{i'}) -\frac{2}{n_1n_2}\sum_{i\in J_1}\sum_{i'\in J_2}h(Z_i,Z_{i'}).
 \label{eq:U_stat_discrepancy_score}
\end{equation}
Let $\phi$ be a feature map with
$h(z,z')=\langle\phi(z),\phi(z')\rangle$, and write
$\bar\phi_a=n_a^{-1}\sum_{i\in J_a}\phi(Z_i)$ for $a=1,2$.
The squared difference of the empirical kernel means is
$D_h^{\rm b}=\|\bar\phi_1-\bar\phi_2\|^2$.
Throughout, $D_h$ denotes either $D_h^{\rm u}$ or $D_h^{\rm b}$,
with the choice fixed as part of the scoring rule and unchanged under
permutation. We use the score $s\sqrt{|D_h|}$, and let $D_h^\pi$
be the same statistic after permuting the pooled observations.
Write $\bar F_1,\bar F_2$ for the average marginal laws on the two groups,
and put
\begin{equation*}
\Theta_h(F,G)=\mathbb E_{F\otimes F}h+\mathbb E_{G\otimes G}h
                         -2\mathbb E_{F\otimes G}h,
 \qquad \theta=\Theta_h(\bar F_1,\bar F_2).
\end{equation*}

\begin{lemma}\label{lem:ustat_validation_separation}
Suppose $\theta>0$, and let $\mu=s\sqrt{\theta}$ and
$\eta=\mu/\sqrt{M_h}$. There are universal constants $c,C,c_0>0$ such that
\begin{align}
 \mathbb P(s\sqrt{|D_h|}\le\mu/2)
 &\le\mathbbm1\{\eta<c_0\}+C\exp(-c\eta^2),\nonumber\\
 \mathbb P_\pi(s\sqrt{|D_h^\pi|}\ge t\mid Z_1,\ldots,Z_n)
 &\le C\exp(-ct^2/M_h),\quad t>0.
 \label{eq:U_local_tails}
\end{align}
The bounds are uniform over the two side sizes and both choices of
$D_h$. The permutation bound does not require $\theta>0$.
\end{lemma}

\begin{proof}[Proof of Lemma~\ref{lem:ustat_validation_separation}]
Since $\|\phi(Z_i)\|^2=h(Z_i,Z_i)\le M_h$,
the empirical variances
$\widehat v_a=n_a^{-1}\sum_{i\in J_a}\|\phi(Z_i)-\bar\phi_a\|^2$
satisfy $0\le\widehat v_a\le M_h$. Expanding the within-group sums in
\eqref{eq:U_stat_discrepancy_score} gives
\begin{equation*}
 0\le D_h^{\rm b}-D_h^{\rm u}
 =\frac{\widehat v_1}{n_1-1}+\frac{\widehat v_2}{n_2-1}
 \le M_h\{(n_1-1)^{-1}+(n_2-1)^{-1}\}
 \le\frac{2M_h}{s^2}.
\end{equation*}
Write $R=\|(\bar\phi_1-\bar\phi_2)-\mathbb E(\bar\phi_1-\bar\phi_2)\|$.
Independence gives $\mathbb E R^2\le M_h/s^2$, and hence
$\mathbb E R\le\sqrt{M_h}/s$. Changing one observation in group $a$
changes $R$ by at most $2\sqrt{M_h}/n_a$, so the squared
bounded-difference constants sum to at most $4M_h/s^2$.
The bounded-differences inequality \citep{mcdiarmid1989method} therefore gives
\(\mathbb P\{R\ge\sqrt{M_h}/s+x\}\le\exp\{-s^2x^2/(2M_h)\}\) for $x>0$.
Since $\|\mathbb E(\bar\phi_1-\bar\phi_2)\|^2=\theta$, on
$\{R<\sqrt{\theta}/4\}$ with $\eta\ge8$ we have
\begin{equation*}
 D_h\ge D_h^{\rm u}\ge(\sqrt{\theta}-R)^2-\frac{2M_h}{s^2}
 >\frac{9\theta}{16}-\frac{\theta}{32}
 =\frac{17\theta}{32}>\frac{\theta}{4}.
\end{equation*}
Also, $\sqrt{M_h}/s\le\sqrt{\theta}/8$ when $\eta\ge8$. Thus
\begin{equation*}
 \mathbb P(s\sqrt{|D_h|}\le\mu/2)
 =\mathbb P(|D_h|\le\theta/4)
 \le\mathbb P(R\ge\sqrt{\theta}/4)
 \le\exp(-\eta^2/128).
\end{equation*}
We may therefore take $c_0=8$. For $\eta<c_0$, the indicator in
\eqref{eq:U_local_tails} gives the trivial bound.

For the permutation bound, condition on $Z_1,\ldots,Z_n$ and,
by symmetry, suppose $n_1\le n_2$. The statistic in
\eqref{eq:U_stat_discrepancy_score} is the average over ordered distinct
pairs from both groups of
$h(x_1,x_2)+h(y_1,y_2)-h(x_1,y_2)-h(x_2,y_1)$.
Let $\Pi_n$ be the set of permutations of $[n]$, and define
\begin{equation*}
\Sigma^2=\frac1{n_1^2(n_1-1)^2}
     \sup_{\pi\in\Pi_n}\sum_{1\le i\ne i'\le n_1}
                  h^2(Z_{\pi(i)},Z_{\pi(i')})
 \le\frac{M_h^2}{n_1(n_1-1)}\le\frac{2M_h^2}{n_1^2}.
\end{equation*}
By \citet[Theorem~6.1]{kim2022minimax}, with bivariate kernel $h$,
for $\Sigma>0$ and $x>0$ we have
\begin{equation}
 \mathbb P_\pi\{(D_h^{\rm u})^\pi\ge x\mid Z_1,\ldots,Z_n\}
 \le\exp\{-c\min(x^2/\Sigma^2,x/\Sigma)\},
 \qquad \Sigma\le\sqrt2M_h/n_1.
 \label{eq:power_U_permutation_concentration}
\end{equation}
If $\Sigma=0$, every off-diagonal pooled kernel value is zero and the
probability in \eqref{eq:power_U_permutation_concentration} is zero
at a positive threshold.

The bound on $D_h^{\rm b}-D_h^{\rm u}$ holds after every permutation.
Since $D_h^{\rm b}\ge0$,
we have $(D_h^{\rm u})^\pi\ge-2M_h/s^2$ and
$0\le D_h^\pi-(D_h^{\rm u})^\pi\le2M_h/s^2$.
Consequently, when $t^2\ge4M_h$, the event
$\{s\sqrt{|D_h^\pi|}\ge t\}$ implies
$(D_h^{\rm u})^\pi\ge(t^2-2M_h)/s^2\ge t^2/(2s^2)$
for either choice of $D_h$. For $\Sigma>0$, substitute
$x=t^2/(2s^2)$ in \eqref{eq:power_U_permutation_concentration}
and use $n_1\ge s^2$ to obtain
\begin{equation*}
\begin{aligned}
 \mathbb P_\pi(s\sqrt{|D_h^\pi|}\ge t\mid Z_1,\ldots,Z_n)
 &\le\exp\left[-c\min\left\{
       \left(\frac{n_1t^2}{s^2M_h}\right)^2,
       \frac{n_1t^2}{s^2M_h}\right\}\right]\\
 &\le\exp\{-c\min(t^4/M_h^2,t^2/M_h)\}
 \le\exp(-ct^2/M_h).
\end{aligned}
\end{equation*}
The last inequality uses $t^2\ge4M_h$. When $\Sigma=0$, the same
probability is zero for $t^2\ge4M_h$. Increasing $C$ covers
$t^2<4M_h$ and completes the proof.
\end{proof}

\subsubsection{Bounded absolute linear contrasts}
\label{sec:appendix-linear-contrast}

Set $Y=\sum_i c_i\psi(Z_i)$ and $Y^\pi=\sum_i c_i\psi(Z_{\pi(i)})$.
If $\theta=\mathbb E_{\bar F_1}\psi-\mathbb E_{\bar F_2}\psi$, then
$\mathbb E Y=s\theta$. We use $|Y|$ as the score and write
$\mu=s|\theta|$.

\begin{lemma}\label{lem:linear_validation_separation}
Suppose $\|\psi\|_\infty\le M_\psi$ for $M_\psi>0$, and $\mu>0$. Then \(\mathbb P(|Y|\le\mu/2)\le2\exp\{-\mu^2/(8M_\psi^2)\}\), and \(\mathbb P_\pi(|Y^\pi|\ge t\mid Z_1,\ldots,Z_n) \le2\exp\{-t^2/(4M_\psi^2)\},\qquad t>0\).
\end{lemma}

\begin{proof}[Proof of Lemma~\ref{lem:linear_validation_separation}]
For the observed score, a small absolute contrast entails a large
deviation from its mean. For the permuted score, we instead use the
conditional finite-population bound. Since $|\mathbb E Y|=\mu$, the event $|Y|\le\mu/2$ implies
$|Y-\mathbb E Y|\ge\mu/2$. Hoeffding's inequality
\citep{hoeffding1963probability} gives
\begin{equation*}
\mathbb P(|Y|\le\mu/2)
 \le\mathbb P(|Y-\mathbb E Y|\ge\mu/2)
 \le2\exp\left\{-\frac{2(\mu/2)^2}{\sum_i(2M_\psi|c_i|)^2}\right\}
 =2\exp\{-\mu^2/(8M_\psi^2)\},
\end{equation*}
where the equality uses \eqref{eq:contrast_coefficient_norms}.
Conditional on the pooled observations, their transformed values are a
fixed finite population in $[-M_\psi,M_\psi]$.
Applying \eqref{eq:power_finite_population_contrast} to that population
proves the second bound.
\end{proof}

\subsubsection{Sub-Gaussian absolute linear contrasts}
\label{sec:appendix-subgaussian-linear-contrast}

For the next two lemmas, suppose the means of $\psi(Z_i)$ take two values.
Write
\begin{equation}
 \psi(Z_i)=m_++\Delta d_i+\varepsilon_i,\qquad
 d_i\in\{0,1\},\quad \mathbb E\varepsilon_i=0,
 \label{eq:two_mean_representation}
\end{equation}
where the labels $d_i$ are fixed. This representation permits different
noise distributions, provided their centered moments satisfy the stated
bounds. Positive $\mu=s|\theta|$ implies $\Delta\ne0$.

\begin{lemma}\label{lem:subgaussian_validation_separation}
Suppose \eqref{eq:two_mean_representation} holds, $\mu>0$, and
$\mathbb E e^{\lambda\varepsilon_i}\le e^{\lambda^2\sigma_\psi^2/2}$
for every $\lambda\in\mathbb R$ and every $i$, with $\sigma_\psi>0$.
Then, for every $t>0$, \(\mathbb P(|Y|\le\mu/2)\le2e^{-\mu^2/(8\sigma_\psi^2)}\), and \(\mathbb P(|Y^\pi|\ge t) \le2e^{-t^2/(16\Delta^2)}+2e^{-t^2/(8\sigma_\psi^2)}\).
The second probability is taken over both the observations and the permutation.
\end{lemma}

\begin{proof}[Proof of Lemma~\ref{lem:subgaussian_validation_separation}]
We first bound the centered observed contrast. For the permutation
contrast, the means themselves are permuted, so we will separate the
random allocation of the two mean levels from the centered errors.
By independence and \eqref{eq:contrast_coefficient_norms},
\begin{equation*}
\mathbb E\exp\{\lambda(Y-\mathbb E Y)\}
 =\prod_i\mathbb E e^{\lambda c_i\varepsilon_i}
 \le\exp\left\{\frac{\lambda^2\sigma_\psi^2}{2}\sum_i c_i^2\right\}
 =e^{\lambda^2\sigma_\psi^2/2}.
\end{equation*}
The exponential Markov inequality, optimized over positive and negative
$\lambda$, therefore gives
$\mathbb P(|Y-\mathbb E Y|\ge x)\le2e^{-x^2/(2\sigma_\psi^2)}$.
Use $\{|Y|\le\mu/2\}\subseteq\{|Y-\mathbb E Y|\ge\mu/2\}$ for the
first assertion.

For the second, the common mean in \eqref{eq:two_mean_representation}
cancels because $\sum_i c_i=0$. Thus
\begin{equation}
 Y^\pi=\Delta H^\pi+R^\pi,\qquad
 H^\pi=\sum_i c_i d_{\pi(i)},\quad
 R^\pi=\sum_i c_i\varepsilon_{\pi(i)}.
 \label{eq:power_label_noise_decomposition}
\end{equation}
The labels are a finite population in $[0,1]$, so using
\eqref{eq:power_finite_population_contrast} with $M=1$ gives
$\mathbb P(|\Delta H^\pi|\ge t/2)\le2e^{-t^2/(16\Delta^2)}$.
Conditional on $\pi$, the centered variables
$\varepsilon_{\pi(i)}$ remain independent. The preceding
moment-generating-function calculation applies to their contrast, since
its squared coefficient norm is one. Hence
\begin{equation*}
 \mathbb P(|Y^\pi|\ge t)
 \le\mathbb P(|\Delta H^\pi|\ge t/2)
       +\mathbb E_\pi\{\mathbb P(|R^\pi|\ge t/2\mid\pi)\}
 \le2e^{-t^2/(16\Delta^2)}+2e^{-t^2/(8\sigma_\psi^2)}.
\end{equation*}
The first inequality follows from \eqref{eq:power_label_noise_decomposition}
and the union bound. Notice that the independent-sum calculation conditions
on $\pi$, not on the realized pooled observations. This completes the proof.
\end{proof}

\subsubsection{Finite-moment absolute linear contrasts}
\label{sec:appendix-finite-moment-linear-even}

\begin{lemma}\label{lem:finite_moment_linear_validation_separation}
Suppose \eqref{eq:two_mean_representation} holds, $\mu>0$, and
$\sup_i\mathbb E|\varepsilon_i|^p\le\nu_p^p$ for a fixed $p>2$ and
$\nu_p>0$. There is $C_p>0$, depending only on $p$, such that \(\mathbb P(|Y|\le\mu/2)\le C_p(\nu_p/\mu)^p\), and \(\mathbb P(|Y^\pi|\ge t) \le2e^{-t^2/(16\Delta^2)}+C_p(\nu_p/t)^p, \qquad t>0\).
\end{lemma}

\begin{proof}[Proof of Lemma~\ref{lem:finite_moment_linear_validation_separation}]
The decomposition \eqref{eq:power_label_noise_decomposition} remains
valid. We replace the exponential bound for its centered-error term by
a moment bound, keeping the finite-population bound for the label term.
Rosenthal's inequality \citep{rosenthal1970subspaces} and
\eqref{eq:contrast_coefficient_norms} give
\begin{equation}
 \begin{aligned}
 \mathbb E\left|\sum_i c_i\varepsilon_i\right|^p
 &\le C_p\left\{\left(\sum_i c_i^2\mathbb E\varepsilon_i^2\right)^{p/2}
       +\sum_i|c_i|^p\mathbb E|\varepsilon_i|^p\right\}\\
 &\le C_p\nu_p^p\left\{\left(\sum_i c_i^2\right)^{p/2}
       +\sum_i|c_i|^p\right\}\le C_p'\nu_p^p.
 \end{aligned}
 \label{eq:power_Rosenthal_bound}
\end{equation}
Here $\mathbb E\varepsilon_i^2\le\nu_p^2$ follows from the $p$th-moment
assumption. Since $|Y|\le\mu/2$ implies a centered deviation of at least
$\mu/2$, Markov's inequality proves the observed-score bound.

For a permutation, use \eqref{eq:power_label_noise_decomposition}.
The label term has the same bound as in the preceding proof.
Conditional on $\pi$, the coefficient vector of $R^\pi$ satisfies the
same norm identities. Applying \eqref{eq:power_Rosenthal_bound} and
Markov's inequality conditionally, and then taking expectations, yields
\begin{equation*}
\begin{aligned}
\mathbb P(|Y^\pi|\ge t)
 &\le\mathbb P(|\Delta H^\pi|\ge t/2)
      +\mathbb E_\pi\{\mathbb P(|R^\pi|\ge t/2\mid\pi)\}\\
 &\le2e^{-t^2/(16\Delta^2)}
      +(2/t)^p\mathbb E_\pi\{\mathbb E(|R^\pi|^p\mid\pi)\}\\
 &\le2e^{-t^2/(16\Delta^2)}+C_p(\nu_p/t)^p.
\end{aligned}
\end{equation*}
This proves the second assertion.
\end{proof}

\subsection{Even-side score separation}
\label{sec:appendix-score-details}

We first evaluate the even-side score at a split fixed conditional on
$\mathcal F^{\mathcal O}$, such as the odd-side estimated split,
and apply the local bounds under this conditioning.
This fixes the split and any odd-side fitted transformation. We apply
the boundedness or centered moment conditions of the corresponding local
bound under this conditioning, with common constants over the targets.
For \texttt{POIS+}, Appendix~\ref{sec:odd_nondeflation_score_conditions}
uses a scan score recomputed on each odd-side permutation.
Write $I_j^{*,\mathcal E}=I_j^*\cap\mathcal E$ and split this set at the
odd-side estimate $\widehat\tau_{k_j^*}^{\mathcal O}$ into
$I_{j,-}^{*,\mathcal E}$ and $I_{j,+}^{*,\mathcal E}$, with sizes
$n_{1,j}^{\mathcal E}$ and $n_{2,j}^{\mathcal E}$. Let
\begin{equation}
 s_{j,T}^{\mathcal E}
 =\left\{\frac{n_{1,j}^{\mathcal E}n_{2,j}^{\mathcal E}}
                   {n_{1,j}^{\mathcal E}+n_{2,j}^{\mathcal E}}\right\}^{1/2}.
 \label{eq:power_even_scale}
\end{equation}
On the event that all targets are reported, disjointness implies that
each contains only its designated change-point, with adjacent laws
$F_{j,-},F_{j,+}$. For deterministic targets, \textup{(A1)} therefore
ensures this geometry for all sufficiently large $T$. For a
data-independent random pool, apply the bounds on the pool event where
this geometry holds and set the error bound to one elsewhere before
averaging.
Let $\zeta_{j,-}^{\mathcal E}$ be the fraction of post-change observations
on the nominal left, and $\zeta_{j,+}^{\mathcal E}$ the fraction of
pre-change observations on the nominal right. Their average laws are
\begin{equation*}
\begin{aligned}
\widetilde F_{j,-}^{\mathcal E} &=(1-\zeta_{j,-}^{\mathcal E})F_{j,-} +\zeta_{j,-}^{\mathcal E}F_{j,+},\\
\widetilde F_{j,+}^{\mathcal E} &=\zeta_{j,+}^{\mathcal E}F_{j,-} +(1-\zeta_{j,+}^{\mathcal E})F_{j,+}.
\end{aligned}
\end{equation*}
These are the average marginal distributions on the two sides.
Subtracting the two mixture laws gives \(\widetilde F_{j,-}^{\mathcal E}-\widetilde F_{j,+}^{\mathcal E} =(1-\zeta_{j,-}^{\mathcal E}-\zeta_{j,+}^{\mathcal E}) (F_{j,-}-F_{j,+})\).
Thus the linear mean difference is multiplied by
$1-\zeta_{j,-}^{\mathcal E}-\zeta_{j,+}^{\mathcal E}$, and the
population kernel discrepancy is multiplied by its square. The latter follows because
$\Theta_h$ is quadratic in the difference of the two laws. This identifies precisely how localization enters the
signals used below.

For the kernel score, use
$V_{k_j^*}^{\mathcal E}=s_{j,T}^{\mathcal E}
\sqrt{|D_h(I_{j,-}^{*,\mathcal E},I_{j,+}^{*,\mathcal E})|}$ and define
\begin{equation*}
\theta_{j,T}^{U,\mathcal E}
  =\Theta_h(\widetilde F_{j,-}^{\mathcal E},\widetilde F_{j,+}^{\mathcal E}),
 \qquad
 \eta_{j,T}^{U,\mathcal E}=s_{j,T}^{\mathcal E}\sqrt{\theta_{j,T}^{U,\mathcal E}/M_h}.
\end{equation*}
For a linear contrast, use the absolute value of the normalized difference
of the two $\psi$-averages and write
\begin{equation*}
\theta_{j,T}^{L,\mathcal E}
   =\mathbb E_{\widetilde F_{j,-}^{\mathcal E}}\psi
                  -\mathbb E_{\widetilde F_{j,+}^{\mathcal E}}\psi,
 \qquad
 \Delta_{j,T}^{L,\mathcal E}
   =\mathbb E_{F_{j,-}}\psi-\mathbb E_{F_{j,+}}\psi.
\end{equation*}
When the denominator is nonzero, let
$\rho_{j,T}^{\mathcal E}=|\theta_{j,T}^{L,\mathcal E}|/
 |\Delta_{j,T}^{L,\mathcal E}|
 =|1-\zeta_{j,-}^{\mathcal E}-\zeta_{j,+}^{\mathcal E}|$.
Otherwise, set it to zero. The effective signals for the three linear
families are
\begin{equation*}
\begin{aligned}
\eta_{j,T}^{L,\mathcal E} &=\frac{s_{j,T}^{\mathcal E}|\theta_{j,T}^{L,\mathcal E}|}{2M_\psi},\\
\overline\eta_{j,T}^{SG,\mathcal E} &=s_{j,T}^{\mathcal E}\rho_{j,T}^{\mathcal E} \wedge\frac{s_{j,T}^{\mathcal E}|\theta_{j,T}^{L,\mathcal E}|}{\sigma_\psi},\\
\eta_{j,T}^{L,p,\mathcal E} &=\frac{s_{j,T}^{\mathcal E}|\theta_{j,T}^{L,\mathcal E}|}{\nu_p}.
\end{aligned}
\end{equation*}
Only the quantity relevant to a given score family is used.
For the mean CUSUM with $\psi(x)=x$, the factor $s_{j,T}^{\mathcal E}|\theta_{j,T}^{L,\mathcal E}|$ is an effective jump size multiplied by the square root of the effective sample size, paralleling the signal-strength quantities in \citet[Assumption~3.3]{fryzlewicz2014wild} and \citet[Theorem~1]{baranowski2019narrowest}.

\begin{proposition}\label{prop:even_score_verification}
For the $U$, $L$, $SG$, and $L^p$ score families, assume, conditional on
$\mathcal F^{\mathcal O}$, the conditions of
Lemmas~\ref{lem:ustat_validation_separation},
\ref{lem:linear_validation_separation},
\ref{lem:subgaussian_validation_separation}, and
\ref{lem:finite_moment_linear_validation_separation}, respectively.
Then the following bounds hold:
\begin{align}
 U:\quad q_{j,T}^{\mathcal E}
 &\le\mathbb E\left[1\wedge\left\{
   \mathbbm1\{\eta_{j,T}^{U,\mathcal E}<c_0\}
   +C(B+1)e^{-c(\eta_{j,T}^{U,\mathcal E})^2}\right\}\right],
 \label{eq:U_basic_signal_condition}\\
 L:\quad q_{j,T}^{\mathcal E}
 &\le\mathbb E\left[1\wedge C(B+1)e^{-c(\eta_{j,T}^{L,\mathcal E})^2}\right],
 \label{eq:L_basic_signal_condition}\\
 SG:\quad q_{j,T}^{\mathcal E}
 &\le\mathbb E\left[1\wedge C(B+1)e^{-c(\overline\eta_{j,T}^{SG,\mathcal E})^2}\right],
 \label{eq:SG_basic_signal_condition}\\
 L^p:\quad q_{j,T}^{\mathcal E}
 &\le\mathbb E\left[1\wedge\left\{
    CB e^{-c(s_{j,T}^{\mathcal E}\rho_{j,T}^{\mathcal E})^2}
   +C_p(B+1)(\eta_{j,T}^{L,p,\mathcal E})^{-p}\right\}\right].
 \label{eq:Lp_basic_signal_condition}
\end{align}
A conditional bound is set to one when a required side size or positive
signal condition fails. The average of the relevant right-hand sides
tending to zero implies \textup{(A2)}, and their sum tending to zero implies
its strengthened version.
\end{proposition}

\begin{proof}[Proof of Proposition~\ref{prop:even_score_verification}]
Fix a target and condition on $\mathcal F^{\mathcal O}$. We first
reduce a failure of score separation to one observed-score lower tail
and $B$ permutation upper tails, then substitute the appropriate local
bounds separately for the four score families. On an admissible target, put
$\mu=s_{j,T}^{\mathcal E}\sqrt{\theta_{j,T}^{U,\mathcal E}}>0$ for the
kernel score and $\mu=s_{j,T}^{\mathcal E}|\theta_{j,T}^{L,\mathcal E}|>0$
for a linear contrast. For this realized value of $\mu$,
\begin{equation*}
\left\{\max_{b\in[B]}V_{k_j^*,b}^{\mathcal E}
                         \ge V_{k_j^*}^{\mathcal E}\right\}
 \subseteq\{V_{k_j^*}^{\mathcal E}\le\mu/2\}
       \cup\bigcup_{b=1}^B\{V_{k_j^*,b}^{\mathcal E}\ge\mu/2\}.
\end{equation*}
Taking conditional probabilities and then expectations, we obtain
\begin{equation}
 q_{j,T}^{\mathcal E}\le\mathbb E\left[1\wedge\left\{
   \mathbb P(V_{k_j^*}^{\mathcal E}\le\mu/2\mid\mathcal F^{\mathcal O})
   +\sum_{b=1}^B\mathbb P(V_{k_j^*,b}^{\mathcal E}\ge\mu/2
                                      \mid\mathcal F^{\mathcal O})
                      \right\}\right].
 \label{eq:power_even_reduction}
\end{equation}

For the kernel score, apply Lemma~\ref{lem:ustat_validation_separation}
at $t=\mu/2$. The
conditional permutation bound holds for every pooled sample, so it also
holds after averaging over that sample. This gives
\eqref{eq:U_basic_signal_condition}. For bounded linear contrasts, the observed bound is
$2\exp\{-\mu^2/(8M_\psi^2)\}$, and each conditional permutation
bound at $\mu/2$ is $2\exp\{-\mu^2/(16M_\psi^2)\}$.
Substituting these two bounds from
Lemma~\ref{lem:linear_validation_separation} gives
\eqref{eq:L_basic_signal_condition}, since
$\mu/M_\psi=2\eta_{j,T}^{L,\mathcal E}$.

For the two unbounded families, \eqref{eq:two_mean_representation} holds
with $d_i$ the pre-change labels. The permutation part of
Lemma~\ref{lem:subgaussian_validation_separation}, applied conditionally
on $\mathcal F^{\mathcal O}$ and then averaged over the independent
permutation, gives a bound of the form \(Ce^{-c(\mu/|\Delta_{j,T}^{L,\mathcal E}|)^2} +Ce^{-c(\mu/\sigma_\psi)^2} \le C'e^{-c'(\overline\eta_{j,T}^{SG,\mathcal E})^2}\).
Here $\mu/|\Delta_{j,T}^{L,\mathcal E}|=
 s_{j,T}^{\mathcal E}\rho_{j,T}^{\mathcal E}$.
The observed bound is no larger after a change of constants. Substitution
in \eqref{eq:power_even_reduction} proves
\eqref{eq:SG_basic_signal_condition}. Finally,
Lemma~\ref{lem:finite_moment_linear_validation_separation} bounds the
observed error by $C_p(\nu_p/\mu)^p$ and the permutation error by
$Ce^{-c(s_{j,T}^{\mathcal E}\rho_{j,T}^{\mathcal E})^2}
 +C_p(\nu_p/\mu)^p$. The observed term occurs once and the polynomial permutation term
occurs $B$ times. The label-allocation term occurs only for the $B$
permutations. This gives the factors $B+1$ and $B$ in
\eqref{eq:Lp_basic_signal_condition}.

For inadmissible side sizes or nonpositive signal, use the trivial
conditional probability bound one in place of the corresponding display.
The last assertions now follow by averaging or summing the bounds.
\end{proof}

We next extend these bounds to scan scores on the even side.
Suppose $v_k$ scans a prespecified grid on each subsequence, using a
fixed kernel or transformation. Define the even-side grid
$\mathcal T_j^{*,\mathcal E}$, observed splitwise score
$S_j^{\mathcal E}(u)$, and permuted scores $S_{j,b}^{\mathcal E}(u)$ as in \eqref{eq:odd_scan_definition} below,
with the odd-side superscript replaced by $\mathcal E$, and set
$L_j^{\mathcal E}=|\mathcal T_j^{*,\mathcal E}|$.
Choose an oracle split $u_j^\circ$ for the analysis and let
$\mu_{j,T}^{\mathcal E}>0$ be its signal, defined as in
\eqref{eq:odd_oracle_signal}. Conditioning first on $\mathcal F^{\mathcal O}$
and then taking expectations gives
\begin{equation*}
q_{j,T}^{\mathcal E}
\le \mathbb E\bigg[1\wedge\bigg\{
\mathbb P\left(S_j^{\mathcal E}(u_j^\circ)
       \le\mu_{j,T}^{\mathcal E}/2\mid\mathcal F^{\mathcal O}\right)+\sum_{b=1}^{B}\sum_{u\in\mathcal T_j^{*,\mathcal E}}
\mathbb P\left(S_{j,b}^{\mathcal E}(u)
       \ge\mu_{j,T}^{\mathcal E}/2\mid\mathcal F^{\mathcal O}\right)
\bigg\}\bigg].
\end{equation*}
Indeed, outside these events the observed maximum exceeds $\mu_{j,T}^{\mathcal E}/2$ and every permuted maximum is below this value. The permutation bounds are uniform over the grid by the local lemmas.
Consequently, the four bounds in
\eqref{eq:U_basic_signal_condition}--\eqref{eq:Lp_basic_signal_condition}
apply with $B$ replaced by $BL_j^{\mathcal E}$, and with each signal
evaluated at the oracle split. The observed-score term still occurs once.
The same admissibility and conditional moment requirements apply.

\subsection{Odd-side scans and signal conditions}
\label{sec:odd_nondeflation_score_conditions}
\label{app:power:g_compatibility}

Here we take $v_k$ to be a scan score on the odd side. Since acceptance depends on the
same odd observations, we control every proposed permutation score
before restricting to accepted proposals. Take a
fixed nonempty split grid
$\mathcal T_j^{*,\mathcal O}$ on $I_j^{*,\mathcal O}=I_j^*\cap\mathcal O$,
and let $L_j^{\mathcal O}=|\mathcal T_j^{*,\mathcal O}|$. Write
\begin{equation}
\begin{aligned}
 V_{k_j^*}^{\mathcal O}
   &=\max_{u\in\mathcal T_j^{*,\mathcal O}}S_j^{\mathcal O}(u),\\
 V_{k_j^*,b}^{\mathcal O}
   &=\max_{u\in\mathcal T_j^{*,\mathcal O}}S_{j,b}^{\mathcal O}(u),
 \qquad b=1,\ldots,B'.
\end{aligned}
 \label{eq:odd_scan_definition}
\end{equation}
Here $S_j^{\mathcal O}(u)$ is the observed score at split $u$, and
$S_{j,b}^{\mathcal O}(u)$ is the score after the $b$th local permutation.
The grid and the scoring function are unchanged under permutation.
For a split $u$, write $n_{1,j}^{\mathcal O}(u),n_{2,j}^{\mathcal O}(u)$
for the side sizes, $s_{j,T}^{\mathcal O}(u)$ for the scale in
\eqref{eq:power_even_scale}, and $\widetilde F_{j,-}^{\mathcal O}(u),
\widetilde F_{j,+}^{\mathcal O}(u)$ for the average laws on the two sides.
For the analysis, choose an oracle split
$u_j^\circ\in\mathcal T_j^{*,\mathcal O}$.

For linear scores, $S_j^{\mathcal O}(u)$ is the absolute value of
the normalized difference of the two observed $\psi$-averages.
The permuted score $S_{j,b}^{\mathcal O}(u)$ uses the same formula on the permuted data. Set
$\theta_{j,T}^{L,\mathcal O}(u)=
\mathbb E_{\widetilde F_{j,-}^{\mathcal O}(u)}\psi-
\mathbb E_{\widetilde F_{j,+}^{\mathcal O}(u)}\psi$.
For kernel scores, use
$S_j^{\mathcal O}(u)=s_{j,T}^{\mathcal O}(u)
 \sqrt{|D_j^{U,\mathcal O}(u)|}$, where $D_j^{U,\mathcal O}(u)$ is
the chosen statistic $D_h$ evaluated at split $u$.
The same formula on the permuted data defines $S_{j,b}^{\mathcal O}(u)$, and
$\theta_{j,T}^{U,\mathcal O}(u)=
\Theta_h(\widetilde F_{j,-}^{\mathcal O}(u),\widetilde F_{j,+}^{\mathcal O}(u))$.
All linear splits have nonempty sides, and all kernel-score splits
have at least two observations per side.

Local proposal draws are independent of the data for every fixed target. The target, grid, oracle split, and functions $h,\psi$ are fixed independently of the odd-side observations. For independently trained functions, the bounds below hold conditional on the training data when their stated conditions hold conditionally, and unconditional bounds follow by averaging.

Define the oracle signal by
\begin{equation}
 \mu_{j,T}^{\mathcal O}
 =\begin{cases}
 s_{j,T}^{\mathcal O}(u_j^\circ)\sqrt{\theta_{j,T}^{U,\mathcal O}(u_j^\circ)},
                                      & \text{kernel score},\\
 s_{j,T}^{\mathcal O}(u_j^\circ)|\theta_{j,T}^{L,\mathcal O}(u_j^\circ)|,
                                      & \text{linear contrast}.
 \end{cases}
 \label{eq:odd_oracle_signal}
\end{equation}
For positive oracle signals, let
$\eta_{j,T}^{U,\mathcal O}=\mu_{j,T}^{\mathcal O}/\sqrt{M_h}$ and
$\eta_{j,T}^{L,\mathcal O}=\mu_{j,T}^{\mathcal O}/(2M_\psi)$.
For the unbounded contrasts, use the single-change target geometry in
Appendix~\ref{sec:appendix-score-details} and set
\begin{equation*}
\Delta_{j,T}^{L,\mathcal O}=\mathbb E_{F_{j,-}}\psi-
                              \mathbb E_{F_{j,+}}\psi,\qquad
 \rho_{j,T}^{\mathcal O}
 =\frac{|\theta_{j,T}^{L,\mathcal O}(u_j^\circ)|}
              {|\Delta_{j,T}^{L,\mathcal O}|},
\end{equation*}
where the denominator is nonzero when $\mu_{j,T}^{\mathcal O}>0$.
Then $\rho_{j,T}^{\mathcal O}\in(0,1]$, and define
\begin{equation*}
\overline\eta_{j,T}^{SG,\mathcal O}
 =s_{j,T}^{\mathcal O}(u_j^\circ)\rho_{j,T}^{\mathcal O}
                              \wedge\mu_{j,T}^{\mathcal O}/\sigma_\psi,
 \qquad
 \eta_{j,T}^{L,p,\mathcal O}=\mu_{j,T}^{\mathcal O}/\nu_p.
\end{equation*}

The next proposition bounds the odd-side part of
Assumption~\ref{ass:odd_nondeflation} at its fixed fraction $\delta$.
To separate the odd-score requirement from the full event $\mathcal G_{j,T}$, write
\begin{equation*}
\begin{aligned}
 \mathcal G_{j,T}^{\mathcal O}
 &=\left\{\max_{b\in\mathcal A_j}V_{k_j^*,b}^{\mathcal O}<\delta V_{k_j^*}^{\mathcal O}\right\},\\
 q_{j,T}^{\mathcal O}
 &=\mathbb P\{k_j^*\in\widehat{\mathcal S},
                           (\mathcal G_{j,T}^{\mathcal O})^c\}.
 \end{aligned}
\end{equation*}
On a reported target,
$\mathcal G_{j,T}=\mathcal G_{j,T}^{\mathcal O}\cap
(\{|\mathcal A_j|\ge1\}\cup\{j\in\mathcal J_T\})$. Put
\begin{equation*}
s_{j,T}^{\rm acc}=\mathbb P\{k_j^*\in\widehat{\mathcal S},\
 V_{k_j^*}^{\mathcal O}>0,\ |\mathcal A_j|=0,\ j\notin\mathcal J_T\}.
\end{equation*}
On a reported target, $j\notin\mathcal J_T$ is precisely the failure
$V_{k_j^*}^{\mathcal E}\le\max_{1\le b\le B}V_{k_j^*,b}^{\mathcal E}$.
When $\mathcal A_j=\varnothing$, $\mathcal G_{j,T}^{\mathcal O}$ reduces to
$\{V_{k_j^*}^{\mathcal O}>0\}$. The events counted by
$s_{j,T}^{\rm acc}$ and $q_{j,T}^{\mathcal O}$ are therefore disjoint,
which gives the decomposition below.

\begin{proposition}\label{prop:odd_score_verification}
Fix $\delta\in(0,1)$ as in Assumption~\ref{ass:odd_nondeflation}.
For each positive-signal target, use
$t=\delta\mu_{j,T}^{\mathcal O}/2$ in the bounds below.
For the $U$, $L$, $SG$, and $L^p$ score families, assume the conditions of
Lemmas~\ref{lem:ustat_validation_separation},
\ref{lem:linear_validation_separation},
\ref{lem:subgaussian_validation_separation}, and
\ref{lem:finite_moment_linear_validation_separation}, respectively.
Then $q_{j,T}^{\mathcal O}\le1\wedge e_{j,T}$ and
\begin{equation}
 \mathbb P\{k_j^*\in\widehat{\mathcal S},\mathcal G_{j,T}^{c}\}=q_{j,T}^{\mathcal O}+s_{j,T}^{\rm acc}\le1\wedge\{e_{j,T}+s_{j,T}^{\rm acc}\},
 \label{eq:odd_scan_reduction_bound}
\end{equation}
where $e_{j,T}$ is the observed lower-tail probability at the oracle split
plus the sum of the splitwise permutation upper-tail probabilities:
\begin{equation*}
e_{j,T}=\mathbb P\{S_j^{\mathcal O}(u_j^\circ)\le\mu_{j,T}^{\mathcal O}/2\}
    +\sum_{b=1}^{B'}\sum_{u\in\mathcal T_j^{*,\mathcal O}}
                     \mathbb P\{S_{j,b}^{\mathcal O}(u)\ge t\},
 \qquad t=\delta\mu_{j,T}^{\mathcal O}/2.
\end{equation*}
It satisfies the corresponding bound below:
\begin{align}
 U:\quad e_{j,T}
 &\le\mathbbm1\{\eta_{j,T}^{U,\mathcal O}<c_0\}
    +Ce^{-c(\eta_{j,T}^{U,\mathcal O})^2}
    +CB'L_j^{\mathcal O}e^{-ct^2/M_h},
 \label{eq:odd_U_score_tails}\\
 L:\quad e_{j,T}
 &\le2e^{-(\mu_{j,T}^{\mathcal O})^2/(8M_\psi^2)}
                 +2B'L_j^{\mathcal O}e^{-t^2/(4M_\psi^2)},
 \label{eq:odd_L_score_tails}\\
 SG:\quad e_{j,T}
 &\le2e^{-(\mu_{j,T}^{\mathcal O})^2/(8\sigma_\psi^2)}
      +2B'L_j^{\mathcal O}
       \left\{e^{-t^2/(16|\Delta_{j,T}^{L,\mathcal O}|^2)}
                     +e^{-t^2/(8\sigma_\psi^2)}\right\},
 \label{eq:odd_SG_score_tails}\\
 L^p:\quad e_{j,T}
 &\le C_p(\nu_p/\mu_{j,T}^{\mathcal O})^p
       +B'L_j^{\mathcal O}
       \left\{2e^{-t^2/(16|\Delta_{j,T}^{L,\mathcal O}|^2)}
                           +C_p(\nu_p/t)^p\right\}.
 \label{eq:odd_Lp_score_tails}
\end{align}
A target with nonpositive signal or inadmissible side sizes is assigned
$e_{j,T}=1$. The averages, or sums, of
$1\wedge(e_{j,T}+s_{j,T}^{\rm acc})$ tending to zero imply, respectively,
$\mathbb E D_T=o(K^*)$ or $\mathbb E D_T=o(1)$, since \(\mathbb E D_T=\sum_{j=1}^{K^*}(q_{j,T}^{\mathcal O}+s_{j,T}^{\rm acc})\).
Under \textup{(A1)}, \eqref{eq:power_failure_count_identity} and
\begin{equation*}
\frac1{K^*}\sum_{j=1}^{K^*}\mathbb P(k_j^*\notin\widehat{\mathcal S})\le\mathbb P\!\left(\{k_1^*,\ldots,k_{K^*}^*\}\nsubseteq\widehat{\mathcal S}\right)=o(1)
\end{equation*}
also show that $\mathbb E D_T=o(K^*)$ is equivalent to $(K^*)^{-1}\sum_{j=1}^{K^*}q_{j,T}^{+}=o(1)$. The requirement $\mathbb P\{|\mathcal J_T|\ge dK^*\}\to1$ in \textup{(B2)} is verified separately below.
\end{proposition}

\begin{proof}[Proof of Proposition~\ref{prop:odd_score_verification}]
Fix an admissible positive-signal target and write
$\mu=\mu_{j,T}^{\mathcal O}$. The target, oracle split, grid, and
scoring function are fixed independently of the odd observations, as
specified above. We first prove the event reduction
\eqref{eq:odd_scan_reduction_bound}, then derive its four explicit tail
bounds, and finally account for the case of no accepted permutations.
On a reported target, suppose that
$S_j^{\mathcal O}(u_j^\circ)>\mu/2$ and every proposed scan score is
less than $\delta\mu/2$. The observed maximum is at least its oracle value.
Thus $V_{k_j^*}^{\mathcal O}>0$ and, for every accepted proposal, \(V_{k_j^*,b}^{\mathcal O}<\delta\mu/2<\delta V_{k_j^*}^{\mathcal O}\).
This event is contained in $\mathcal G_{j,T}^{\mathcal O}$, including
when $\mathcal A_j=\varnothing$. By the union bound,
\begin{align}
 q_{j,T}^{\mathcal O}
 &\le\mathbb P\{S_j^{\mathcal O}(u_j^\circ)\le\mu/2\}
    +\sum_{b=1}^{B'}\mathbb P\{V_{k_j^*,b}^{\mathcal O}\ge\delta\mu/2\}
       \nonumber\\
 &\le\mathbb P\{S_j^{\mathcal O}(u_j^\circ)\le\mu/2\}
    +\sum_{b=1}^{B'}\sum_{u\in\mathcal T_j^{*,\mathcal O}}
          \mathbb P\{S_{j,b}^{\mathcal O}(u)\ge\delta\mu/2\}
 =e_{j,T}.
 \label{eq:odd_oracle_scan_bound}
\end{align}
The second inequality uses the fixed-grid scan definition. Uniform bounds
on the last probabilities give the factor $B'L_j^{\mathcal O}$ below.
In particular, an accepted-score failure has been bounded by a
proposed-score failure, without treating accepted scores as unconditional
uniform-permutation scores.

For the kernel score, Lemma~\ref{lem:ustat_validation_separation}
controls the observed probability at $u_j^\circ$. At any grid split, its
conditional permutation bound gives
\(\mathbb P_\pi\{s_{j,T}^{\mathcal O}(u)\sqrt{|D_h^\pi|}\ge t
\mid Z_1,\ldots,Z_n\}\le Ce^{-ct^2/M_h}\), where $D_h^\pi$ is the
statistic evaluated at that split after permutation. This bound holds
for every pooled sample and is uniform over all admissible splits. Taking expectations
and substituting in \eqref{eq:odd_oracle_scan_bound} yields
\eqref{eq:odd_U_score_tails}.

For bounded linear contrasts, Lemma~\ref{lem:linear_validation_separation}
gives the observed probability at the oracle split as at most
$2\exp\{-\mu^2/(8M_\psi^2)\}$. At any grid split, its conditional
permutation bound is $2\exp\{-t^2/(4M_\psi^2)\}$ for every pooled
sample. Taking expectations preserves this bound. There are $B'L_j^{\mathcal O}$
permutation-and-split pairs, so \eqref{eq:odd_oracle_scan_bound} gives
\eqref{eq:odd_L_score_tails}.

For the sub-Gaussian family, the two-mean representation
\eqref{eq:two_mean_representation} holds across the entire target, with
the same mean jump $\Delta_{j,T}^{L,\mathcal O}$ at every grid split.
The coefficient vector changes with the split, but its squared norm is
always one. Lemma~\ref{lem:subgaussian_validation_separation} therefore
gives, uniformly over $u\in\mathcal T_j^{*,\mathcal O}$, \(\mathbb P\{S_{j,b}^{\mathcal O}(u)\ge t\} \le2e^{-t^2/(16|\Delta_{j,T}^{L,\mathcal O}|^2)} +2e^{-t^2/(8\sigma_\psi^2)}\).
Its observed-score bound at $u_j^\circ$ is
$2e^{-\mu^2/(8\sigma_\psi^2)}$. Adding the one observed term and
the $B'L_j^{\mathcal O}$ permutation terms yields
\eqref{eq:odd_SG_score_tails}. These probabilities average over both
the odd observations and the proposal.

For the finite-moment family, the same uniformity follows from
$\sum_i c_i^2=1$ and $\sum_i|c_i|^p\le1$ at every split.
Lemma~\ref{lem:finite_moment_linear_validation_separation} gives \(\mathbb P\{S_{j,b}^{\mathcal O}(u)\ge t\} \le2e^{-t^2/(16|\Delta_{j,T}^{L,\mathcal O}|^2)} +C_p(\nu_p/t)^p\),
and the observed term is at most $C_p(\nu_p/\mu)^p$.
Substitution in \eqref{eq:odd_oracle_scan_bound} proves
\eqref{eq:odd_Lp_score_tails}, retaining the separate exponential
label term and polynomial noise term.
The four family-specific bounds establish the assertion about
$q_{j,T}^{\mathcal O}$. We now account for targets with $\mathcal A_j=\varnothing$ and $j\notin\mathcal J_T$. On a reported target, a failure of
$\mathcal G_{j,T}$ is either a failure of
$\mathcal G_{j,T}^{\mathcal O}$, or occurs while
$\mathcal G_{j,T}^{\mathcal O}$ holds, $\mathcal A_j=\varnothing$, and
$j\notin\mathcal J_T$. On an empty accepted set,
$\mathcal G_{j,T}^{\mathcal O}$ holds exactly when
$V_{k_j^*}^{\mathcal O}>0$. The two cases are disjoint, which proves \(\mathbb P\{k_j^*\in\widehat{\mathcal S},\mathcal G_{j,T}^c\}=q_{j,T}^{\mathcal O}+s_{j,T}^{\rm acc}\).
Summing over $j$ gives the identity for $\mathbb E D_T$. Truncating
each target probability at one gives \eqref{eq:odd_scan_reduction_bound}.
Averaging or summing proves the remaining statements.
\end{proof}

Since $\delta\in(0,1)$ is fixed, the preceding expressions reduce to
\begin{equation}
\begin{aligned}
U:\quad e_{j,T}
 &\le\mathbbm1\{\eta_{j,T}^{U,\mathcal O}<c_0\}
    +C(1+B'L_j^{\mathcal O})e^{-c(\eta_{j,T}^{U,\mathcal O})^2},\\
L:\quad e_{j,T}
 &\le C(1+B'L_j^{\mathcal O})e^{-c(\eta_{j,T}^{L,\mathcal O})^2},
 \\
 SG:\quad e_{j,T}
 &\le C(1+B'L_j^{\mathcal O})e^{-c(\overline\eta_{j,T}^{SG,\mathcal O})^2},
 \\
 L^p:\quad e_{j,T}
 &\le CB'L_j^{\mathcal O}
       e^{-c\{s_{j,T}^{\mathcal O}(u_j^\circ)\rho_{j,T}^{\mathcal O}\}^2}
       +C_p(1+B'L_j^{\mathcal O})(\eta_{j,T}^{L,p,\mathcal O})^{-p}.
\end{aligned}
 \label{eq:odd_U_existing_order}
\end{equation}
Indeed, substitute $t=\delta\mu_{j,T}^{\mathcal O}/2$ and use
$\mu_{j,T}^{\mathcal O}/|\Delta_{j,T}^{L,\mathcal O}|=
 s_{j,T}^{\mathcal O}(u_j^\circ)\rho_{j,T}^{\mathcal O}$.
For the kernel score, $(\eta_{j,T}^{U,\mathcal O})^2=
s_{j,T}^{\mathcal O}(u_j^\circ)^2\theta_{j,T}^{U,\mathcal O}(u_j^\circ)/M_h$.
The fixed fraction $\delta$ changes only the constants, which in these
four displays may depend on $\delta$ but not on $j,T$.

Let $L_T^{\mathcal O}=\max_jL_j^{\mathcal O}$. Sufficient uniform
conditions for average odd-side score separation are, for the respective families,
\begin{equation}
 \begin{aligned}
 \frac{\min_j(\eta_{j,T}^{U,\mathcal O})^2}{\log(2+B'L_T^{\mathcal O})}
 &\longrightarrow\infty,\\
 \frac{\min_j(\eta_{j,T}^{L,\mathcal O})^2}{\log(2+B'L_T^{\mathcal O})}
 &\longrightarrow\infty,\\
 \frac{\min_j(\overline\eta_{j,T}^{SG,\mathcal O})^2}{\log(2+B'L_T^{\mathcal O})}
 &\longrightarrow\infty.
 \end{aligned}
 \label{eq:odd_uniform_average_signals}
\end{equation}
For the strengthened version,
replace the logarithm by $\log(2+K^*B'L_T^{\mathcal O})$ for each family.
For finite-moment contrasts, one sufficient strengthened condition is
\begin{equation}
 \begin{aligned}
 \frac{\min_j\{s_{j,T}^{\mathcal O}(u_j^\circ)\rho_{j,T}^{\mathcal O}\}^2}
      {\log(2+K^*B'L_T^{\mathcal O})}&\longrightarrow\infty,\\
 \frac{\min_j(\eta_{j,T}^{L,p,\mathcal O})^p}
      {K^*(1+B'L_T^{\mathcal O})}&\longrightarrow\infty.
 \end{aligned}
 \label{eq:odd_moment_simultaneous_signals}
\end{equation}
Remove $K^*$ inside the logarithm and in the denominator for an average
condition. The corresponding averages or sums of
$s_{j,T}^{\rm acc}$ must also vanish to control $\mathbb E D_T$.

To verify these assertions, let $a_T$ be the minimum squared signal
in the relevant numerator of \eqref{eq:odd_uniform_average_signals}. If
$a_T/\log(2+K^*B'L_T^{\mathcal O})\to\infty$, then
$K^*(1+B'L_T^{\mathcal O})e^{-c a_T}\to0$ for every fixed $c>0$.
For the kernel score, $\min_j(\eta_{j,T}^{U,\mathcal O})^2\to\infty$
also removes the indicator in \eqref{eq:odd_U_existing_order}. The
permutation term has the same squared-signal exponent as the observed term. For the finite-moment family, the second
condition in \eqref{eq:odd_moment_simultaneous_signals} controls the summed
polynomial error. Thus simultaneous control of that term has a
multiplicative, rather than merely logarithmic, cost in $K^*$.

For a random interval pool or an independently trained score, average
the conditional bounds over that randomness. Uniform rates on events
whose probabilities tend to one suffice for the averaged, truncated
bound. For the summed bound, include the sum of the exceptional-event
probabilities.

The same quantities can be expressed through jump size and localization.
If the two oracle side sizes are of order $\omega_j$, then \((\overline\eta_{j,T}^{SG,\mathcal O})^2 \asymp \omega_j(\rho_{j,T}^{\mathcal O})^2 \min\{1,|\Delta_{j,T}^{L,\mathcal O}|^2/\sigma_\psi^2\}\).
This follows from $s_{j,T}^{\mathcal O}(u_j^\circ)^2\asymp\omega_j$ and
$|\theta_{j,T}^{L,\mathcal O}(u_j^\circ)|=
 \rho_{j,T}^{\mathcal O}|\Delta_{j,T}^{L,\mathcal O}|$.
If $|u_j^\circ-\tau_j^*|=o(\omega_j)$ and both sides have at least a
constant multiple of $\omega_j$ observations, then
$\rho_{j,T}^{\mathcal O}=1-o(1)$. Minima over targets require uniform
versions of these geometric statements.
For the kernel score, writing $\zeta_-,\zeta_+$ for the two contamination
fractions gives \(\Theta_h(\widetilde F_{j,-}^{\mathcal O}(u_j^\circ), \widetilde F_{j,+}^{\mathcal O}(u_j^\circ)) =(1-\zeta_--\zeta_+)^2\Theta_h(F_{j,-},F_{j,+})\),
because the difference of the average laws is
$(1-\zeta_--\zeta_+)(F_{j,-}-F_{j,+})$ and $\Theta_h$ is quadratic
in this difference.



\subsection{The number of accepted permutations}
\label{app:power:acceptance}

The preceding subsection controls the odd-score part of the joint
failure probability. We now bound the term with $\mathcal A_j=\varnothing$
and $j\notin\mathcal J_T$, and verify
$\mathbb P\{|\mathcal J_T|\ge dK^*\}\to1$.

On a reported target, let
$r_{j,T}^{\rm acc}=|\widetilde\Pi_{k_j^*}|/|\Pi_{k_j^*}|$,
and set $r_{j,T}^{\rm acc}=0$ for unreported targets. Conditional on
$\mathcal F^{\mathcal O}$, the odd-side observations, acceptance rule, and detector
randomization are fixed. The $B'$ proposals are independent uniform
draws from $\Pi_{k_j^*}$, so each is accepted independently with
probability $r_{j,T}^{\rm acc}$. Thus, on
$\{k_j^*\in\widehat{\mathcal S}\}$, the conditional law of
$|\mathcal A_j|$ is $\operatorname{Bin}(B',r_{j,T}^{\rm acc})$, and
\begin{equation}
 \mathbb P\bigl(|\mathcal A_j|=0\mid\mathcal F^{\mathcal O}\bigr)
 =(1-r_{j,T}^{\rm acc})^{B'}.
 \label{eq:odd_zero_acceptance_probability}
\end{equation}

For the empty-acceptance term $s_{j,T}^{\rm acc}$, define \(e_{j,T}^{\mathcal E} =\mathbb P\left\{V_{k_j^*}^{\mathcal E} \le\max_{1\le b\le B}V_{k_j^*,b}^{\mathcal E} \,\middle|\,\mathcal F^{\mathcal O}\right\}\).
Let $b_{j,T}^{\mathcal E}$ be the conditional bound inside the outer
expectation in \eqref{eq:U_basic_signal_condition}--\eqref{eq:Lp_basic_signal_condition},
truncated at one, for a fixed split. For a scan rule, use the corresponding
scan bound in Appendix~\ref{sec:appendix-score-details}.
The corresponding conditional argument gives
$e_{j,T}^{\mathcal E}\le b_{j,T}^{\mathcal E}$ almost surely, with the
bound set to one when its signal or side-size conditions fail.
Conditional on $\mathcal F^{\mathcal O}$, the event of no accepted
proposal is independent of the even-side observations and permutation
draws. Since the reported-target and positive odd-score indicators are
also measurable under this conditioning,
\eqref{eq:odd_zero_acceptance_probability} gives
\begin{equation}
 s_{j,T}^{\rm acc}
 =\mathbb E\left[
 \mathbbm1\{k_j^*\in\widehat{\mathcal S},V_{k_j^*}^{\mathcal O}>0\}
 (1-r_{j,T}^{\rm acc})^{B'}e_{j,T}^{\mathcal E}\right]
 \le\mathbb E\left[
 \mathbbm1\{k_j^*\in\widehat{\mathcal S}\}
 (1-r_{j,T}^{\rm acc})^{B'}b_{j,T}^{\mathcal E}\right].
 \label{eq:odd_even_empty_acceptance_bound}
\end{equation}
Together with \eqref{eq:odd_scan_reduction_bound}, this proves
\begin{equation*}
\mathbb P\{k_j^*\in\widehat{\mathcal S},\mathcal G_{j,T}^{c}\}\le
 1\wedge\left\{e_{j,T}+
 \mathbb E\left[\mathbbm1\{k_j^*\in\widehat{\mathcal S}\}
       (1-r_{j,T}^{\rm acc})^{B'}b_{j,T}^{\mathcal E}\right]\right\}.
\end{equation*}
If the averages or sums of these bounds tend to zero, then
$\mathbb E D_T=o(K^*)$ or $\mathbb E D_T=o(1)$, respectively.
The expectation in \eqref{eq:odd_even_empty_acceptance_bound}
allows both factors to depend on the odd-side observations.

To verify $\mathbb P(|\mathcal J_T|\ge dK^*)\to1$, suppose there is a set of target labels
$\mathcal I_T\subseteq[K^*]$, fixed conditional on the interval pool,
and a fixed $\varepsilon\in(0,1-d)$ such that \(|\mathcal I_T|\ge(d+\varepsilon)K^*,1/{K^*}\mathbb E\left[\sum_{j\in\mathcal I_T} b_{j,T}^{\mathcal E}\right]\longrightarrow0\).
On $\mathcal M_T$, all these targets are reported. If
$|\mathcal J_T|<dK^*$, more than $\varepsilon K^*$ of the labels in
$\mathcal I_T$ must fail even-side separation. Markov's inequality gives
\begin{equation}
 \mathbb P(|\mathcal J_T|<dK^*)
 \le\mathbb P(\mathcal M_T^c)
       +\frac{\mathbb E[\sum_{j\in\mathcal I_T}b_{j,T}^{\mathcal E}]}
                    {\varepsilon K^*}
 \longrightarrow0.
 \label{eq:score_subset_support_bound}
\end{equation}
For a deterministic set, the expectation in this bound is simply
$\sum_{j\in\mathcal I_T}\mathbb E b_{j,T}^{\mathcal E}$.
Under \textup{(A1)}, $\mathbb E D_T=o(K^*)$ and \eqref{eq:score_subset_support_bound} imply the joint separation requirement in Assumption~\ref{ass:odd_nondeflation} by \eqref{eq:power_failure_count_identity}.

If the detector statistic and the odd-side score coincide and acceptance requires $V_{k_j^*,b}^{\mathcal O}>\lambda_T$, with
$\lambda_T\ge0$, then the event used in
Proposition~\ref{prop:odd_score_verification} together with $|\mathcal A_j|\ge1$ implies
$\lambda_T<\delta\mu_{j,T}^{\mathcal O}/2$.
For $|\mathcal A_j|\ge1$ and positive observed transformed score, \(W_{k_j^*}^{\mathcal O} \le\frac{|\mathcal A_j|+1} {1+|\mathcal A_j|g(\lambda_T)/g(V_{k_j^*}^{\mathcal O})}\).
The bound follows from the monotonicity of $g$ on $[0,\infty)$, which is implied by \textup{(B1)}.

\section{Numerical Details}

We give the data-generating mechanisms, scores, and evaluation conventions
used in the numerical studies. The nominal level is $\alpha=0.1$, with
$B=B'=499$ permutations and $g(v)=v^{12}$ for the \texttt{POIS+} weights.
Detector tuning parameters and numerical implementation details are
specified in the code.

\subsection{Distributional Simulation Implementation}
\label{app:midpoint_detector_details}

For Section~\ref{subsec:sim_midpoint_detectors}, the sequence starts in
state $F_A$ and alternates between $F_A$ and $F_B$ at 20 fixed, unequally
spaced change-points. The locations are specified in the code.
All four detectors use the same data in each repetition, with settings
held fixed across the four scenarios.

Write $d_O^2$ for the median pairwise squared distance in the odd sample.
For \texttt{POIS} and \texttt{POIS+}, we use the kernel
\[
k_O(x,y)=\frac13\sum_{h\in\{.25,.5,1\}}
\exp\!\left\{-\frac{(x-y)^2}{2h^2d_O^2}\right\}.
\]
For an interval with observations reindexed as $X_1,\ldots,X_m$, the score is
\[
v=\max_{5\le q\le m-5}
\sqrt{\frac{q(m-q)}{m}}\,
\widehat{\operatorname{MMD}}_{k_O}(X_{1:q},X_{q+1:m}),
\]
where $\widehat{\operatorname{MMD}}$ is the square root of the nonnegative
biased estimator of squared MMD. The bandwidths remain fixed, and each
permutation recomputes the scan maximum. Each odd-side proposal reruns the
detector and reconstructs its midpoint intervals. It is accepted when the
target interval endpoints and the number of reported intervals are unchanged.

For \texttt{MOPS}, let $\widehat\mu$ and $\widehat v$ be the odd-sample
mean and sample variance, and use the score features
\[
s(x)=\left(\frac{x-\widehat\mu}{\widehat v},\,
-\frac12+\frac{(x-\widehat\mu)^2}{2\widehat v}\right)^{\mathsf T}.
\]
Let $\widehat\Omega$ be the inverse of a regularized covariance estimate
formed from adjacent odd-side score differences. Write $n_j,n_{j+1}$
for the adjacent odd-side segment lengths and $\overline s_L^{\mathcal O}$,
$\overline s_R^{\mathcal O}$, $\overline s_L^{\mathcal E}$, and
$\overline s_R^{\mathcal E}$ for the mean score features on either side
of the detected split within its midpoint interval. The mirror statistic is
\[
W_j^{\mathrm{loc}}=\frac{n_jn_{j+1}}{n_j+n_{j+1}}
(\overline s_L^{\mathcal O}-\overline s_R^{\mathcal O})^{\mathsf T}
\widehat\Omega(\overline s_L^{\mathcal E}-\overline s_R^{\mathcal E}).
\]
We select intervals with $W_j^{\mathrm{loc}}\ge t$, where $t$ is the
smallest positive observed magnitude satisfying
$(1+\#\{j:W_j^{\mathrm{loc}}\le-t\})/(1\vee\#\{j:W_j^{\mathrm{loc}}\ge t\})\le\alpha$.
The selection is empty if no such threshold exists.
All methods are evaluated on the midpoint intervals from
Example~\ref{exam:midpoint_pool}.

\subsection{High-Dimensional Simulation Implementation}
\label{app:hd_implementation}

For Section~\ref{subsec:sim_hd_mean_changes}, the initial mean is zero.
At each change-point, ten coordinates are drawn uniformly without
replacement, independently across changes. Each selected coordinate
receives an independent jump of $2\sqrt{3/10}$ or $-2\sqrt{3/10}$ with
equal probability, and the jumps accumulate across segments. The noise
coordinates are mutually independent. Their standardized distributions
are $N(0,1)$, $t_4/\sqrt2$, $(\chi_3^2-3)/\sqrt6$, and
$(e^{Z/2}-e^{1/8})/\sqrt{(e^{1/4}-1)e^{1/4}}$, where $Z\sim N(0,1)$.

The seeded pool is constructed on the paired index set
$\{1,\ldots,T/2\}$. Each interval $[a,b]$ maps to $[2a-1,2b]$ on the
original sequence. Detector coordinates are divided by scaled marginal
median absolute deviations. Let $\widehat u_k$ be the sparse direction
estimated on an odd-side interval, converted to the original coordinate
scale and normalized to have unit Euclidean norm. At the detected split,
let $L_k,R_k$ be the even-side blocks, with sizes $n_L,n_R$. The
even-side score is
\[
v_k(X_{I_k^{\mathcal E}})
=\sqrt{\frac{n_Ln_R}{n_L+n_R}}
\left|\operatorname{tm}_{0.2}\{\widehat u_k^{\mathsf T}X_t:t\in L_k\}
-\operatorname{tm}_{0.2}\{\widehat u_k^{\mathsf T}X_t:t\in R_k\}\right|,
\]
where $\operatorname{tm}_{0.2}$ removes the lowest and highest
$\lfloor0.2n\rfloor$ values from a block of size $n$ before averaging.
The direction and split remain fixed in even-side permutations.
For odd-side weights, the score is the maximum Euclidean norm of the
CUSUM vector over admissible splits, computed from the scaled odd
observations. The coordinate scales remain fixed under local permutations.

\texttt{MOPS} uses the same odd-side detector intervals.
\texttt{ART} uses the full-sample detector and calibrates over the
complete seeded pool. Interval lengths count original observations,
including both endpoints. Average length is computed within each nonempty
repetition and then averaged over those repetitions.
Runtime includes detection and subsequent inference and is averaged
over all repetitions.

\subsection{Profile Simulation Settings}
\label{app:profile_generation}

For Section~\ref{subsec:sim_profile_pois}, define
$\ell(x)=10/\{1+\exp(7.5-x)\}$ for $x\in[0,15]$, and
$q(u)=\sqrt{180}\{(u-1/2)^2-1/12\}$ and $b(u)=\sqrt{12}(u-1/2)$
for $u\in[0,1]$. The functions $q$ and $b$ have zero mean and unit
squared integral. Table~\ref{tab:profile_generation} specifies the
mean and random-effect functions for state $s\in\{0,1\}$.

\begin{table}[htbp]
\centering
\begin{minipage}{\linewidth}
\setstretch{1}
\caption{Mean and random-effect functions for the four profile scenarios.}
\label{tab:profile_generation}
\small
\setlength{\tabcolsep}{5pt}
\renewcommand{\arraystretch}{1.3}
\begin{tabular}{@{}lll@{}}
\toprule
Scenario & $m_s(x)$ & $r(x)$ \\
\midrule
\texttt{Mean} & $(1-0.7s)\ell(x)$ & $0$ \\
\texttt{Anomaly} & $\ell(x)+13.5s\exp\{-(x-4.5)^2/2\}$ & $0$ \\
\texttt{Adv-PC} & $\ell(x)+6.5s\cos(2\pi x/15)$ & $5.5\sin(2\pi x/15)$ \\
\texttt{Curvature} & $\ell(x)+5.5s q(x/15)$ & $5b(x/15)$ \\
\bottomrule
\end{tabular}
\end{minipage}
\end{table}

With $\phi$ denoting the standard normal density, the pointwise error
variance is
\[
\sigma^2(x)=
\begin{cases}
0.5+\dfrac4{1.2}\phi\!\left(\dfrac{x-7.5}{1.2}\right),
 & \text{in \texttt{Mean}},\\[4pt]
0.5, & \text{in the other three scenarios}.
\end{cases}
\]
Each profile has 50 independent design points from $U(0,15)$.
The $Z_i$'s and pointwise errors are mutually independent and independent
of the design points and change-point locations. In \texttt{Adv-PC},
the cosine mean shift is orthogonal in $L_2[0,15]$ to the random sine component.
Draw ten distinct integers uniformly from
$\{1,\ldots,T-75(K^*+1)\}$, order them as $u_1<\cdots<u_{10}$, and set
$\tau_j^*=u_j+75j$. The state starts at zero and switches at each
change-point. Conditional on these locations, profiles are i.i.d.
within each segment, and permutations act on whole profiles.

Each profile is smoothed separately by local linear regression.
The fitted mean and residual-variance functions are evaluated at ten
equally spaced points in $[0,15]$, giving a 20-dimensional vector.
The mean bandwidth uses a direct plug-in selector and the variance
bandwidth a normal-reference rule. The summary coordinates are centered
and scaled using the odd sample. With half-window $h=6$ on each split,
let $\bar z_{L,q},\bar z_{R,q}$ denote coordinate means in the two
even-side blocks, and let $\widehat v_{LR,q}$ be the average of their
sample variances with a positive numerical floor specified in the code.
The even-side score is
\[
v_k=\sum_{q=1}^{20}
\frac{(\bar z_{L,q}-\bar z_{R,q})^2}{\widehat v_{LR,q}}.
\]
The odd-side scan and weighting score instead use coordinate variances
over the full odd sequence. The reporting radius is $2h=12$ profiles
on the original sequence, so an interior interval contains 25 profiles.
\texttt{SPD-OSE} uses odd-side screening and even-side mirror statistics.
\texttt{ART} uses all profile summaries and a fixed grid with half-window
12, with 499 permutations. All methods use the same reporting radius.

\subsection{Real Data Implementation Details}
\label{app:realdata_implementation}

For the S\&P 500 example in Section~\ref{sec:real1_sp500}, stocks with
incomplete or nonpositive adjusted prices are removed. Returns are
winsorized at their stock-specific 1\% and 99\% quantiles and centered
stock by stock. The standardized logarithm of their daily cross-sectional
mean square gives the scalar signal used by all procedures.
The odd-side NOT detector uses the maximum absolute CUSUM score over
a seeded interval pool, with at least five observations per side.
An odd interval $[a,b]$ maps to $[2a-1,\min\{2b,T\}]$.
\texttt{NSP} uses the full sequence at level 0.1 with 499 null simulations.
In Figure~\ref{fig:real1_sp500}, dots mark detector splits for the NOT
candidates, \texttt{POIS}, and \texttt{POIS+}, and interval midpoints
for \texttt{NSP}.

For HAPT in Section~\ref{sec:real_har}, odd/even splitting acts on
successive nonoverlapping 0.5-second feature blocks.
Within each raw block, channel $q$ is filtered as
$x_{t,q}-\phi_qx_{t-1,q}$, with coefficients estimated from the first
recordings of subjects 1--6. Each block is represented by 73 features
describing posture, movement, spectrum, and temporal variation.
Features are centered by their medians, divided by scaled median absolute
deviations, and transformed by $\operatorname{asinh}$. They are then
rescaled by feature group. Scaling uses odd blocks for split-based
procedures and all blocks for full-data comparisons.
For the resulting feature vectors $z_t$, the score on left and right
windows containing $n_L$ and $n_R$ vectors, respectively, is
\[
v=\sqrt{\frac{n_Ln_R}{n_L+n_R}}\,\|\bar z_L-\bar z_R\|_2.
\]
The odd-side scan has half-window $h=5$. Its interval $[a,b]$ maps to
$[2a-1,2b]$, giving ten-second interior intervals on the original
feature sequence. The full detector uses the same physical window grid.
\texttt{MOPS} forms mirror statistics within the odd-side candidates,
with coordinate selection based on odd-side contrasts.
The \texttt{ART} adaptation uses first-principal-component rank scores
and calibrates over the full fixed window grid at level 0.1 with
499 permutations.

For each recording and seed, let $N$ be the number of selected intervals,
$J$ the number of reference change ranges, and $M$ the maximum number of
one-to-one matches, with a match allowed when the interval and range
intersect. We compute
\[
\mathrm{Precision}=\frac MN,\qquad
\mathrm{Recall}_{\mathrm{match}}=\frac MJ,\qquad
F_1=\frac{2M}{N+J}.
\]
Coverage is the proportion of reference ranges intersected by at least
one selected interval. For empty selections, coverage, matching recall,
and $F_1$ are zero, and precision is omitted from its average.
Metrics are averaged across recordings within each seed and then across
three seeds. For \texttt{MOPS}, seven recordings have empty selections
under each seed.

\begin{figure}[h]
\centering
\includegraphics[width=\linewidth]{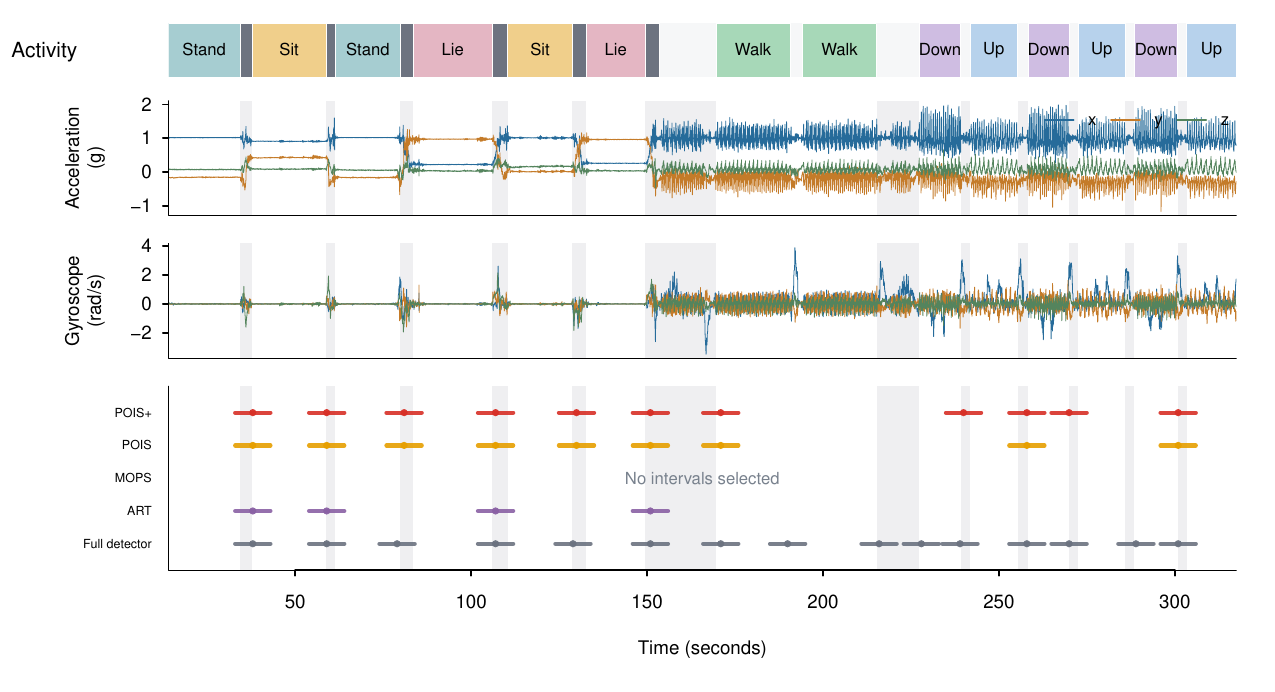}
\caption{HAPT subject 11: activity labels, raw signals, and reported intervals
under the same settings as the main HAPT figure, using the first permutation
run. Gray bands mark reference changes. Dots mark fixed test splits.}
\label{fig:har_additional}
\end{figure}

Figure~\ref{fig:har_additional} gives an additional illustration for
subject 11. \texttt{POIS+} selects 11 intervals and matches ten of the
12 reference changes, compared with eight matches for \texttt{POIS}
and four for \texttt{ART}. \texttt{MOPS} selects no intervals.

\putbib[references]
\end{bibunit}

\end{document}